\documentclass[11pt, twoside, openright]{report}

\usepackage[a4paper,width=150mm,top=25mm,bottom=25mm,bindingoffset=6mm]{geometry}
\usepackage{sectsty}
\allsectionsfont{\scshape}
\usepackage[titletoc]{appendix}
\usepackage[nottoc]{tocbibind}
\usepackage{chngcntr}

\usepackage{fancyhdr}
\fancypagestyle{plain}{
  \fancyhf{}
  \fancyfoot[C]{\footnotesize\thepage}}

\usepackage{amsmath, amsthm, amssymb}
\usepackage[T1]{fontenc}
\usepackage{newpxtext,newpxmath}
\usepackage{bm}
\usepackage{mathtools}
\usepackage[bottom, norule]{footmisc} 
\usepackage{graphicx}
\usepackage[mathscr]{euscript}
\usepackage{tensor}
\usepackage{hyperref}
\usepackage{cancel}
\usepackage{array}
\usepackage{tensor}
\usepackage[letterspace=135]{microtype}
\usepackage{emptypage}
\usepackage{pdfpages}

\usepackage{tikz-cd, tikz}
\usetikzlibrary{decorations.pathmorphing}
\usetikzlibrary{arrows}
\usetikzlibrary{patterns}

\usepackage[tableposition=top]{caption}
\newcolumntype{M}[1]{>{\centering\arraybackslash}m{#1}}

\hypersetup{colorlinks,
    linkcolor={red!50!black},
    citecolor={red!50!black},
    urlcolor={red!50!black}}

\theoremstyle{definition}
\newtheorem{defn}{Definition}[section]
\newtheorem{rmk}[defn]{Remark}
\newtheorem{example}[defn]{Example}

\theoremstyle{plain}
\newtheorem{thm}[defn]{Theorem}

\newtheorem{cor}[defn]{Corollary}
\newtheorem{lem}[defn]{Lemma}
\newtheorem{prop}[defn]{Proposition}
\newtheorem{claim}[defn]{Claim}

\numberwithin{equation}{section}

\usepackage{xpatch}
\xpatchcmd{\proof}{\itshape}{\normalfont\proofnameformat}{}{}
\newcommand{\proofnameformat}{}
\renewcommand{\proofnameformat}{\bfseries}

\newcommand{\diff}{\mrm{d}}
\newcommand{\Diff}{\mrm{D}}

\newcommand{\inv}{^{-1}}

\newcommand{\T}[1]{\text{#1}}
\newcommand{\mbb}[1]{\mathbb{#1}}
\newcommand{\mc}[1]{\mathcal{#1}}
\newcommand{\mf}[1]{\mathfrak{#1}}
\newcommand{\mrm}[1]{\mathrm{#1}}
\newcommand{\cc}[1]{\overline{#1}}

\newcommand{\R}{\mathbb{R}}
\newcommand{\C}{\mathbb{C}}
\newcommand{\Z}{\mathbb{Z}}

\newcommand{\K}{\mathbb{K}}

\newcommand{\GL}{\mrm{GL}}
\newcommand{\OR}{\mrm{O}}
\newcommand{\SO}{\mrm{SO}}
\newcommand{\U}{\mrm{U}}
\newcommand{\SU}{\mrm{SU}}
\newcommand{\SL}{\mrm{SL}}
\newcommand{\Spin}{\mrm{Spin}}
\newcommand{\Pin}{\mrm{Pin}}
\newcommand{\Sp}{\mrm{Sp}}
\newcommand{\Cl}{\mrm{Cl}}
\newcommand{\E}{\mrm{E}}

\newcommand{\so}{\mf{so}}
\newcommand{\su}{\mf{su}}
\newcommand{\spin}{\mf{spin}}

\newcommand{\End}{\mrm{End}}

\newcommand{\Hom}{\mrm{Hom}}
\newcommand{\Aut}{\mrm{Aut}}
\newcommand{\Ad}{\mrm{Ad}}
\newcommand{\ad}{\mrm{ad}}

\newcommand{\id}{\mrm{Id}}
\newcommand{\diag}{\mrm{diag}}

\newcommand{\thetaw}{\theta_\mrm{w}}

\DeclareMathOperator{\rank}{rank}

\DeclareMathOperator{\spann}{span}

\DeclareMathOperator{\imag}{Im}
\DeclareMathOperator{\tr}{tr}

\newcommand{\toenv}[4]{\begin{equation*}
			\begin{tikzcd}[row sep=0cm, ampersand replacement=\&]
				#1 \ar[r] \& {#2} \\
				#3 \ar[r,mapsto] \& {#4}
			\end{tikzcd}
		\end{equation*}}
		
\newcommand{\toenvnum}[5]{\begin{equation}
			\begin{tikzcd}[row sep=0cm, ampersand replacement=\&]
				#1 \ar[r] \& {#2} \\
				#3 \ar[r,mapsto] \& {#4}
			\end{tikzcd}
			\label{#5}
		\end{equation}}

\begin{document}

\pagenumbering{Roman}


\begin{titlepage}
	\begin{center}
        \vspace*{0.5cm}
        \Huge				
        \textbf{\textsc{A Mathematical Construction\\ of an\\ $\mathrm{E}_6$ Grand Unified Theory}}
				\vspace{1cm}
				
				\large
				A Thesis Presented for the Degree of\\
        \emph{Master of Science}
				\vspace{1.3cm}
								
				\includegraphics[width=0.4\textwidth]{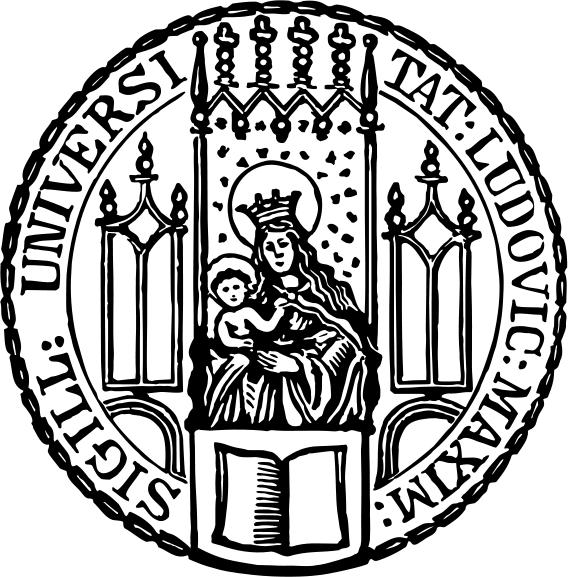}
				\vspace{1cm}
				
				\LARGE
        \textsc{Vivian Anthony Britto}
				\vspace{0.8cm}
				
				\large
				under the supervision of\\
				\smallskip
				\emph{Prof.~Dr.~Mark Hamilton}
				\vspace{0.7cm}
				
        Mathematisches Institut\\
        Ludwig-Maximilians-Universit\"at M\"unchen\\
				
				\vspace{0.4cm}
				September 2017
				
    \end{center}
\end{titlepage}

\clearpage
\thispagestyle{empty}
\mbox{}
\clearpage

\clearpage
\thispagestyle{empty}
\null\vfill
\begin{quotation}
\raggedleft
\emph{La gloria di colui che tutto move\\
per l'universo penetra, e risplende\\
in una parte pi\'u e meno altrove.}\\
\medskip
\textsc{Dante Alighieri, \emph{Paradiso}, Canto I}	
\end{quotation}
\vfill\vfill
\clearpage

\chapter*{Declaration}

I hereby declare that I am the sole author of this thesis. Where the work of others has been consulted, this is duly acknowledged in the references that appear at the end of this text. I have not used any other unnamed sources. All verbatim or referential use of the sources named in the references has been specifically indicated in the text. This work was not previously presented to another examination board and has not been published.

\vspace{2cm}
\begin{flushright}
	21st September 2017, M\"unchen
\end{flushright}

\chapter*{Acknowledgements}

I am indebted to Prof.~Hamilton for his guidance and patience during the writing of this thesis, and for making me a better mathematician and researcher in the process. Gratitude is also owed to Dr.~Robert Helling for many enlightening discussions, and for providing the inspiration for the final chapter of this paper.

I was lucky to have the finest colleagues at LMU, some of whom deserve a special mention: Alex Tabler and Danu Thung for their companionship, and the endless conversations about mathematics, physics, LaTeX, and everything in between; Martin Dupont for his camaraderie and careful proofreading of the final draft, and Agnes Rugel, who gently walked me through my many existential crises. Finally, I must thank my parents, to whom I owe everything, and my brother, for being a constant source of joy in my life.

\chapter*{Introduction}

Nature is simple. This article of faith, often taken for granted, sometimes fought over bitterly, has ever been at the centre of physicists' attempts to describe reality. A wonderful declaration of this is found in the third book of the \emph{Principia}, in a passage in which Newton sharpens the proverbial razor of that famous monk from Ockham: ``We are to admit no more causes of natural things than such as are both true and sufficient to explain their appearances. To this purpose the philosophers say that Nature does nothing in vain, and more is in vain when less will serve; for Nature is pleased with simplicity and affects not the pomp of superfluous causes.''

The Standard Model of particle physics is decidedly not simple. It is nevertheless an absolute triumph of modern science, remarkable for its economy in asserting but the single guiding principle of gauge symmetry within the framework of quantum field theory. Less pleasing is the actual $\U(1) \times \SU(2) \times \SU(3)$ gauge group itself, with the seemingly arbitrary charges of the $\U(1)$ hypercharge group. Other aspects also demand explanation: the Higgs mechanism and corresponding hierarchy problem, the observed non-zero masses of neutrinos, the 19 unrelated parameters within the Standard Model that need to be fine-tuned, and the inability to account for a potential cold dark matter particle; theoretical puzzles include the mathematical validity of the path integral, the chirality of the leptons and quarks, and the fact that these particles come in pairs. Perhaps most damning of all, the Standard Model cannot account for gravitation, because quantum field theories of gravity generally break down before reaching the Planck scale.

Grand unification attempts to answer some of these questions by positing that the symmetry of the Standard Model is a broken one, a shadow of some other, more fundamental symmetry of nature that is only accessible at extremely high energies. Mathematically, this corresponds to embedding the symmetry group of the Standard Model $G_\mrm{SM}$ into a larger, often simpler Lie Group, and picking a representation of the same such that it reduces to the Standard Model fermion representation when one restricts to $G_\mrm{SM}$. The first example came in 1974, when Georgi and Glashow proposed their $\SU(5)$ theory; though it was definitively disproved some twenty years later, it remains the prototypical grand unified theory for its aesthetic simplicity; its most important feature was certainly its logical explanation of the fractional charges of the quarks. This is a virtue that can be extended to the $\Spin (10)$ grand unified theory, another brainchild of Howard Georgi. The spinor representations in which it accommodates the Standard Model fermions are completely natural, separating the left- and right-handed particles in two different irreducible representations. We will study this theory in some detail in chapter \ref{spin10gutchapter}, since it plays a vital role in the construction of the $\E_6$ grand unified theory, the focus of this paper.

An $\E_6$ grand unified theory first appeared in 1976, due to G\"ursey, Ramond and Sikivie. Of the five exceptional groups, $\E_6$ is considered the most attractive for unification due to the following reasons: (i) it contains both $\Spin (10) \times \U(1)$ and $\SU (3) \times \SU(3) \times \SU(3)$ as maximal subgroups, each of which admit embeddings of the Standard Model; (ii) uniquely among the exceptional groups, it admits complex representations; in particular, its 27 dimensional fundamental representation accommodates one generation of left-handed fermions under the usual charge assignments; (iii) all of its representations are anomaly-free. We will discuss each of these aspects in the coming pages.

This thesis was originally conceived as an extension of Baez and Huerta's analysis in \emph{The Algebra of Grand Unified Theories} \cite{baez10} to the case of the $\E_6$ theory. There, the authors undertook a mathematical introduction to the representation theory of the Standard Model, and of grand unified theories---they treated in some detail the $\SU(5)$, $\Spin  (10)$ and Pati-Salam models, and then considered how these theories may be related to each other. Their pedagogical approach forms the basis for most of chapter \ref{standardmodelchapter}, and chapters \ref{thesu5gutsection} and \ref{thespin10gutsubsection}, where we, respectively, introduce the Standard Model, and prove that $\SU(5)$ and $\Spin (10)$ each extend it. Chapter \ref{e6isagutsection}, where we show that $\E_6$ is a grand unified theory, is the culmination of this project.

John Adams' \emph{Lectures on Exceptional Lie Groups} \cite{adams96} was the other major influence on this paper. His lucid presentation of the construction of these ``curiosities of Nature'' contained the three main ingredients we needed to prove that $\E_6$ extended the $\Spin (10)$ theory: (i) an explicit realisation of the subgroup $\Spin (10) \times \U(1) / \Z_4 \subset \E_6$; (ii) a route to the $27$-dimensional fundamental representations of $\E_6$, and (iii) the characterisation of the restriction of these representations to $\Spin (10) \times \U(1)$. Moreover, his introduction to Clifford Algebras, and the $\Spin$ groups and their representations, enabled us to deepen Baez and Huerta's discussion of the $\Spin (10)$ grand unified theory; in particular, we were able to make more precise the connection between the $\SU(5)$ fermion representation $\Lambda^* \C^5$, and the spinor representations $\Delta^\pm$ of $\Spin (10)$

Chapters \ref{e6isagutsection} and \ref{thenewfermionssubsec} contain our modest contributions to the literature. In the former, we explicitly check that that $\Z_4$ kernel of the homomorphism $\Spin (10) \times \U(1) \to \E_6$ acts trivially on every fermion; this is absolutely essential (in the cascade of unified theories that we consider) for $\E_6$ to extend the $\Spin (10)$ theory, and hence the Standard Model. We believe the reason that this result does not appear anywhere in the (predominantly physics) literature on the subject is the same reason that the $\Z_6$ kernel of the homomorphism $G_\mrm{SM} \to \SU(5)$ is rarely mentioned: physicists are often content to deal with these symmetry groups at the level of Lie algebras, which are indifferent to finite quotients of Lie groups. This affection for Lie algebras extends to their discussions of symmetry breaking in grand unified theories, which are almost universally analysed using Dynkin diagrams and related techniques. While computationally preferable, we felt that following this method would break with the spirit of the rest of the paper, so we attempted to understand symmetry breaking, in particular the symmetry breaking of the exotic $\E_6$ fermions under $\Spin (10) \to \SU(5)$, using a different approach: we explicitly embedded $\su(5) \hookrightarrow \so(10) \cong \spin (10)$, and then solved the related eigenvalue problem; this is the work of chapter \ref{thenewfermionssubsec}. The result of this calculation is table \ref{particlesinthenofe6}, where one sees how the Standard Model fermions and their new exotic compatriots fit into the fundamental representation of $\E_6$. This apparent bounty of new physics was the impetus for the final chapter, on the phenomenology of grand unified theories.

To avoid getting lost in quantum field theory, we restricted ourselves to the following question in chapter \ref{aspectsofphenochapter}: are there any predictions of grand unified theories that come solely out of representation theory? One of the most famous is certainly the Weinberg angle, and we treat this in section \ref{weakmixinganglesec}. We also consider in some detail, because it has a rather nice mathematical interpretation, the issue of anomaly cancellation; this is not a phenomenological prediction of grand unified theories per se, but rather, a requirement on their fermion representations: in section \ref{anomalycancellationsec}, we present Okubo's proof \cite{okubo77} that all representations of $\E_6$ are anomaly-free. We devote the final section of this paper to a brief but general discussion on the signatures of grand unified theories, and their present outlook.

\tableofcontents

\cleardoublepage\pagenumbering{arabic}

\chapter{The Standard Model}

\label{standardmodelchapter}

Any discussion of grand unification must begin with the Standard Model of particle physics. This rather uninspiringly-named theory is in fact a theory of almost everything: it describes three of the four known fundamental forces in the universe (the electromagnetic, weak, and strong interactions), as well as classifying all known elementary particles. It was developed in stages throughout the latter half of the twentieth century, with the current formulation being finalized in the mid-1970s upon experimental confirmation of the existence of quarks. The history of this development is a fascinating subject in its own right, featuring brilliant scientists, and set against the backdrop of some of the darkest periods of the last century. We refer the reader to \cite{crease96} for the history, and to \cite{hoddeson06} for a remarkable collection of scientific essays from the pioneers of the field.

Since those early days, experimental confirmation of Standard Model predictions have only added to its credence: highlights include the discover of the top quark in 1995, the tau neutrino in 2000, and the Higgs boson in 2012. Indeed, it can be said that the stunning experimental success of the Standard Model is often a cause for frustration among modern theoretical physicists, many of whom are holding out for evidence of new particles to lend support to the many projects of physics ``beyond the Standard Model''; a highly-readable overview of the major contenders can be found in \cite{lykken10}.

We have already encountered some of the shortcomings of the Standard Model: it does not fully explain baryon asymmetry, incorporate the full theory of gravitation as described by general relativity, or account for the accelerating expansion of the universe as possibly described by dark energy; the model does not contain any viable dark matter particle that possesses all of the required properties deduced from observational cosmology; it also does not incorporate neutrino oscillations and their non-zero masses. Understanding these difficulties is beyond the scope of this paper\footnote{Reference \cite{troitsky12} is a helpful starting point.}, but it is nevertheless clear that they are a strong motivation to look for other, hopefully more complete theories. In any case, the Standard Model lies at the heart of all model-building, of which grand unified theory is a part, so we absolutely must understand it before we move on.

\section{Preliminaries}

The basic theory of mathematical groups and their representations is really all that we will need to understand the algebra of the Standard Model, and of grand unified theories; in the first section below, we will briefly review the necessary concepts, and set the notation. We mention some references, noting that these represent but a sample of the literature: the book by Hall \cite{hall04} is a solid introduction to Lie groups, algebras and their representations; Fulton and Harris' text \cite{fulton99} on representation theory goes even deeper; for an approach more geared towards physicists, the book of Fuchs and Schweigert \cite{fuchs03} is an excellent resource.

In section \ref{symmetriessubsec}, we will formulate and motivate the two fundamental principles of the representation theory of particle physics. Though these rules are surprisingly easy to state and work with, their origins do require some preparation to appreciate, since they are best encountered within the framework of mathematical gauge theory. References for this field abound, we mention but three: for the mathematically inclined, there is the venerable text by Bleecker \cite{bleecker81}, and the lecture notes by Hamilton \cite{hamilton15}; the physicist can turn to Nakahara \cite{nakahara03} for his concise and clear presentation.

\subsection{Basic Definitions and Matrix Groups}
\label{basicdefinitionssubsection}
 
We follow \cite{adams96} in these paragraphs. A \emph{Lie group} $G$ is a group which is also a smooth manifold such that the maps $G \times G \to G$, $(g, h) \mapsto gh$ and $G \to G$, $g \mapsto g\inv$ are smooth. A \emph{homomorphism} $\theta : G \to H$ of Lie groups is a homomorphism of groups which is also a smooth map. A subgroup $H \subset G$ is said to be \emph{normal} if and only if $gH = Hg$ for all $g \in G$; a Lie group is called \emph{simple} if it possess no non-trivial connected normal subgroups.

A \emph{representation} $V$ of $G$, where $V$ is a finite-dimensional vector space over a field $\K = \R$ or $\C$, can be thought of as a map $G \times V \to V$ such that for all $g \in G$, $v \in V$, we have
\begin{itemize}
	\item $ev = v$ and $g (g' v) = (g g') v$, and
	\item $gv$ is a continuous function of $g$ and $v$, which is additionally $\K$-linear in $v$.
\end{itemize}
By choosing a basis for $V$, we get isomorphisms $V \cong \K^n$ for some integer $n$ and
\begin{equation*}
	\End (V) \cong M_n (\K) = \left\{n \times n \T{ matrices with entries in } \K \right\}\ .
\end{equation*}

A linear subspace $U \subset V$ is called \emph{$G$-invariant} if $gu \in U$ for all $u \in U$. A representation $V$ is said to be \emph{irreducible} if its only $G$-invariant subspaces are the trivial ones, $0$ and $V$; else, it is said to be \emph{reducible}. Irreducible representations will be fundamental in what follows; for brevity, we will call them irreps.

The general linear group of $V$,
\begin{equation*}
	\GL (V) = \Aut (V) = \{ A \in \End (V) \mid \exists\ A\inv \in \End (V) \}\ ,
\end{equation*}
is a group which is an open subset of $\End (V)$, and hence a smooth manifold. The product and inverse map in $\GL (V)$ are smooth, so $\GL (V)$ is a Lie group. If the dimension of $V$ over $\K$ is $n$, we will write $\GL(n, \K)$ for $\GL (V)$.

We can choose on $V$ a Hermitian form $\langle\ ,\ \rangle$ such that for $v, w \in V$, $\lambda \in \K$
\begin{itemize}
	\item $\langle v , w \lambda \rangle = \langle v , w \rangle \lambda$, and
	\item $\langle w , v \rangle = \cc{\langle v , w \rangle}$;
\end{itemize}
this form then defines taking the \emph{conjugate transpose} via $x^\dagger y := \langle x , y \rangle$. The subgroups $\{A \in \GL (n, \K) \mid A^\dagger A = \id \}$ are denoted $\OR(n), \U(n), \Sp (n)$ respectively for $\K = \R , \C, \mbb{H}$. For $\K = \R$ or $\C$, we also get the \emph{special linear} subgroups $\{A \in \GL(n, \K) \mid \det A = 1\}$ denoted $\SL (n, \R)$ and $\SL (n, \C)$. Finally, we have $\SO (n) = \SL (n, \R) \cap \OR(n)$ and $\SU(n) = \SL (n, \C) \cap \U(n)$. All these groups are collectively called \emph{classical groups}, and are in fact Lie groups. The groups $\OR (n), \SO (n), \U (n), \SU(n), \Sp (n)$ are all compact, while $\GL (n, \K)$and $\SL (n, \K)$ ($n \neq 1$) are not.

Since Lie groups possess the structure of a manifold, it is sensible to talk about tangent vectors; if the Lie group also happens to be  modelled on some vector space $V$, as the classical groups are, then the tangent spaces at each point, $T_p G$ will in fact be isomorphic to $V$. More will be said once we make the following

\begin{defn}[Lie Algebra]
	A Lie algebra is a vector space $V$ over a field $\K$ equipped with an operation $[\ ,\ ] : V \times V \to V$ called a \emph{Lie bracket}, which satisfies the following axioms.
	\begin{itemize}
		\item Bilinearity: $[ax + by, z] = a[x, z] + b[y, z]$ for all scalars $a, b \in \K$ and vectors $x, y, z \in V$.
		\item Alternativity: $[x, x] = 0$ for all $x \in V$.
		\item The Jacobi identity: $[x, [y, z]] + [z, [x, y]] + [y, [z, x]] = 0$ for all $x, y, z \in V$.
	\end{itemize}
	We will follow the standard convention of denoting the Lie algebra of a group by the same letters in lower-case Fraktur font, e.g.~$T_e G =: \mf{g}$.
\end{defn}

\begin{prop}
	For $G = \OR (n), \U(n)$ or $\Sp(n)$, a matrix $A \in \mf{g}$ if and only if $A^\dagger + A = 0$. For the groups $G = \SL (n, \R), \SL (n, \C)$, $A \in \mf{g}$ if and only if $\tr A$, the trace of $A$ vanishes.
\end{prop}

\begin{thm}\leavevmode \samepage
\label{bigthmaboutliealgebras}
	\begin{enumerate}
		\item For any Lie group $G$, the space $T_e G$ of tangent vectors to the identity element is a Lie algebra over $\K$.
		\item For any matrix group $G \subset \GL (V)$, the Lie bracket is given by the commutator, $[X, Y] = XY - YX$.
		\item If $\theta : G \to H$ is a homomorphism of Lie groups, the induced map $\diff \theta : \mf{g} \to \mf{h}$ is a homomorphism of Lie algebras.
		\item For any representation $V$ of $G$, $V$ is a representation of $\mf{g}$.
		\item For a matrix group $G$ acting on $V$ as the endomorphism group, $\mf{g}$ acts in the same way.
		\item If $H$ acts on $V$ and we are given a homomorphism $\theta : G \to H$ of Lie groups so that $G$ acts on $V$, the resulting action of $\mf{g}$ on $V$ is $X \cdot v = (\diff \theta (X))\cdot v$, where $X \in \mf{g}$.
	\end{enumerate}
\end{thm}

The proof of these statements can be found in any of the references listed at the beginning of this section. We note that in item (iv) above, a representation of a Lie algebra is defined in the obvious way: a map from $\mf{g} \times V \to V$ that is linear for $v \in V$ and respects the Lie bracket, $[X, Y] v = X (Yv) - Y (Xv)$.

The final definition in this section is of significant importance to us, and it goes as follows. $G$ acts on itself by conjugation, $c_g : G \to G$, $h \mapsto g h g\inv$, and this is clearly a homomorphism. We hence obtain by item (iii) in the theorem above, for each $g \in G$, a corresponding Lie algebra automorphism $\diff c_g : \mf{g} \to \mf{g}$. This is the \emph{adjoint representation} of $G$ on its Lie algebra, $\Ad : G \to \Aut \ \mf{g}$. The differential of this representation gives the \emph{adjoint representation of the Lie algebra} on itself; this map is $\ad: \mf{g} \to \GL (\mf{g})$, $\ad_X (Y) = [X, Y]$.

\subsection{Symmetries}
\label{symmetriessubsec}

In the mathematics of particle physics, all vector spaces must be complex because of foundational axioms in quantum mechanics.\footnote{See the Dirac-von Neumann axioms \cite{dirac30, neumann32}.} With this in mind, the two fundamental principles of the representation theory of particle physics can be stated in a few words: given a gauge symmetry $G$ of a theory, the fermions (matter fields) of this theory are basis vectors of unitary irreps of $G$, while the gauge bosons, which mediate forces, are basis vectors of its complexified adjoint representation (which is irreducible if $G$ is simple). In this section, we will try to motivate these claims, assuming some prior knowledge of gauge theory, the arena in which classical field theory plays out.

For a field theory on a spacetime, the \emph{Lagrangian} is the name given to a function that describes the dynamics and all the interactions of the fields. There are a few core principles that one must follow when attempting to write down such a function, of which we will only consider \emph{symmetry}: the Lagrangian, and hence the laws of physics, should be invariant under the transformations of some specified symmetry group $G$ (such as the Poincar\'e group, which encodes the symmetry of Minkowski spacetime). It is clear at the outset that the fields (particles) out of which the Lagrangian is built must live in irreps of $G$, to keep the invariance under $G$ manifest. The unitarity requirement on the irreps then arises quite naturally from the desire to compute observables (matrix elements). For example, a state $| \psi \rangle \mapsto P | \psi \rangle$ under a Poincar\'e transformation $P$, so an observable would transform as
\begin{equation*}
\begin{tikzcd}[row sep = 0cm]
	\langle \psi_1 | \psi_2 \rangle = M \ar[r, mapsto, , "P"] & M = \langle \psi_1 | P^\dagger P | \psi_2 \rangle\ .\\
\end{tikzcd}
\end{equation*}
Therefore, we need $P^\dagger P = \id$, i.e.~$P$ must be unitary, if the matrix element is to be invariant.\footnote{We used the Poincar\'e group here since it is a natural example, but we will not be concerned with it in what follows. The reason for this is that the unitary irreps of this group are all infinite dimensional, as proved by Wigner in 1939 \cite{wigner39}, and we wish to restrict to finite-dimensional representation theory. See \cite[pp.~109-103]{schwartz14} for further discussion.}

It is more involved to see why the gauge bosons live the adjoint representation; we begin with a plausibility argument from physics known as \emph{minimal coupling}. Consider a matter field $\psi (x)$: for local gauge transformations, $\psi(x) \mapsto g(x) \psi (x)$, and ordinary derivatives transform as
\begin{equation*}
	\partial_\mu \psi (x) \mapsto g(x) \partial_\mu \psi(x) + (\partial_\mu g(x)) \psi (x)\ ,
\end{equation*}
that is, inhomogeneously. What we would like instead is a gauge-covariant derivative $\Diff_\mu$, which transforms as $\Diff_\mu \psi (x) \mapsto g(x) \Diff_\mu \psi (x)$. To achieve this, we define
\begin{equation*}
	\Diff_\mu \psi (x) = \partial_\mu \psi (x) - A_\mu \psi (x)\ ,
\end{equation*}
where $A_\mu$ is a Lie algebra-valued 1-form written in a local basis. As the notation suggests, this is our gauge boson, and it is forced to have the transformation law
\begin{equation}
\label{transformationlawforgaugebosoneqn}
	A_\mu \mapsto A'_\mu (x) = g (x) A_\mu g\inv (x) + (\partial_\mu g (x)) g\inv (x)\ .
\end{equation}
The first term is the desired adjoint transformation of $A_\mu$. The second term vanishes if $g (x)$ is taken to be constant locally; this is referred to as a \emph{rigid gauge transformation}. The group of rigid physical transformations form a group isomorphic to $G$.

By way of motivation, the above should suffice (and is often the last word in physics textbooks). But let us go deeper. In mathematics, matter fields like $\psi (x)$ are sections of the vector bundle $P \times_\rho V \to M$ associated to a $G$-principal bundle $P \to M$ by a representation $\rho : G \to \GL (V)$. The derivative $\Diff_\mu$ above is usually denoted $\nabla_\mu$ in mathematics, and is in fact the covariant derivative on the vector bundle associated to a connection 1-form $A$ on the principal bundle; in coordinate-free notation, for a vector field $X \in \mf{X} (M)$, it outputs a section
\begin{equation*}
	\nabla^A_X \psi = \diff \psi (X) + \diff \rho (A_s (X)) \psi\ ,
\end{equation*}
where $A_s = A \circ \Diff s \in \Omega^1 (U, \mf{g})$ is the local gauge field for $s : M \supset U \to P$ a choice of local gauge, and $\diff \rho : \mf{g} \to \End (V)$ is the induced Lie algebra representation. A clue that something more needs to be said about the somewhat ad hoc imposition of the transformation law (\ref{transformationlawforgaugebosoneqn}) for $A$ can be found in the fact that in the parlance of gauge theory, the corresponding result is purely a statement about what happens when we change gauge from $s_i : U_i \to P$ to $s_j : U_j \to P$ on the principal bundle\footnote{See \cite[Theorem 5.25]{hamilton16}.}, and has nothing whatsoever to do with the associated vector bundle, where all the physics takes place. Instead, what is needed is the following: one can show that the difference between two connections is in fact a section (field) on the base manifold $M$ with values in the vector bundle $\Ad (P) := P \times_{\Ad} \mf{g}$, associated to $P \to M$ by the adjoint representation of $G$. (Heuristically, one can see this from equation (\ref{transformationlawforgaugebosoneqn}) by noting that the difference of two transformed gauge fields kills the $(\partial_\mu g (x)) g\inv (x)$ term.) Naturally, sections of this bundle then transform in the adjoint representation, as desired. A rather beautiful physical interpretation of this result is found in \cite{hamilton16}: in quantum field theory, particles in general are described as excitations of a given vacuum field; in the case of a gauge field, one has to declare the vacuum field to be a certain specific connection 1-form $A^0$ on the principal bundle, with reference to which all other gauge fields would then by described\footnote{N.b.~the form $A^0 \equiv 0$ is \emph{not} a connection.}; by the result stated in the previous paragraph, this difference (excitation) $A - A^0$ can then be identified with a 1-form on the spacetime $M$, with values in $\Ad (P)$ that hence transforms in the adjoint representation of $G$.

Now that we have the transformation rules for all the particles in our theory, a natural question arises: how can the gauge bosons be said to ``mediate forces''? The mathematical mechanism is in fact quite straightforward. When we say that a force is invariant under the action of some group, this corresponds to the statement that any physical process caused by this force should be described by an ``intertwining operator'', which is a linear operator that respects the action of the group under consideration. More precisely, suppose that $V$ and $W$ are finite-dimensional Hilbert spaces on which some group $G$ acts as unitary operators. Then a linear operator $F : V \to W$ is an \emph{intertwining operator} if $F ( g \psi) = g F (\psi)$ for every $\psi \in V$ and $g \in G$.  Now we saw in theorem \ref{bigthmaboutliealgebras} that a representation $\rho : G \times V \to V$  of a group $G$ gives rise to a representation of its Lie algebra $\mf{g}$ on $V$; we think of this as the linear map $\diff \rho : \mf{g} \otimes V \to V$.  It is easy to check that this map is an intertwining operator, and it hence gives the gauge bosons agency to act on particles.

\section{The Fundamental Forces}

We begin our brief exposition of the Standard Model proper with the representation theory of quantum chromodynamics (QCD), since it is the most straightforward application of the principles that we encountered in the previous section. Many great minds contributed to the development of this theory\footnote{Reference \cite[Ch.~4]{crease96} has an account of the history.}, but it was the trio of Fritzsch, Gell-Mann and Leutwyler who formulated the concept of colour as the source of a ``strong field'' in a Yang-Mills theory in 1973 \cite{fritzsch73}.

This will be followed by a description of the weak force, which will then be expanded to include electromagnetism in the section on the electroweak interaction; this milestone in the history of unification was due to independent work by Glashow \cite{glashow70}, Salam \cite{salam68} and Weinberg \cite{weinberg67}, for which they were jointly awarded the Nobel prize in 1979. We will follow the article of Baez and Huerta \cite{baez10} in this section, and in the one following, on the fermion representation of the Standard Model.

\subsection{The Strong Interaction}

Let us begin with the nucleons of high-school chemistry, the protons and neutrons. It turns out that they are not fundamental particles, but are instead made up of other particles called \emph{quarks}, which come in a number of different \emph{flavours}. It takes two flavours to make protons and neutrons, the \emph{up} quark $u$, and the \emph{down} quark $d$: the proton can be written as $p = uud$, and the neutron, $n = udd$ (the notation will be clarified momentarily). It follows from the charges of the proton (+1) and neutron (0) that $u$ has a charge $2/3$, and $d$, $- 1/3$.

Quark \emph{confinement} is one of two defining characteristics of QCD, the other being \emph{asymptotic freedom}. The latter is unfortunately outside the scope of this paper; we refer the reader to a review article by Gross, one of the discoverers of aymptotic freedom \cite{gross99}. Now quark confinement, loosely speaking, is the statement that the force between quarks does not diminish as they are separated; thus, they are forever bound into \emph{hadrons} such as the proton and the neutron. Let us try to understand this mathematically. Each flavour of quark comes in three different states called \emph{colours}: \emph{red} ($r$), \emph{green} ($g$), and \emph{blue} ($b$). This means that the Hilbert space for a single quark is $\C^3$, with $r$, $g$, and $b$ the standard basis vectors. The \emph{colour symmetry group} $\SU (3)$ acts on $\C^3$ in the obvious way, via its fundamental representation. Since both up and down quarks come in three colour states, there are really six kinds of quarks in matter: three up quarks, spanning a copy of $\C^3$;	$u^r, u^b, u^g \in \C^3$, and similarly for down quarks. The group $\SU(3)$ acts on each space. All six quarks taken together span the vector space $\C^3 \oplus \C^3 \cong \C^2 \otimes \C^3$, where $\C^2$ is spanned by $u$ and $d$. Confinement amounts to the following decree: all observed states must be white, i.e.~invariant under the action of $\SU(3)$. Hence, we can never see an individual quark, nor particles made from two quarks, because there are no vectors in $\C^3$ or $\C^3 \otimes \C^3$ which transform trivially under $\SU(3)$. But we do see see particles made up of three quarks, such as nucleons, because there are unit vectors in $\C^3 \otimes \C^3 \otimes \C^3$ fixed by $\SU(3)$. Indeed, as a representation of $\SU (3)$, $\C^3 \otimes \C^3 \otimes \C^3$ contains precisely one copy of the trivial representation: the antisymmetric rank-three tensors, $\Lambda^3 \C^3$. This one-dimensional vector space is spanned by the wedge product of all three basis vectors, $r \wedge b \wedge g \in \Lambda^3 \C^3$, so up to normalisation, this must be colour state of a nucleon. We also now see that the colour terminology is well-chosen, since an equal mixture of red, green, and blue light is white. Hence, confinement is intimately related to colour. An explanation of the quark flavours is postponed until the next section.

We will have much more to say about the quarks, but as an introduction, what we have above suffices: the strong force is concerned with the quarks; the up and the down quarks together span the representation $\C^2 \otimes \C^3$ of $\SU(3)$, where $\C^2$ is trivial under $\SU(3)$. In the previous section, we took the trouble to understand how gauge bosons transform and act, and we now we reap the fruits of that labour: from the standpoint of representation theory, all there is to say is that strong force is mediated by the \emph{gluons}, usually denoted by $g$, which live in $\C \otimes \su(3) = \mf{sl} (3, \C)$, the complexified adjoint representation of $\SU(3)$. They act on quarks via the standard action of $\mf{sl} (3, \C)$ on $\C^3$.

\subsection{The Weak Interaction}
 
Our story of the weak force begins, interestingly enough, in early attempts to describe the strong force, particularly in the work of Heisenberg in 1932 \cite{heisenberg32}. He hypothesised that the proton and nucleon were the two possible observed states of a nucleon; a nucleon would hence live in the simplest Hilbert space possible for such a setup: $\C^2 = \C \oplus \C$. Shortly thereafer, in 1936, Cassen and Condon \cite{cassen36} suggested that the $\C^2$ space of nucleons is acted upon by $\SU(2)$, emphasising the analogy with the spin of an electron, which is also described by vectors in $\C^2$ acted upon by $\SU(2)$. The property that distinguishes the proton from the neutron was hence dubbed \emph{isospin}: the proton was declared to be isospin up, $I_3 = + 1/2$, and the neutron isospin down, $I_3 = - 1/2$. The charge and the isospin of the nucleons $N$ were seen to be related in the following simple way:
\begin{equation*}
	Q(N) = I_3 (N) + \frac{1}{2}\ .
\end{equation*}
This turned out to be a special instance of what came to be called the \emph{Gell-Mann--Nishijima formula} (abbreviated as the NNG formula):
\begin{equation*}
	Q = I_3 + \frac{Y}{2}\ .
\end{equation*}
$Y$ is a quantity called \emph{hypercharge} which depends only on the ``family'' of the particle. For the moment, this simply means that $Y$ is required to be constant on representations of the isospin symmetry group, $\SU(2)$. To now understand how all of this relates to the modern theory of weak interactions, we have to introduce a new particle.

Along with the electron $e^-$ and the up and down quarks, the \emph{neutrinos} $\nu$ form the first generation of fundamental fermions. They carry no charge and no colour, and interact only through the weak force, first proposed by Enrico Fermi in 1933 \cite{fermi34}. The weak force is \emph{chiral}, i.e.~it cares about the handedness of particles: every particle thus far discussed comes in left- and right-handed varieties, which we will denote by subscript-$L$ and -$R$ respectively. Remarkably, the weak force interacts only with the left-handed particles, and right-handed \emph{antiparticles}. We have been silent about antiparticles until now, but they are quite simple to introduce: to each particle, there is a corresponding antiparticle, which is just like the original particle, but with opposite charge and isospin; mathematically, this just means that we pass to the dual representation. Returning to the weak interaction, when the neutron decays for example, we always have
\begin{equation*}
	n_L \to p_L + e^-_L + \cc{\nu}_R\ ,
\end{equation*}
and never
\begin{equation*}
	n_R \to p_R + e^-_R + \cc{\nu}_L\ .
\end{equation*}
This parity violation of the weak force, proposed by Yang and Lee in 1956 \cite{lee56} is still startling; no other physics, classical or quantum, looks different when viewed in a mirror. One important corollary of this oddity is that the right-handed neutrino $\nu_R$ has never been observed directly; we will discuss this particle in the context of grand unified theories in sections \ref{thespin10gutsubsection} and \ref{furtherreadingsec}.

The isospin mentioned above is an extremely useful quantity since it is conserved during quantum interactions; as such, we would like to extend it to weak interactions. First, for the proton and neutron to have the right isospins of $\pm \frac{1}{2}$, we must have the isospin of the up and down quarks defined to be $ \pm \frac{1}{2}$ respectively (making these particles the up and down states at which their names hint). A quick check then shows that isospin is not automatically conserved in weak interactions; for example, in the above neutron decay,
\begin{equation*}
	u_L d_L d_L \to u_L u_L d_L + e^-_L + \cc{\nu}_R\ ,
\end{equation*}
the right-hand side has $I_3 = - 1/2$ while the left-hand side has $I_3 = 1/2$. What is needed is an extension of the concept of isospin to the \emph{leptons}, i.e.~the particles which do not feel the strong force, $e^-$ and $\nu$; simply setting $I_3 (\nu_L) = \frac{1}{2}$ and $I_3 (e^-_L) = - \frac{1}{2}$ does the trick. This extension of isospin is called \emph{weak isospin}, and unlike the isospin of the nucleons, is an exact symmetry. We will simply refer to it as isospin from now on.

We come to the description of the weak force. This is a theory with the isospin symmetry group $\SU(2)$; the particles in the same representation are paired up in \emph{doublets},
\begin{equation*}
	\begin{pmatrix}
		\nu_L\\
		e^-_L
	\end{pmatrix}\ ,\quad \begin{pmatrix}
		u_L\\
		d_L
	\end{pmatrix}\ ,
\end{equation*}
with the particle with the higher $I_3$ on the top; this is just a shorthand way of writing that these particles live in (and span) the same irrep $\C^2$ of $\SU(2)$. The fact that only the left-handed particles are combined into doublets reflects the fact that only they participate in weak interactions. Every right-handed fermion, on the other hand, is trivial under $\SU(2)$: they are called \emph{singlets}, and span the trivial representation $\C$.

The particles in the doublets interact via the exchange of the so-called $W$ bosons,
\begin{equation*}
	W^+ = \begin{pmatrix}
		0 & 1\\
		0 & 0
	\end{pmatrix}\ , \quad \ W^0 = \begin{pmatrix}
		1 & 0\\
		0 & -1
	\end{pmatrix}\ , \quad\ W^- = \begin{pmatrix}
		0 & 0\\
		1 & 0
	\end{pmatrix}\ .
\end{equation*}
As we would expect, these span the complexified adjoint representation of $\SU(2)$, $\mf{sl}(2, \C)$, and they act on each of the particles in the doublets via the action of $\mf{sl} (2, \C)$.

We close this section with the afore-promised explanation of quark flavour splitting. Recall that colour is related to confinement; in much the same way, flavour is related to isospin. Indeed, we can use quarks to explain the isospin symmetry of the nucleons: protons and neutrons are so similar, with nearly the same mass and strong interactions, because $u$ and $d$ quarks are so similar, with nearly the same mass, and truly identical colours. As mentioned above, the isospin of the proton and neutron arises from the isospin of the quarks, once we define $I_3 (u) = 1/2$, and $I_3 (d) = - 1/2$; we see that the proton obtains the right $I_3$:
\begin{equation*}
	I_3 (p) = \frac{1}{2} + \frac{1}{2} - \frac{1}{2} = \frac{1}{2}\ ,
\end{equation*}
and a quick check shows the same for the neutron. This is a good start, but what we really need to do is to confirm that $p$ and $n$ span a copy of the fundamental representation $\C^2$ of $\SU(2)$. It turns out that the states $u \otimes u \otimes d$ and $u \otimes d \otimes d$ do \emph{not} span a copy of the fundamental representation of $\SU(2)$ inside $\C^2 \otimes \C^2 \otimes \C^2$; what is needed, for the proton for instance, is some linear combination of the $I_3 = 1/2$ flavour states which are made of two $u$'s and one $d$:
\begin{equation*}
	u \otimes u \otimes d\ , \quad u \otimes d \otimes u\ , \quad d \otimes u \otimes u \quad \in \C^2 \otimes \C^2 \otimes \C^2\ .
\end{equation*}
The exact linear combination required to make $p$ and $n$ work also involves the spin of the quarks, which is outside the scope of our discussion. What we can do however, is see that this is at least possible, i.e.~that $\C^2 \otimes \C^2 \otimes \C^2$ really does contain a copy of the fundamental representation $\C^2$ of $\SU(2)$. First note that any rank-2 tensor can be decomposed into symmetric and antisymmetric parts, $\C^2 \otimes \C^2 \cong \mrm{Sym}^2 \C^2 \oplus \Lambda^2 \C^2$. Now $\mrm{Sym}^2 \C^2$ is the unique 3-dimensional irrep of $\SU(2)$, and $\Lambda^2 \C^2$, as the top exterior power of its fundamental representation $\C^2$, is the trivial 1-dimensional irrep; as a representation of $\SU(2)$, we therefore have
\begin{equation*}
	\C^2 \otimes \C^2 \otimes \C^2 \cong \C^2 \otimes ( \mrm{Sym}^2 \C^2 \oplus \C) \cong (\C^2 \otimes \mrm{Sym}^2 \C^2) \oplus \C^2\ . 
\end{equation*}
So indeed, $\C^2$ is a subrepresentation of $\C^2 \otimes \C^2 \otimes \C^2$. As a final remark, we note that the NNG formula still works for quarks, provided we define the hypercharge for both quarks to be $Y = 1/3$.
 
\subsection{The Electroweak Interaction}
\label{theewinteractionsec}

All the fermions have now been grouped into $\SU(2)$ representations based on their isospin. Let us now consider the other piece of NNG formula, hypercharge. Just as we did for isospin, we can extend the notion of hypercharge to encompass the leptons, calling this new quantity \emph{weak hypercharge}. It is a matter of simple arithmetic to see that we must have $Y = -1$ for both left-handed leptons; for right handed leptons, since $I_3 = 0$, we must set $Y = 2Q$.

How can we understand hypercharge? Let us frame the discussion by briefly discussing isospin again: it is an observable, and we know from quantum mechanics that it hence corresponds to a self-adjoint operator; indeed, from an eigenvalue expression like $\hat{I}_3 \nu_L = \frac{1}{2} \nu_L$, it is easy to see that we must have
\begin{equation*}
	\hat{I}_3 = \begin{pmatrix}
		1/2 & 0\\
		0 & -1/2
	\end{pmatrix}\ .
\end{equation*}
The story with hypercharge is similar: corresponding to hypercharge $Y$ is an observable $\hat{Y}$, which is also proportional to a gauge boson, although this gauge boson lives in the complexified adjoint representation of $\U(1)$.

The details are as follows. Particles with hypercharge $Y$ span irreps\footnote{Since $\U(1)$ is abelian, all of its irreps are one-dimensional.} $\C_Y$ of $\U(1)$; by $\C_Y$ we denote the one-dimensional vector space $\C$ with action of $\alpha \in \U(1)$ given by
\begin{equation*}
	\alpha \cdot z = \alpha^{3Y} z\ .
\end{equation*}
The factor of three is inserted because $Y$ is not guaranteed to be an integer, but only an integer multiple of $1/3$. For example, the left-handed leptons $\nu_L$ and $e^-_L$ both have hypercharge $Y = -1$, so each span a copy of $\C_{-1}$. Hence, $\nu_L , e^-_L \in \C_{-1} \otimes \C^2$, where the $\C^2$ is trivial under $\U(1)$.

Now, given a particle $\psi \in \C_Y$, to find out how the gauge boson in $\C \otimes \mf{u}(1) \cong \C$ acts on it, we can differentiate the $\U(1)$ action above. We obtain
\begin{equation*}
	i \cdot \psi = 3 i Y \psi \implies \hat{Y} = \frac{1}{3} \in \C\ .
\end{equation*}
Following convention, we set the so-called $B$ boson equal to $\hat{Y}$; particles with hypercharge interact by exchanging this boson. Note that the $B$ boson is a lot like the familiar photon, and the hypercharge force which $B$ mediates is a lot like electromagnetism, except that its strength is proportional to hypercharge rather than charge.

The unification of electromagnetism and the weak force is called the \emph{electroweak interaction}. This is a $\U(1) \times \SU(2)$ gauge theory, and we have now encountered it in full detail: the fermions live in representations of hypercharge $\U(1)$ and weak isospin $\SU(2)$, and we tensor these together to get representations of  $\U(1) \times \SU(2)$. These fermions interact by exchanging $B$ and $W$ bosons, which span $\C \oplus \mf{sl} (2, \C)$, the complexified adjoint representation of $\U(1) \times \SU(2)$.

We close with a word on \emph{symmetry breaking}. Despite electroweak unification, electromagnetism and the weak force are very different at low energies, including most interactions in the everyday world. Electromagnetism is a force of infinite range that we can describe by a $\U(1)$ gauge theory with the photon as gauge boson, while the weak force is of very short range and mediated by the $W$ and $Z$ bosons: we ``define'' the photon and the $Z$ boson by the following relation:
\begin{equation}
\label{photozbosonthetaw}
	\begin{pmatrix}
		\gamma\\
		Z
	\end{pmatrix} = \begin{pmatrix}
		\cos{\thetaw} & \sin{\thetaw}\\
		-\sin{\thetaw} & \cos{\thetaw}
	\end{pmatrix} \begin{pmatrix}
		B\\
		W^0
	\end{pmatrix}\ .
\end{equation}
We have introduced here the \emph{weak mixing angle}, or \emph{Weinberg angle} $\thetaw$; it can be thought of as the parameter that characterises how far the $B - W^0$ boson plane has been rotated by symmetry breaking, which is the mechanism that allows the full electroweak $\U(1) \times \SU(2)$ symmetry group to be hidden away at low energies, and replaced with the electromagnetic subgroup $\U(1)$. Moreover, the electromagnetic $\U(1)$ is not the obvious factor $\U(1) \times 1$, but another copy, wrapped around inside $\U(1) \times \SU(2)$ in a manner given by the NNG formula. Unfortunately, the dynamics of electroweak symmetry breaking is outside of our scope; we refer the reader to \cite[Ch.~29.1]{schwartz14} for the details. We will discuss symmetry breaking from a representation theoretic viewpoint in section \ref{thenewfermionssubsec}, and return to the Weinberg angle in section \ref{weakmixinganglesec}.

\section{The Standard Model Representation}

We are now in a position to put the whole Standard Model together in a single picture. It has the gauge group
\begin{equation*}
	G_{\mrm{SM}} = \U(1) \times \SU(2) \times \SU(3)\ ,
\end{equation*}
and the fundamental fermions described thus far combine in representations of this group. We summarise this information in the table below.
\begingroup
\renewcommand{\arraystretch}{1.7}
\begin{table}[ht]
\centering
\caption{The Standard Model Fermions}
	 \begin{tabular}{c c c} 
 \hline
 Particle Name & Symbol &  $\U(1) \times \SU(2) \times \SU(3)$ Rep. \\ [0.5em] 
 \hline\hline \\[-1em]
 Left-handed leptons & $\begin{pmatrix}
	\nu_L\\
	e^-_L
 \end{pmatrix}$ & $\C_{-1} \otimes \C^2 \otimes \C$\\ 

 Left-handed quarks & $\begin{pmatrix}
	u^r_L , u^g_L , u^b_L\\
	d^r_L , d^g_L , d^b_L	
 \end{pmatrix}$ & $\C_{1/3} \otimes \C^2 \otimes \C^3$\\ 

 Right-handed neutrino & $\nu_R$ & $\C_{0} \otimes \C \otimes \C$\\
 
 Right-handed electron & $e^-_R$ & $\C_{-2} \otimes \C \otimes \C$\\ 

 Right-handed up quarks & $\left( u^r_R , u^b_R, u^g_R \right)$ & $\C_{4/3} \otimes \C \otimes \C^3$\\

 Right-handed down quarks & $\left( d^r_R , d^b_R, d^g_R \right)$ & $\C_{-2/3} \otimes \C \otimes \C^3$\\[1em]
\hline
 \end{tabular}
\label{thesmfermionstable}
\end{table}
\endgroup

All the representations of $G_{\mrm{SM}}$ in the right-hand column are irreducible, since they are made by tensoring irreps of this group's three factors. On the other hand, if we take the direct sum of all these irreps,
\begin{equation*}
	F = (\C_{-1} \otimes \C^2 \otimes \C) \oplus \cdots \oplus ( \C_{-2/3} \otimes \C \otimes \C^3 )\ ,
\end{equation*}
we get a reducible representation containing all the first-generation fermions in the Standard Model. We call $F$ the \emph{fermion representation}. If we take the dual of $F$, we get a representation describing all the antifermions in the first generation. Taking the direct sum of these spaces, $F \oplus \cc{F}$, we get a representation of $G_{\mrm{SM}}$ that we will call the \emph{Standard Model representation}; it contains all the first-generation elementary particles in the Standard Model. The fermions interact by exchanging gauge bosons that live in the complexified adjoint representation of $G_{\mrm{SM}}$.
\begingroup
\renewcommand{\arraystretch}{1.3}
\begin{table}[h]
\caption{The Standard Model Gauge Bosons}
\centering
	\begin{tabular}{c c c} 
 \hline
 Force & Gauge Boson & Symbol \\ [0.5em] 
 \hline\hline \\[-1em]
 Electromagnetism & Photon & $\gamma$\\
Weak Force & $W$ and $Z$ bosons & $W^+$, $W^-$ and $Z$\\
Strong Force & Gluons & $g$\\[1em]
\hline
 \end{tabular}
\label{thesmgaugebosonstable}
\end{table}
\endgroup

\subsubsection{Generations}
For the purposes of describing grand unified theories, the above description of the Standard Model is all we need. For completeness however, we tabulate below the \emph{second} and \emph{third} generations of fermions, evidence of which first arose in the 1930s; an elegant summary of the physics can be found in \cite{hariri77}.
\begingroup
\renewcommand{\arraystretch}{1.5}
\begin{table}[h]
\centering
\caption{Quarks by Generation}
\begin{tabular}{M{9em} M{9em} M{9em}} 
	\hline
	1st Generation & 2nd Generation & 3rd Generation\\
	\hline
\end{tabular}
\begin{tabular}{M{4em} M{4em} M{4em} M{4em} M{4em} M{4em}}
Name & Symbol & Name & Symbol & Name & Symbol\\ 
 \hline\hline \\[-1em]
 Up & $u$ & Charm & $c$ & Top & $t$\\
 Down & $d$ & Strange & $s$ & Bottom & $b$\\ \\[-1em]
\hline
 \end{tabular}
\label{quarksbygenstable}
\end{table}

\begin{table}[h]
\centering
\caption{Leptons by Generation}
\begin{tabular}{M{9em} M{9em} M{9em}} 
	\hline
  1st Generation & 2nd Generation & 3rd Generation\\
	\hline
\end{tabular}
\begin{tabular}{M{4em} M{4em} M{4em} M{4em} M{4em} M{4em}}
Name & Symbol & Name & Symbol & Name & Symbol\\
 \hline\hline \\[-1em]
 Electron neutrino & $\nu_e$ & Muon neutrino & $\nu_\mu$ & Tau neutrino & $\nu_\tau$\\
 Electron & $e^-$ & Muon & $\mu^-$ & Tau & $\tau^-$ \\ \\[-1em]
\hline
 \end{tabular}
\label{leptonsbygentable}
\end{table}
\endgroup

Notice that we thus have a pattern in the Standard Model: there are as many flavours of quarks as there are of leptons. The Pati-Salam model explains this pattern by unifying quarks and leptons, but we will unfortunately not treat this theory here; the interested reader is referred to \cite[Ch.~3.3]{baez10}.

The second and third generations of quarks and charged leptons differ from the first by being more massive and able to decay into particles of the earlier generations. The various neutrinos do not decay, and for a long time it was thought they were massless, but now it is known that some and perhaps all are massive. This allows them to change back and forth from one type to another, a phenomenon called neutrino oscillation\footnote{See the article \cite{bilenky14} for a brief review of the theoretical and experimental aspects.}; the Standard Model explain this phenomenon by recourse to the famous ``Higgs mechanism''\footnote{See \cite{hamilton15} and \cite[Ch.~28]{schwartz14}.}. For our purposes however, the generations are identical: as representations of $G_{\mrm{SM}}$, each generation spans another copy of $F$, with the corresponding generation of antifermions spanning a copy of $\cc{F}$. All told, we thus have three copies of the SM representation, $F \oplus \cc{F}$. We will only need to discuss one generation, so we find it convenient to speak as if $F \oplus \cc{F}$ contains particles of the first generation. This redundancy in the Standard Model, three sets of very similar particles, remains a mystery.

\chapter{The $\Spin (10)$ Grand Unified Theory}

\label{spin10gutchapter}
Due to spontaneous symmetry breaking, not all of the symmetries of the Standard Model are seen in everyday life---the symmetries encoded by $G_\mrm{SM}$ are symmetries of the laws of physics, but not necessarily of the vacuum. Grand unified theories attempt to answer the question, what if this process continues? That is, could the symmetries of the Standard Model be just a subset of all the symmetries in nature? By way of motivation, consider that from a representation theoretic standpoint alone, the Standard Model leaves a lot to be desired: ``The representations of $G_\mrm{SM}$ seem ad hoc. Why these? What about all the seemingly arbitrary hypercharges? Why do both leptons and quarks come in left- and right- handed varieties, which transform so differently? Why do quarks come in charges which are units $1/3$ times and electron's charge? Why are there the same number of quarks and leptons?'' \cite[p.~32]{baez10}

In this chapter, we will encounter the earliest attempts to probe these questions. We begin by introducing some additional concepts in representation theory and motivating the exceptional Lie groups, after which we will elucidate which groups are to be considered potential grand unification groups. Then we will turn to, both from necessity and because of its intrinsic interest, the $\SU(5)$ grand unified theory. Thereafter, we will focus our attention on Clifford algebras; it is a short step from there to the Spin groups; once we then understand their representations, we will prove that $\Spin (10)$ extends the Standard Model in section \ref{thespin10gutsubsection}.
%

\section{Characters and Weights of Representations}
\label{charactersandweightsofrepssubsec}

The study of characters, and root and weight systems, is fundamental to representation theory. Our modest goals in this section of simply defining these terms and stating the main results will doubtless do a severe injustice to this branch of mathematics; we point to \cite[Ch.~8]{hall04} for a lucid introduction and additional references. We will follow \cite{adams96} here.

Particle physics demands that we restrict to complex representations, so let us do so right at the outset, reaping the added benefit that over $\C$, every irrep of a compact abelian group is 1-dimensional.

\begin{rmk}[Structure Maps]
\label{structuremapsremark}
	We note that this restriction involves no sacrifice of generality: consider that a representation $V$ over the quanterions $\mbb{H}$ is certainly a representation over $\C$; together with a conjugate linear structure $G$-map $j : V \to V$ such that $j^2 = -1$, $ij = -ji$, this representation does in fact return the original $\mbb{H}$-representation; on the other hand, a representation $V$ over $\R$ gives $V \otimes_\R \C$ and this carries a conjugate linear structure map $j : v \otimes z \mapsto v \otimes \cc{z}$ such that $j^2 = 1$; we can regard $V$ as the +1 eigenspace of $j$ (or the -1 eigenspace).
\end{rmk}

\begin{defn}[Characters]
	Suppose $V$ is a representation of $G$ over $\C$. Then its character is given by
	\toenv{\chi_V^{} :\& [-29pt] G}{\C\ ,}{\& g}{\tr_\C (g : V \to V)\ .}
\end{defn}

It follows from the definition that characters are \emph{class functions}, i.e.~$\chi_V^{} (ghg\inv) = \chi_V^{} (h)$ for all $g, h \in G$. Also,
\begin{align*}
	\chi_{V \oplus W}^{} (g) &= \chi_V^{} (g) + \chi_W^{} (g)\ ,\\
	\chi_{V \otimes W}^{} (g) &= \chi_V^{} (g) \cdot \chi_W^{} (g)\ .
\end{align*}
The following result clarifies the importance of characters.

\begin{thm}
	If $\chi_V^{} = \chi_W^{}$, then $V \cong W$.
\end{thm}

A proof is found in \cite[pp.~46--52]{adams83}. Consider now the torus, $T = \prod S^1$. Because $T$ is a compact, connected abelian group, the exponential map $\exp : \mf{t} \to T$ is a homomorphism, and we can thus regard $T$ as $T \cong \mf{t}/ \Gamma$, where $\Gamma := \ker \exp$ is a discrete subgroup of $\mf{t}$, called the \emph{integer lattice} of $T$.

Homomorphisms $\theta : T \to T'$ are easily described. We need only check for a linear map $\varphi : \mf{t} \to \mf{t}'$ such that $\varphi (\Gamma) \subset \Gamma'$, and if so, then $\theta = \tilde{\varphi} : \mf{t}/\Gamma \to \mf{t}'/\Gamma'$. All continuous homomorphisms arise in this way, and all 1-dimensional representations of $T$ arise from linear maps $\varphi : \mf{t} \to \mf{u}(1)$. Here we encounter the first connection to representation theory: given a representation $V$ of $T$, there are linear maps $\varphi: \mf{t} \to \mf{u}(1)$ such that $V$ decomposes as a direct sum of non-zero sub-representations $V_\varphi$, where $\mf{t}$ acts on $V_\varphi$ by $\tau (x) = \varphi (\tau) x$ for $\tau \in \mf{t}, x \in V_\varphi$.

\begin{defn}[Weights]
\label{weightsofarepdefn}
	The linear maps $\varphi$ on $\mf{t}$ are called the weights of $V$. The dimension of $V_\varphi$ is the \emph{multiplicity} of $\varphi$.
\end{defn}

\subsection{Sketch of the Classification of Compact Lie Groups}
\label{sketchofliegroupclasssubsec}

The remarkable history of the more than one-hundred-and-fifty years of Lie theory is studied in \cite{varadarajan07}; the Killing-Cartan classification of Lie groups is arguably the highpoint of this story, and certainly a significant achievement of modern mathematics. This section is the briefest of summaries of this classification scheme, and is important for us for two reasons: (i) it motivates the existence of the exceptional Lie groups, and (ii) the roots and weights of the classical Lie groups that we will derive along the way will be instrumental in constructing $\E_6$.

\begin{defn}[Maximal Torus]
\label{maximaltorusdefn}
	A maximal torus in a compact connected Lie group $G$ is a subgroup $T$ which is (i) a torus, and (ii) maximal, i.e.~if $T \subset T' \subset G$ for $T'$ a torus, then $T' = T$.
\end{defn}

\begin{example} The maximal torii of the classical Lie groups are as follows.
\label{maximaltoursofun}
	\begin{enumerate}
		\item In $\U(n)$, consider the subgroup of matrices $\T{diag} (e^{2\pi i x_1} , \ldots , e^{2 \pi i x_n})$, for $x_j \in \R$. This is a maximal torus in $\U(n)$: any matrix in $\U(n)$ which commutes with all diagonal matrices must be diagonal, and hence in $T$. Thus, $T$ is maximal among all abelian subgroups, connected or not.
		\item Since $\C \subset \mbb{H}$, the matrices of (i) are in $\Sp (n)$, and they form a torus in $\Sp (n)$. Since $\C^n$ can be regarded as $\R^{2n}$, we get an embedding $\U(n) \to \SO(2n)$ and we can again take the corresponding matrices, namely
		\begin{equation*}
			\begin{pmatrix}
				\cos 2 \pi x_1 & - \sin 2 \pi x_1 &&&&&&\\
				\sin 2 \pi x_1 & \cos 2 \pi x_1 &&&&&&\\
				&& \cos 2 \pi x_2 & - 2 \pi \sin x_2 &&&&\\
				&& \sin 2 \pi x_2 & \cos 2 \pi x_2 &&&&\\
				&&&& \ddots &&&\\
				&&&&& \cos 2 \pi x_n & - \sin 2 \pi x_n &\\
				&&&&& \sin 2 \pi x_n & \cos 2 \pi x_n &\\
			\end{pmatrix}\ .
		\end{equation*}
		These will form a torus in $\SO (2n)$.
		\item We can embed $\R^{2n}$ in $\R^{2n + 1}$ and thus $\SO(2n)$ in $\SO (2n + 1)$, where we map $A \mapsto \begin{psmallmatrix}
			A & 0\\
			0 & 1
		\end{psmallmatrix}$. This is the corresponding torus in $\SO (2n + 1)$.
		\item In $\SU(n)$ we take the matrices of (i) subject to $\sum x_i = 0$ to get a maximal torus. 		
 	\end{enumerate}
\end{example}

Maximal tori are fundamental in representation theory, as the following results demonstrate. Their proofs can be found in \cite[pp.~89--95]{adams83}.

\begin{thm}
	Let $T \subset G$ be a maximal torus of a compact, connected Lie group $G$. Then any $g \in G$ is conjugate to some element of $T$. That is, there exist elements $t \in T$, $h \in G$ such that $g = h t h\inv$.
\end{thm}

\begin{cor}
	If $V, W$ are representations of a compact connected Lie group $G$ and $\chi_V^{}\big|_T^{} = \chi_W^{}\big|_T^{}$, then $\chi_V^{} = \chi_W^{}$, so $V \cong W$.
\end{cor}

Hence the weights (together with the multiplicities) of a representation $V$ of $G$, determine $V$ up to equivalence. We also have
 
\begin{cor}
	Given two maximal tori $T, T'$ in a compact connected Lie group $G$, there exists an inner automorphism of $G$ taking $T$ to $T'$.
\end{cor}

It follows from this corollary that any property of $G$ defined by reference to a maximal torus $T$ is independent of the choice of $T$. The most important example of this is the following

\begin{defn}[Rank]
	The rank of a compact connected Lie group $G$ is the dimension of the maximal torus of $G$. We will usually write $l = \rank G$.
\end{defn}

Suppose now that $T \subset G$ is a torus (not necessarily maximal). Then $G$ acts on $\mf{g}$ via the adjoint representation, so $T$ acts on $\mf{g}$ by restriction and $\mf{g} \otimes \C$ splits as a sum of 1-dimensional representations of $T$, with $T$ acting trivially on $\mf{t} \subset \mf{g}$. Thus the trivial 1-dimensional representation occurs at least $d = \dim T$ times. In fact, we have

\begin{prop}
	If $G$ is compact, then $T$ is maximal if and only if the trivial 1-dimensional representation occurs exactly $d$ times.
\end{prop}

A proof of this result can be found in \cite[p.~83]{adams83}. Henceforth, we suppose that $T$ is maximal and set $d = l$.

\begin{defn}[Roots]
\label{rootsdefn}
	The roots of a compact connected Lie group $G$ are the weights of the adjoint representation, excluding 0 (which occurs $l$ times).
\end{defn}
  
The roots are thus $\R$-linear functions on $\mf{t}$, that is, elements of $\mf{t}^*$. Since the adjoint representation of $G$ is real, the 1-dimensional summands of $\mf{g} \otimes \C$ occur in conjugate pairs and the roots occur in pairs $\pm \theta$.

\begin{example}
\label{rootsoftheothermatrixgroups}
The roots of the classical Lie groups are as follows.
\begin{itemize}
	\item For the maximal torus of $\U(n)$ described in example \ref{maximaltoursofun} above, the weights are 0, $n$ times, and $\pm (x_i - x_j)$, where $1 \leq i < j \leq n$.
	\vspace{-\topsep}
	\begin{proof}
	First note that $\mf{u}(n) \otimes \C \cong \mf{gl} (n, \C) \cong \End(\C^n)$ since $B \mapsto - B^\dagger$ is a conjugate linear structure map and its +1 eigenspace is $\mf{u} (n)$. Now take the basis $\{e_j\}$ of standard column vectors for $\C^n$; for $i < j$, define linear maps $\theta_{ij} \in \End(\C^n)$ by $\theta_{ij} (e_j) = e_i$, $\theta_{ij}(e_k) = 0$ for $k \neq j$. The matrix of $\theta_{ij}$ has a 1 in the $ij$-th place and zeroes elsewhere. The $\theta_{ij}$ are eigenvectors of the action of $T$ with eigenvalues $\exp (2 \pi i (x_i - x_j))$, so $x_i - x_j$ are eigenvalues for the action of $\mf{t}$ on $\mf{u}(n)$. (Here, we are taking $\mf{t}$ to be the diagonal matrices $\T{diag} ( i y_1, \ldots , i y_n)$, $y_j \in \R$, and $x_i \in \mf{t}^*$ is given by $x_i (\T{diag} ( i y_1, \ldots , i y_n)) = i y_i$.)
\end{proof}
\item The roots of the other matrix groups are\smallbreak
\begingroup
\renewcommand{\arraystretch}{1.2}
\begin{tabular}{l l l l l l}
\label{weightsofothermatrixgroups}
	$\SU(n)$ & : & $\pm (x_i - x_j)$, & $1 \leq i < j \leq n$; & 0, & $(n - 1)$ times.\\
	$\SO(2n)$ & : & $\pm x_i \pm x_j$, & $1 \leq i < j \leq n$; & 0, & $n$ times.\\
	$\SO(2n + 1)$ & : & $\pm x_i \pm x_j$, & $1 \leq i < j \leq n$; & 0, & $n$ times;\\
	&& $\pm x_i$, & $1 \leq i \leq n$.&&\\
	$\Sp(n)$ & : & $\pm x_i \pm x_j$, & $1 \leq i < j \leq n$; & 0, & $n$ times;\\
		&& $\pm 2 x_i$, & $1 \leq i \leq n$.&&\\
\end{tabular}
\endgroup
\end{itemize}
\end{example}

\begin{defn}[Weyl Group]
	The Weyl group $W$ of a compact, connected Lie group $G$ is the group of those automorphisms of a maximal torus which are given by inner automorphisms of $G$.
\end{defn}

\begin{example}
	In $\U(2)$, conjugation by $\begin{psmallmatrix}
		0 & -1\\
		1 & 0
	\end{psmallmatrix}$ is an element of $W$, and
	\begin{equation*}
		\begin{pmatrix}
			0 & -1\\
			1 & 0
		\end{pmatrix} \begin{pmatrix}
			e^{2\pi i x_1} & 0\\
			0 & e^{2 \pi i x_2}
		\end{pmatrix} \begin{pmatrix}
			0 & 1\\
			-1 & 0
		\end{pmatrix} = \begin{pmatrix}
			e^{2\pi i x_2} & 0\\
			0 & e^{2 \pi i x_1}
		\end{pmatrix} \in T\ .
	\end{equation*}
The Weyl groups of the other classical matrix groups are as follows \cite[pp.114--116]{adams83}\smallbreak
\begingroup
\renewcommand{\arraystretch}{1.2}
\begin{tabular}{l l p{9cm}}
	$\U(n)$ and $\SU(n)$ & : & $W$ = any permutation of $x_1, \ldots , x_n$.\\
	$\Sp (n)$ and $\SO (2n + 1)$ & : & $W = $ the group generated by all permutations of $x_1 , \ldots , x_n$ and all sign changes of $x_i$.\\
	$\SO(2n)$ & : & $W = $ the group generated by all permutations of $x_1 , \ldots , x_n$ and an even number of sign changes of $x_i$.
\end{tabular}
\endgroup
\end{example}

The Weyl group $W$ acts on $\mf{t}$ and permutes roots. If we regard the roots as elements of $\mf{t}^*$, they form a configuration with great symmetry and very distinctive properties \cite[Ch.~5]{adams83}. The Dynkin diagram encodes this configuration, as we proceed to describe.

We may choose on $\mf{t}^*$  a positive definite inner product invariant under $W$, so that we can define the lengths and angles of roots.  For each pair of roots $\pm \theta$, the kernel, $\ker \theta = \ker (-\theta)$, is a hyperplane in $\mf{t}$ called a \emph{root plane}. Conversely, it can be shown that each root plane comes from only one pair of roots, $\pm \theta$. The root planes form a figure in $\mf{t}$ called the \emph{(infinitesimal) Stiefel diagram}.

The root planes divide $\mf{t}$ into convex open sets called \emph{Weyl chambers}, and the Weyl group permutes these chambers in a way which is simply transitive. We choose one and call it the \emph{fundamental Weyl chamber} (FWC); we denote it by $C$. A root $\theta$ is \emph{positive} (resp.~\emph{negative}) if $\theta > 0$ (resp.~$\theta < 0$) on $C$. A positive root is \emph{simple} if it defines a wall of $C$; in the Stiefel diagram of $\SO (5)$ shown in figure \ref{stiefeldiagramofso5}, the roots $x_1$ and $x_1 - x_2$ are simple.
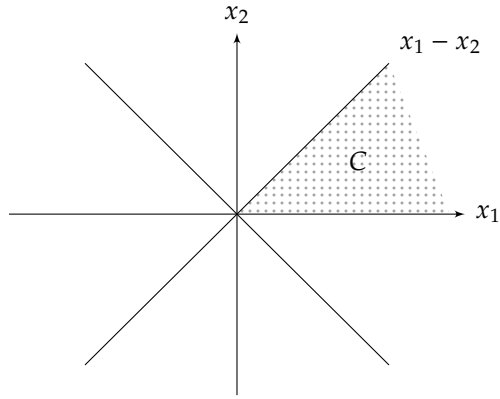
\begin{figure}
\centering
	\begin{tikzpicture}
	\draw[-latex'] 
		(-3, 0) -- (3, 0) node [right] {\small $x_1$};
	\draw[-latex'] 
		(0, -2.4) -- (0, 2.4) node [above] {\small $x_2$};
	\draw
		(-2, -2) -- (2, 2) node [above right] {\small $x_1 - x_2$}
		(-2, 2) -- (2, -2);
	\draw[pattern = dots, draw = none, opacity = 0.6]
		(0,0) -- (2.8,0) -- (2, 2);
	\draw (1.6, 0.7) node {\small $C$};
	\end{tikzpicture}
	\caption{The Stiefel Diagram of $\SO(5)$}
	\label{stiefeldiagramofso5}
\end{figure}

The \emph{Dynkin diagram} has one node for each simple root (i.e.~for each wall of the fundamental Weyl chamber). These nodes are joined by the following number of bonds:
\begin{equation*}
	\begin{cases}
		0 & \T{if the roots are orthogonal}\ ,\\
		1 & \T{if the roots are at } 60^{\circ} \T{ or }120^{\circ}\ ,\\
		2 & \T{if the roots are at } 45^{\circ} \T{ or }135^{\circ}\ ,\\
		3 & \T{if the roots are at } 30^{\circ} \T{ or }150^{\circ}\ .
	\end{cases}
\end{equation*}
These are the only possibilities \cite[pp.118--121]{adams83}. By definition, the Dynkin diagram of the torus is empty.

\begin{example} The Dynkin diagrams of the classical Lie groups are as follows.
\label{bigexamplewithdynkindiagrams}
	\begin{enumerate}
		\item For $\U (n)$, take $C$ to have $x_1 < x_2 < \cdots < x_n$. We obtain\\
			\begin{minipage}{\linewidth}
				\centering
				\bigskip
           \begin{tikzpicture}[scale=0.9, transform shape]
						\draw[fill=black]
						(0,0) circle [radius=.08] node [above] {\small $-x_1 + x_2$}
						(2,0) circle [radius=.08] node [above] {\small $-x_2 + x_3$}
						(4,0) circle [radius=.08] node [above] {\small $-x_3 + x_4$}
						(8,0) circle [radius=.08] node [above] {\small $-x_{n-1} + x_n$};
						\draw 
						(0, 0) -- (2, 0)
						(2, 0) -- (4, 0)
						(4, 0) -- (4.5, 0)
						(7, 0) -- (8, 0);
						\draw
						(5,0) -- (6.5, 0) [thick, loosely dotted];											
						\end{tikzpicture}
						\bigskip
      \end{minipage}
			$\SU (n)$ has the same Dynkin diagram as $\U(n)$, which is traditionally labelled $A_{n-1}$. (We always take the usual inner product on $\R^n$.)
		\item For $\SO(2n)$, take $C$ to have $-x_2 < x_1 < x_2 < x_3 < \cdots < x_n$. Then the diagram is denoted $D_n$ and is given by\\
			\begin{minipage}{\linewidth}
				\centering
				\bigskip
           \begin{tikzpicture}[scale=0.9, transform shape]
						\draw[fill=black]
						(-2, 1) circle [radius=.08] node [above] {\small $x_1 + x_2$\ \ \ \ \ \ \ \ \ \ }
						(-2, -1) circle [radius=.08] node [above] {\small $-x_1 + x_2$\ \ \ \ \ \ \ \ \ \ \ \ }
						(0,0) circle [radius=.08] node [above] {\small \ \ \ \ \ \ $-x_2 + x_3$}
						(2,0) circle [radius=.08] node [above] {\small $-x_3 + x_4$}
						(6,0) circle [radius=.08] node [above] {\small $-x_3 + x_4$}
						(8,0) circle [radius=.08] node [above] {\small $-x_{n-1} + x_n$};
						\draw
						(-2, 1) -- (0, 0)
						(-2, -1) -- (0, 0)
						(0, 0) -- (2, 0)
						(2, 0) -- (2.5, 0)
						(5, 0) -- (6, 0)
						(6, 0) -- (8, 0);
						\draw
						(3,0) -- (4.5, 0) [thick, loosely dotted];											
						\end{tikzpicture}
					\bigskip
      \end{minipage}
		\item For $\SO (2n + 1)$, take $C$ to have $0 < x_1 < x_2 < \cdots < x_n$. Then we have the Dynkin diagram $B_n$:\\
					\begin{minipage}{\linewidth}
				\centering
				\bigskip
           \begin{tikzpicture}[scale=0.9, transform shape]
						\draw[fill=black]
						(0,0) circle [radius=.08] node [above] {\small $x_1$} node [below, align=center] {\ \\ \small short}
						(2,0) circle [radius=.08] node [above] {\small $-x_1 + x_2$}
						(4,0) circle [radius=.08] node [above] {\small $-x_2 + x_3$}
						(8,0) circle [radius=.08] node [above] {\small $-x_{n-1} + x_n$};
						\draw 
						(0, 0.04) -- (2, 0.04)
						(0, -0.04) -- (2, -0.04)						
						(2, 0) -- (4, 0)
						(4, 0) -- (4.5, 0)
						(7, 0) -- (8, 0);
						\draw
						(5,0) -- (6.5, 0) [thick, loosely dotted];											
						\end{tikzpicture}
					\bigskip
      \end{minipage}
			\item For $\Sp (n)$, take $C$ to have $0 < x_1 < x_2 < \cdots < x_n$. Then we have the Dynkin diagram $C_n$:\\
					\begin{minipage}{\linewidth}
				\centering
				\bigskip
           \begin{tikzpicture}[scale=0.9, transform shape]
						\draw[fill=black]
						(0,0) circle [radius=.08] node [above] {\small $2x_1$} node [below, align=center] {\ \\ \small long}
						(2,0) circle [radius=.08] node [above] {\small $-x_1 + x_2$}
						(4,0) circle [radius=.08] node [above] {\small $-x_2 + x_3$}
						(8,0) circle [radius=.08] node [above] {\small $-x_{n-1} + x_n$};
						\draw 
						(0, 0.04) -- (2, 0.04)
						(0, -0.04) -- (2, -0.04)						
						(2, 0) -- (4, 0)
						(4, 0) -- (4.5, 0)
						(7, 0) -- (8, 0);
						\draw
						(5,0) -- (6.5, 0) [thick, loosely dotted];											
						\end{tikzpicture}
					\bigskip
      \end{minipage}
			Note that ``short'' (resp.~``long'') means $x_1 \cdot x_1 = 1 < 2$ (resp.~$2x_1 \cdot 2x_1 = 4 > 2$).
	\end{enumerate}
\end{example}

We now finally have a rule which associates to a compact Lie group a Dynkin diagram. The upshot: if $G$ and $G'$ are compact Lie groups, then $\mf{g}$ is isomorphic to $\mf{g}'$ if and only if $\mf{g} \otimes \C$ is isomorphic to $\mf{g}' \otimes \C$; thus the Dynkin diagram determines $\mf{g}$, and hence $G$, locally. In particular, corresponding to each Dynkin diagram, there is a unique compact, connected, simply connected Lie group, because to each of the diagrams in the \emph{Killing-Cartan classification}, i.e.~example \ref{bigexamplewithdynkindiagrams}, plus the exceptional Dynkin diagrams below, there is a unique simple Lie algebra---in fact, every complex simple Lie algebra is isomorphic to one of the algebras in this classification scheme---and hence a unique connected, simply connected, compact, simple Lie group. As we saw above, the groups $\SU(n + 1)$, $n \geq 1$, and $\T{Sp} (n)$, $n \geq 3$ correspond to $A_n$ and $C_n$ respectively; the groups $\Spin (2n + 1)$, $n \geq 2$ and $\Spin (2n)$, $n \geq 4$ correspond to $B_n$ and $D_n$. All these groups have rank $n$ and are pairwise non-isomorphic. The non-classical or ``exceptional'' Dynkin diagrams are as follows; the notation and conventions are explained in \cite[Ch.~9]{adams96}.

\begingroup
\renewcommand*{\arraystretch}{2.5}
\begin{tabular}{m{5em} l}
	$\mrm{G}_2$ &	\noindent\parbox[c]{\hsize}{\begin{tikzpicture}[scale=0.9, transform shape]
							\draw[fill=black]
							(0,0) circle [radius=.08] 
							(1,0) circle [radius=.08];
							\draw 
							(0,-.06) --++ (1,0)
							(0, 0) --++ (1, 0)
							(0,+.06) --++ (1,0);                      
					  \end{tikzpicture}}\\
	$\mrm{F}_4$ & \noindent\parbox[c]{\hsize}{\begin{tikzpicture}[scale=0.9, transform shape]
						\draw[fill=black]
						(0,0) circle [radius=.08] node [above] {\small 1} node [below] {\small $s$}
						(1,0) circle [radius=.08] node [above] {\small 2} node [below] {\small $s$}
						(2,0) circle [radius=.08] node [above] {\small 3} node [below] {\small $l$}
						(3,0) circle [radius=.08] node [above] {\small 4} node [below] {\small $l$};
						\draw 
						(0, 0) -- (1, 0)
						(1, 0.04) -- (2, 0.04)
						(1, -0.04) -- (2, -0.04)
						(2, 0) -- (3, 0);										
					 \end{tikzpicture}} \\
	$\E_6$ & \noindent\parbox[c]{\hsize}{\begin{tikzpicture}[scale=0.9, transform shape]
						\draw[fill=black]
						(0,0) circle [radius=.08] node [below] {\small 1} 
						(1,0) circle [radius=.08] node [below] {\small 2} 
						(2,0) circle [radius=.08] node [below] {\small 4}
						(2,1) circle [radius=.08] node [above] {\small 5}
						(3,0) circle [radius=.08] node [below] {\small 3}
						(4,0) circle [radius=.08] node [below] {\small 6};
						\draw 
						(0, 0) -- (1, 0)
						(1, 0) -- (2, 0)
						(2, 0) -- (2, 1)
						(2, 0) -- (3, 0)
						(3, 0) -- (4, 0);										
					 \end{tikzpicture}}\\
	$\E_7$ & \noindent\parbox[c]{\hsize}{\begin{tikzpicture}[scale=0.9, transform shape]
						\draw[fill=black]
						(0,0) circle [radius=.08] node [below] {\small 7} 
						(1,0) circle [radius=.08] node [below] {\small 1} 
						(2,0) circle [radius=.08] node [below] {\small 2}
						(3,0) circle [radius=.08] node [below] {\small 4}
						(3,1) circle [radius=.08] node [above] {\small 3}
						(4,0) circle [radius=.08] node [below] {\small 5}
						(5,0) circle [radius=.08] node [below] {\small 6};
						\draw 
						(0, 0) -- (1, 0)
						(1, 0) -- (2, 0)
						(2, 0) -- (3, 0)
						(3, 0) -- (3, 1)
						(3, 0) -- (4, 0)
						(4, 0) -- (5, 0);										
					 \end{tikzpicture}}\\
	$\E_8$ & \noindent\parbox[c]{\hsize}{\begin{tikzpicture}[scale=0.9, transform shape]
						\draw[fill=black]
						(0,0) circle [radius=.08] node [below] {\small 1} 
						(1,0) circle [radius=.08] node [below] {\small 2} 
						(2,0) circle [radius=.08] node [below] {\small 4}
						(2,1) circle [radius=.08] node [above] {\small 3}
						(3,0) circle [radius=.08] node [below] {\small 5}
						(4,0) circle [radius=.08] node [below] {\small 6}
						(5,0) circle [radius=.08] node [below] {\small 7}
						(6,0) circle [radius=.08] node [below] {\small 8};
						\draw 
						(0, 0) -- (1, 0)
						(1, 0) -- (2, 0)
						(2, 0) -- (2, 1)
						(2, 0) -- (3, 0)
						(3, 0) -- (4, 0)
						(4, 0) -- (5, 0)										
						(5, 0) -- (6, 0);										
					 \end{tikzpicture}}
\end{tabular}
\endgroup
\bigskip

We are primarily interested in the Lie group corresponding to the diagram $\E_6$, but to arrive at the same, we will need to construct the group $\E_8$; we do so in section \ref{theconstructionofe8sec}. Additionally, we describe the construction of the smallest exceptional Lie group $\mrm{G}_2$ in an appendix, by way of an illustrative example.

\section{Possible Grand Unification Groups}

By way of motivating grand unified theories, we have already raised several questions about the unsatisfactory aesthetics of the Standard Model representation. We introduce now considerations of a more technical nature, which will help us ``classify'' grand unification groups as it were, i.e.~understand which groups are preferred from the plenitude of available Lie groups that contain embeddings of the Standard Model.

\subsubsection{Coupling Constants in Gauge Theories}

\begin{defn}[Killing Form]
\label{thekillingformdefn}
	Let $A$ be a Lie algebra. The map
	\toenv{A \times A}{\C\ ,}{(X, Y)_K}{\tr (\ad (X) \circ \ad (Y))}
	is called its Killing form.
\end{defn}

The symmetry and bilinearity of this form are easy to check; it is also immediately clear that the Killing form of an abelian Lie algebra is zero. More interesting is the following

\begin{prop}
	Let $\sigma : A \to A$ be a Lie algebra automorphism. Then $(\sigma X, \sigma Y)_K = (X, Y)_K$ for all $X, Y \in A$. For $A = \mf{g}$ the Lie algebra of some Lie group, this holds in particular for the automorphism $\Ad (g)$ for an arbitrary $g \in G$.
\end{prop}

\begin{proof}
	Since $\sigma$ is a Lie algebra automorphism, we have
	\begin{equation*}
		\ad (\sigma X) Y = [\sigma X, Y] = \sigma ([X, \sigma\inv Y]) = (\sigma \circ \ad (X) \circ \sigma\inv) (Y)\ .
	\end{equation*}
	We hence compute
	\begin{align*}
		(\sigma X, \sigma Y)_K &= \tr (\ad (\sigma X) \circ \ad (\sigma Y))\\
		&= \tr (\sigma \circ \ad (X) \circ \ad (Y) \circ \sigma\inv)\\
		&= \tr (\ad (X) \circ \ad (Y))\\
		&= (X, Y)_K \qedhere
	\end{align*}
\end{proof}

We introduce some more terminology: if $\mf{g}$ is the Lie algebra of a compact Lie group $G$, it is in turn called \emph{compact}; a subspace $\mf{i} \subset \mf{g}$ is called an \emph{ideal} if it is closed under the Lie bracket, and satisfies $[\mf{g} , \mf{i}] \subseteq \mf{i}$; a Lie algebra is called \emph{simple} if its only ideals are $0$ and itself. Now, one can show that for $\mf{g}$ compact and simple, the negative of the Killing form is a positive-definite inner product; moreover, it turns out that up to a positive constant, it is the unique such form. The proof of this is not very hard, and can be found in \cite[Ch.~8.1]{fuchs03}, for example. From this, one can deduce the following result (see \cite[Ch.~2.10]{hamilton16} for a proof) which will in turn finally allow us to make the definition that we are after.

\begin{thm}
	Let $G$ be a compact, connected Lie group of the form
	\begin{equation*}
		G = \U(1) \times \cdots \times \U(1) \times G_1 \times \cdots \times G_n\ ,
	\end{equation*}
	(up to a finite quotient) where the $G_i$ are simple. Let $k : \mf{g} \times {g} \to \C$ be an $\Ad$-invariant positive definite scalar product on the Lie algebra $\mf{g}$. Then $k$ is the orthogonal sum of
	\begin{itemize}
		\item a positive definite scalar product $k_0$ on the abelian algebra $\mf{u}(1) \oplus \cdots \oplus \mf{u} (1)$, and
		\item $\Ad_{G_i}$-invariant positive definite scalar products $k_i$'s on the Lie algebras $\mf{g}_i$.
	\end{itemize}
	The scalar product $k_0$ is determined by a positive definite symmetric matrix, and the scalar products $k_i$ are determined by positive constants relative to some fixed $\Ad$-invariant positive definite scalar products on the corresponding Lie algebras (such as the negative Killing form).
\end{thm}

\begin{defn}[Coupling Constants]
	The constants that determine the positive definite $\Ad$-invariant scalar products on the abelian ideal $\mf{u} (1) \oplus \cdots \oplus \mf{u}(1)$ and the $\mf{g}_i$-summands relative to some standard $\Ad$-invariant scalar products, are called coupling constants.	
\end{defn}

Some insight from physics is in order. Gauge couplings are simply numbers, determined by experiment, that fix the interaction strength of the field that they correspond to. They are encountered most directly in \emph{pure Yang-Mills theories}, which lie at the heart of both electroweak unification and QCD. We consider them briefly, returning to the framework of gauge theory; we follow \cite[Ch.~7.2]{hamilton16}. Let $G: P \to M$ be a principal bundle with the structure group $G$ compact and finite dimensional. Further, fix an $\Ad$-invariant positive-definite scalar product $k$ on $\mf{g}$ as in the theorem above, and a $k$-orthonormal basis for $\mf{g}$. For $A$ a connection $1$-form with curvature $2$-form $F^A \in \Omega^2 (P, \mf{g})$, in a local gauge $s : M \supset U \to P$, the field strength is given by
\begin{equation*}
	F^A_s := s^* F^A \in \Omega^2 (M, \mf{g})\ .
\end{equation*}
The Yang-Mills Lagrangian is then simply defined by
\begin{equation*}
	\mc{L}_\mrm{YM}^{} = - \frac{1}{2} k (F_s^A , F_s^A)\ .
\end{equation*}
In the case that $G$ is simple, for instance, there is a single coupling constant $g > 0$ and it is clear that this ``determines the field strength'' in the sense that it directly scales the inner product that appears in the Lagrangian.

This brings us to unification. The coupling constants of the strong, weak, and electromagnetic interactions are known to be different at low energies ($\approx 1\,\mrm{GeV}$); to wit, the strong interaction is observed to be much stronger (obviously) than the weak and electromagnetic couplings. However in principle, there is no reason that the couplings cannot be unified at high energies, because in quantum field theory, these constants are in fact not constant; they depend on the energy scale. This phenomenon is known as \emph{renormalisation group running}. Calculations show (see \cite[Ch.~5.5]{mohapatra02}) that if the coupling constants are normalised as in the previous paragraph, i.e.~taken to be orthonormal with respect to the Killing form on $\mf{g}_\mrm{SM}$, the renormalisation group equation indicates that they roughly converge at high energies. This is a plausibility argument for a grand unification group with a single coupling constant, unifying the three forces of the Standard Model at high energies; this can only occur if the unification group is simple, or a product of identical simple groups, where the coupling constant for each factor is set the same by forcing the theory to have some sort of discrete symmetry. This is the first demand that we will make of any potential grand unification group.

\subsubsection{Chirality and Complex Representations}

In mathematics, the term ``complex representation'' simply refers to a group representation on a complex vector space; the term as used in physics denotes something different, and it is related to chirality. As we have seen, the weak force, and hence the Standard Model, is chiral. This unexpected feature detracts significantly from the symmetry of the rest of the theory, and one might expect that grand unified theories behave more naturally, or at least somehow explain this parity violation. But this is in fact not the case: Georgi \cite{georgi79} and Barbieri et al.~\cite{barbieri80b} have argued that the fermions that would have to be introduced into an achiral grand unified theory to recover the chirality of the Standard Model on symmetry breaking would be unacceptably heavy; this is an instance of the \emph{Survival Hypothesis}, which we will discuss in more detail in section \ref{furtherreadingsec}. For the moment, we will content ourselves with defining a complex representation, and seeing how it is concerned with chirality.

\begin{defn}[Complex Representation]
	Two representations $\pi_1 : G \to \GL (V_1)$ and $\pi_2 : G \to \GL (V_2)$ of a group $G$ are said to be \emph{equivalent} if there is an intertwining operator from $V_1$ to $V_2$ such that it is also a vector space isomorphism. If $\pi : G \to \GL (V)$ is a representation, the \emph{complex conjugate representation} $\cc{\pi}$ is defined over the complex conjugate vector space $\cc{V}$ by $\cc{\pi} (g) = \cc{\pi (g)}$. A representation of a group is said to be complex if it is not equivalent to its complex conjugate representation.
\end{defn}

The connection to handedness is pretty straightforward. We know that the way to get the antiparticle representation from the particle representation is simply to pass to the dual; for example, $\cc{\nu}_R \in \C_{1} \otimes \C^2 \otimes \C \cong \cc{\C_{-1} \otimes \C^2 \otimes \C} \ni \nu_L$ (we have used here the fact that $\C^2 \cong \cc{\C^2}$ under $\SU(2)$). So if we take a direct sum of the representations of all the left-handed fermions, call this $f_L$, it stands to reason that the direct sum of all the right-handed fermion representations is given by $f_R = \cc{f_L}$. Therefore, if $f_L \cong \cc{f_L}$, i.e.~if $f_L$ is real, the theory is manifestly achiral, since the right-handed particles transform as the left; such theories are called \emph{vectorlike}; on the other hand, if $f_L$ is complex, the theory is chiral. We will hence demand that our grand unification groups admit complex representations, to preserve this feature of the standard model. Let us summarise our work in the following

\begin{defn}[Possible Unification Group]
\label{possiblegrandunifgroupdefn}
	We call a Lie group $G$ a possible unification group if it satisfies the following properties.
	\begin{itemize}
		\item $G$ is simple, or a product of several copies of the same simple group.
		\item	$G$ contains (perhaps up to a finite quotient) the Standard Model gauge group $G_\mrm{SM}$.
		\item $G$ admits complex representations.
	\end{itemize}
\end{defn}

\subsubsection{Classification of Unification Groups}

We will restrict our discussion to Lie groups with rank less than $7$, since they are generally considered the most interesting for unification\footnote{For a discussion on higher rank unification groups, see \cite[Ch.~3.4]{langacker80} and \cite[Ch.~3]{slansky81}.}; they are listed as follows:
\begin{itemize}
	\item rank $1$: $\SU(2)$;
	\item rank $2$: $\SU(3)$, $\Spin (5)$, $\mrm{G}_2$;
	\item rank $3$: $\SU(4)$, $\Spin (7)$, $\Sp (3)$;
	\item rank $4$: $\SU(5)$, $\Spin (8)$, $\Spin (9)$, $\Sp (4)$, $\mrm{F}_4$\;
	\item rank $5$: $\SU(6)$, $\Spin (10)$, $\Spin (11)$, $\Sp (5)$;
	\item rank $6$: $\SU(7)$, $\Spin (12)$, $\Spin (13)$, $\Sp (6)$, $\E_6$.
\end{itemize}
Of these, the only ones for which we have not explicitly computed the rank are the exceptional groups $\mrm{G}_2$, $\mrm{F}_4$ and $\E_6$. Since we will momentarily eliminate $\mrm{F}_4$ as a possible unification group, we will not bother with this computation\footnote{The interested reader may refer to \cite[Ch.~8]{adams96}.}; the rank of $\E_6$ is computed in the proof of theorem \ref{bigliegrouptabletheorem}, and of $\mrm{G}_2$ in the appendix.

Mehta and Srivastava have classified the complex representations of all the classical Lie groups: only the $\SU(n)$'s, for $n > 2$, the $\Spin (4n + 2)$'s, for $n \geq 1$, and $\E_6$ admit complex representations \cite{mehta66, mehta66b}. Together with the fact that the gauge group of the Standard Model $G_\mrm{SM} = \U(1) \times \SU(2) \times \SU(3)$ has rank equal to $1 + 1 + 2 = 4$, and the further requirement on simplicity from definition \ref{possiblegrandunifgroupdefn}, we can immediately thin down the above list significantly.\footnote{In section \ref{anomalycancellationsec} we will consider the issue of anomaly cancellation, and its consequences for grand unified theories. This rather subtle requirement from quantum field theory is hard to motivate from a representation theoretic standpoint alone (though it does have a nice interpretation in the same), and was hence omitted in this section. In any case, it has no bearing on our list of possible unification groups.}

\begin{prop}
	The only possible grand unification groups with rank less than $7$ are the following:
	\begin{itemize}
		\item rank $4$: $\SU(3)^2$ and $\SU(5)$;
		\item rank $5$: $\SU(6)$ and $\Spin (10)$;
		\item rank $6$: $\SU(3)^3$, $\SU(4)^2$, $\SU(7)$ and $\E_6$.
	\end{itemize}
\end{prop}

We provide references for these grand unified theories, where they exist. In the same paper \cite{georgi74} in which they proposed the $\SU(5)$ theory, Georgi and Glashow ruled out an $\SU(3)^2$ theory for physical reasons, leaving $\SU(5)$ the unique rank $4$ unification group; we will turn to this theory in the next section. A theory with unification group $\SU(6)$ was suggested in $2005$ by Hartanto and Handoko \cite{hartanto05}, while the $\Spin (10)$ grand unified theory was put forward by Georgi in $1974$ \cite{georgi74b} and Fritzsch and Minkowski in $1975$ \cite{fritzsch75}. Finally, a theory with $\SU(3)^3$ as gauge group, called \emph{trinification}, was demonstrated by de R\'uluja et al.~in $1984$ \cite{rujula75}, an $\SU(7)$ grand unification theory was studied by Umemura and Yamamoto in $1981$ \cite{umemura81}, and the subject of this thesis, the $\E_6$ grand unified theory, first appeared in a 1976 paper by G\"ursey et al.~\cite{gursey75}.

\section{The $\SU(5)$ Grand Unified Theory}
\label{thesu5gutsection}

Georgi and Glashow's $\SU(5)$ extension of the Standard Model was the first grand unified theory, and is still considered the prototypical example of the same. Unfortunately this theory has since been ruled out by experiment: it predicts that protons will decay faster than the current lower bound on proton lifetime.\footnote{Detailed studies and reviews of this theory abound in the literature, see \cite{mohapatra02, rosen} and references therein.} Our focus here will be simply to show what exactly we mean when we say that $\SU(5)$ is a grand unified theory; the questions we will ask and methodology we will develop will be highly instructive for us when we later consider the $\Spin (10)$ theory, and eventually the one of $\E_6$. We closely follow \cite{baez10} in this section.

For integers $m, n \geq 1$, let us define $\mrm{S} (\U (m) \times \U(n)) = \{ (A, B) \in \U(m) \times \U(n) \mid \det A \cdot \det B = 1 \}$. This Lie group is naturally a subgroup of $\SU (m + n)$ under the embedding
	\toenv{\mrm{S}(\U (m) \times \U(n))}{\SU(m + n)\ ,}{(A, B)}{\begin{pmatrix}
	A & 0\\
	0 & B
\end{pmatrix}\ .}
The key to the whole $\SU(5)$ theory is the following: the subgroup $\mrm{S} (\U(2) \times \U(3))$ is isomorphic to $G_\mrm{SM}$, modulo a finite subgroup. More precisely, consider the map
\toenv{\phi : \&[-29pt] \U(1) \times \SU(2) \times \SU(3)}{\SU(5)\ ,}{\& (\alpha, A, B)}{\begin{pmatrix}
	\alpha^3 A & 0\\
	0 & \alpha^{-2} B
\end{pmatrix}\ ;}
this is clearly a homomorphism from $G_\mrm{SM}$ to $\mrm{S} (\U(2) \times \U(3))$. Equally clear is the fact that it is not injective: its kernel is all elements of the form $(\alpha, \alpha^{-3}, \alpha^2)$. This kernel is $\Z_6$, because scalar matrices $\alpha^{-3}$ and $\alpha^2$ live in $\SU(2)$ and $\SU(3)$ simultaneously if and only if $\alpha$ is a sixth root of unity. So in short order, we have obtained
\begin{equation*}
	G_\mrm{SM} / \Z_6 \cong \mrm{S} (\U(2) \times \U(3)) \hookrightarrow \SU(5)\ .
\end{equation*}

This sets up a test that the $\SU(5)$ theory must pass for it to have any chance of success: not all representations of $G_\mrm{SM}$ factor through $G_\mrm{SM} / \Z_6$, but all those coming from representations of $\SU(5)$ must do so. In particular, we have to check that $\Z_6$ acts trivially on all the irreps inside $F$, that is, it must act trivially on all fermions (and antifermions, but that amounts to the same thing). For this to be true, some non-trivial relations between hypercharge, isospin and colour must hold. Consider for example the electron	$e_L^- \in \C_{-1} \otimes \C^2 \otimes \C$; for any $\alpha \in \Z_6$ we need $(\alpha, \alpha^{-3} , \alpha^2)$ to act trivially on this particle. We compute,
\begin{equation*}
	(\alpha, \alpha^{-3} , \alpha^2) \cdot e_L^- = \alpha^{-3} \alpha^{-3} e_L^- = \alpha^{-6} e_L^- = e_L^-\ ,
\end{equation*}
since $\alpha$ is a sixth root of unity. In principle, there are 15 other such cases to check, but these can be reduced to just four hypercharge relations that must be satisfied:
\begin{itemize}
	\item for the left-handed quarks, $Y = $ even integer $+1/3$,
	\item for the left-handed leptons, $Y = $ odd integer,
	\item for the right-handed quarks, $Y = $ odd integer $+1/3$, and
	\item for the right-handed leptons, $Y = $ even integer.
\end{itemize}
A glance at table \ref{thesmfermionstable} shows that all of these equalities hold, so our $\SU(5)$ theory has passed its first test. We remark here that not only is $\Z_6$ contained in the kernel of the Standard Model representation, but it is in fact the entire kernel. Hence, one could say that $G_\mrm{SM} / \Z_6$ is the ``true'' gauge group of the Standard Model.

Our next order of business is to find a representation of $\SU(5)$ that extends the Standard Model representation, and there is a beautiful choice that works: the exterior algebra $\Lambda^* \C^5$.  We have to check that pulling back the representation from $\SU(5)$ to $G_\mrm{SM}$ using $\phi$ gives the Standard Model representation $F \oplus \cc{F}$; our strategy will be to use the fact that, being representations of compact Lie groups, both $F \oplus \cc{F}$ and $\Lambda^* \C^5$ are completely reducible, and can be written as the direct sum of irreps; we will then match them up one irrep at a time. We already know what the decomposition of $F \oplus \cc{F}$ into irreps is, so let us look at $\Lambda^* \C^5$. Any element $A \in \SU(5)$ acts as an automorphism of the exterior algebra: $A \cdot (v \wedge w) = Av \wedge Aw$, where $v, w \in \Lambda^* \C^5$. Since we know how $A$ acts on vectors in $\C^5$, and these generate $\Lambda^*\C^5$, this rule is enough to tell us how $A$ acts on all of $\Lambda^* \C^5$. This action respects grades in $\Lambda^* \C^5$, so each exterior power in
\begin{equation*}
	\Lambda^* \C^5 \cong \Lambda^0\C^5 \oplus \Lambda^1\C^5 \oplus \Lambda^2\C^5 \oplus \Lambda^3\C^5 \oplus \Lambda^4\C^5 \oplus \Lambda^5\C^5
\end{equation*}
is a subrepresentation. More than that, they are all irreps of $\SU(5)$, though this is not so easy to see; we refer the reader to \cite[Ch.~15.2]{fulton99} for a proof.

$\Lambda^0 \C^5$ and $\Lambda^5 \C^5$ are both trivial irreps of $G_\mrm{SM}$, and there are exactly two trivial irreps in $F \oplus \cc{F}$, namely $\langle \nu_R \rangle$ and $\langle \cc{\nu}_L \rangle$ (we use the angle brackets to denote the Hilbert space spanned by a vector or vectors). Hence, these irreps must match up; we will select $\Lambda^0 \C^5 \cong \langle \cc{\nu}_L \rangle$ and $\Lambda^5 \C^5 \cong \langle \nu_R \rangle$ for reasons that will be clear in a moment. Consider next the irrep $\Lambda^1 \C^5 \cong \C^5$. The group $G_\mrm{SM}$ acts on $\C^5$ via $\phi$; just by inspection, we see that this action preserves a splitting of $\C^5$ into $\C^2 \oplus \C^3$, with the $\C^2$ part transforming in the hypercharge representation $\C_1$, and the $\C^3$ piece transforming in $\C_{-2/3}$. From table \ref{thesmfermionstable} then, we see that we must have
\begin{align*}
	\Lambda^1 \C^5 &\cong \left( \C_1 \otimes \C^2 \otimes \C \right) \oplus \left( \C_{-2/3} \otimes \C \otimes \C^3 \right)\\
	&\cong \left\langle {\begin{matrix} e_R^+ \\ \cc{\nu}_R \end{matrix} }\right\rangle \oplus \langle d_R \rangle\ ,	
\end{align*}
where we once again used the self-duality of $\C^2$ under $\SU(2)$.

The remainder of the irrep matching is similarly straightforward. The final result is as follows:
\begin{equation}
\label{su5particleassignment}
\begin{aligned}
	\Lambda^0 \C^5 &\cong \langle \cc{\nu}_L \rangle\ ,\qquad &&\Lambda^1 \C^5 \cong \left\langle {\begin{matrix} e_R^+ \\ \cc{\nu}_R \end{matrix} }\right\rangle \oplus \langle d_R \rangle\ ,\\
	\Lambda^2 \C^5 &\cong \langle e_L^+ \rangle \oplus \left\langle {\begin{matrix} u_L \\ d_L \end{matrix} }\right\rangle \oplus \langle \cc{u}_L \rangle\ ,\qquad &&\Lambda^3 \C^5 \cong \langle e_R^- \rangle \oplus \left\langle {\begin{matrix} \cc{d}_R \\ \cc{u}_R \end{matrix} }\right\rangle \oplus \langle u_R \rangle\ ,\\
	\Lambda^4 \C^5 &\cong \left\langle {\begin{matrix} \nu_L \\ e_L^- \end{matrix} }\right\rangle \oplus \langle \cc{d}_L \rangle \ ,\qquad &&\Lambda^5 \C^5 \cong \langle \nu_R \rangle\ .
\end{aligned}
\end{equation}
Hence, $\Lambda^* \C^5 \cong F \oplus \cc{F}$, as desired. Notice that our choice $\Lambda^0 \C^5 \cong \langle \cc{\nu}_L \rangle$ has led to a rather pleasing pattern: the left-handed particles transform in the even grades, while the right handed particles transform in the odd ones. At the level of the $\SU(5)$ theory, this is nice but not essential; for the $\Spin (10)$ theory, this is the only possibility; we will return to this point in section \ref{thespin10gutsubsection}.

We have now shown everything we needed to show: the mapping above defines a linear isomorphism $F \oplus \cc{F} \to \Lambda^* \C^5$ between representations of $G_\mrm{SM}$, i.e.~these representations are the same when we identify $\mrm{S} (\U(2) \times \U(3))$ with $G_\mrm{SM} / \Z_6$ using the isomorphism induced by $\phi$. This can be neatly summarised in a commuting diagram, the main result of this section. 

\begin{thm}
\label{su(5)isagut}
	$\SU(5)$ is a grand unified theory, i.e.~the following square commutes:
	\begin{equation*}
	\begin{tikzcd}[column sep = large, row sep = large]
		G_\mrm{SM} / \Z_6 \ar[r, hook] \ar[d] & \SU(5) \ar[d] \\
		F \oplus \cc{F} \ar[r, "\cong"] & \Lambda^* \C^5
	\end{tikzcd}
	\end{equation*}
\end{thm}

\section{Clifford Algebras}
\label{cliffordalgebrassubsection}

To approach the $\Spin (10)$ grand unified theory, we need to understand Clifford algebras, which are the most natural environment in which to study the Spin groups. Moreover, many of the results that we will obtain will be required to construct $\E_6$ in due course.

Clifford algebras are a generalisation of the complex numbers, quaternions and octonions: indeed, they are sometimes constructed in the literature by adding the required number of square-roots of $-1$ to the algebra of the real numbers. We will not take this route, pursuing the more formal (and standard) treatment found in \cite{adams96} and \cite{fulton99}, for example. In what follows, $V$ is a finite-dimensional vector space over $\K = \R$ or $\C$.

\begin{defn}[Tensor Algebra]
\label{tensoralgebradefn}
	For a non-negative integer $k$, we define the \emph{$k$th tensor power} of $V$ to be the tensor product of $V$ with itself $k$ times:
	\begin{equation*}
		T^k V = V^{\otimes k} = \underbrace{V \otimes V \otimes \cdots \otimes V}_{k \mrm{ times}}\ . 
	\end{equation*}
	By convention, $T^0 V = \K$. Then the tensor algebra is given by
	\begin{equation*}
		T(V) = \bigoplus_{k = 0}^{\infty} T^k V\ .
	\end{equation*}
\end{defn}

We define multiplication as follows: if $v_1 \otimes \cdots \otimes v_p \in V^{\otimes p}$ and $w_1 \otimes \cdots \otimes w_q \in V^{\otimes q}$, then their product is $v_1 \otimes \cdots \otimes v_p \otimes w_1 \otimes \cdots \otimes w_q \in V^{\otimes (p + q)}$. For example, if $V$ has a basis $\{x. y\}$, then $T (V)$ has a basis $\{ 1, x, y, xy, yx, x^2, y^2, \ldots \}$ (the tensor product symbol has been omitted for brevity). In general, $T (V)$ is a free associative algebra.

\begin{defn}[Clifford Algebra]
Let $V$ be endowed with a symmetric bilinear form $\langle \ ,\ \rangle$. Let $J$ denote the two-sided ideal in $T(V)$ generated by the set $\left\{ v \otimes v - \langle v , v \rangle \cdot 1 \mid v \in V \right\}$, and define
\begin{equation*}
	\Cl (V) := T(V) / J\ ;
\end{equation*}
this is the Clifford Algebra over $(V , \langle \cdot , \cdot \rangle)$.
\end{defn}

It is clear that this is equivalent to the characterisation one usually sees, namely, that the Clifford algebra is the associative algebra freely generated by $V$ with relations
\begin{equation}
	vw + wv = -2 \langle v , w \rangle\ .
\label{cliffordalgebrarelationeqn}
\end{equation} 
A word on notation: though we should define a map $V \to \Cl (V)$ denoting the composition $V \hookrightarrow T(V) \to T(V) / J = \Cl (V)$, this is usually omitted in practice, and we write $v$ for an element of $V$ or its image in $\Cl (V)$.

Now, $T(V)$ is a $\Z_+$-graded algebra; let us set
\begin{equation*}
	T_0(V) := \bigoplus_{n\ \mrm{even}} V^{\otimes n}\ ,\ \ \ T_1(V) := \bigoplus_{n\ \mrm{odd}} V^{\otimes n}\ .
\end{equation*}
Then $v \otimes v - \langle v , v \rangle \cdot 1 \in T_0(V)$ and $J = J_0 + J_1$, where $J_i = J \cap T_i(V)$, and
\begin{equation*}
	\Cl (V) = \Cl_0 (V) \oplus \Cl_1 (V)\ ,
\end{equation*}
where $\Cl_i (V) = T_i(V) / J_i$. Hence, $\Cl (V)$ is a $\Z_2$-graded algebra.

\begin{prop}
\label{cliffordalgebradirectsum}
	If $V = V' \oplus V''$ with $\langle v' , v'' \rangle = 0$ for all $v' \in V'$, $v'' \in V''$, then
	\begin{equation*}
		\Cl (V) \cong \Cl (V') \otimes \Cl (V'')
	\end{equation*}
	is its Clifford algebra.
\end{prop}

The proof of this standard result can be found in, for example, \cite[pp.~14--15]{adams96}. Its corollaries bridge the gap between our abstract construction of Clifford algebras, and the motivation of generalising the complex numbers.

\begin{cor}
	If $\dim_\K V = n$ and $\{ e_1 \ldots e_n\}$ is an orthogonal basis for $V$ with $\langle e_i , e_j \rangle = \lambda_j \delta_{ij}$, then $\dim_\K \Cl (V) = 2^n$, and $\left\{ \prod e_j^{i_j} \right\}$ is a basis, where $i_j$ is 0 or 1.
\end{cor}

\begin{example}
	Suppose $V$ is 1-dimensional with basis $\{ e \}$. Then $\Cl (V)$ has basis $\{1 , e\}$, because $T(V)$ has basis $\left\{ 1, e^1, e^2, \ldots \right\}$ and $J$ has a basis $\left\{ e^2 - \langle e, e \rangle \cdot 1 , e^3 - \langle e , e \rangle \cdot e , \ldots \right\}$. Hence, $\Cl (\R) = \C$, as desired. The next result shows that $\Cl(\C) = \mbb{H}$.
\end{example}

\begin{cor}
\label{cliffordbasiscor}
	Again, assume that $\{ e_i \}$ is a basis for $V$ and that $\langle e_i , e_j \rangle = - \delta_{ij}$. Then the products in $\Cl (V)$ are determined by the following relations:
	\begin{itemize}
		\item $e_i^2 = -1$,
		\item $e_i e_j = - e_j e_i$, for $i \neq j$.
	\end{itemize}
\end{cor}

\subsubsection{Structure Maps on Clifford Algebras}

We define structure maps on Clifford algebras, analogously to remark \ref{structuremapsremark}. Consider $\alpha : \Cl(V) \to \Cl (V)$, induced by $-1 : V \to V$; we have that $\alpha |_{\Cl_0 (V)}^{} = + 1$, and $\alpha |_{\Cl_1 (V)}^{} = -1$. Further, define $\beta : T(V) \to T(V)$ by $\beta (v_1 \otimes \cdots \otimes v_n) = v_n \otimes \cdots \otimes v_1$. This is an anti-automorphism, $\beta (xy) = \beta (y) \beta (x)$, and induces $\beta : \Cl (V) \to \Cl(V)$ with $\beta |_{V}^{} = 1$. Finally, $\gamma := \alpha \beta = \beta \alpha : \Cl (V) \to \Cl (V)$ is an anti-automorphism such that $\gamma |_V^{} = -1$.

\begin{example}\leavevmode
	\begin{itemize}
		\item $\Cl(\R)$ has generators $\{ 1, i \}$ with $\alpha (1) = 1$, $\alpha(i) = -i$, $\beta (1) = 1$, $\beta (i) = i$, $\gamma (1) = 1$ and $\gamma (i) = -i$.
		\item	The algebra $\mbb{H} = \Cl (\C)$ has a basis $\{ 1, i, j, ij = k \}$, with the action of $\alpha, \beta, \gamma$ given by
		\begin{center}
		\begingroup
			\renewcommand{\arraystretch}{1.3}
			\begin{tabular}{c | c c c c}
			& 1 & $i$ & $j$ & $k$\\
			\hline
			$\alpha$ & 1 & $-i$ & $-j$ & $k$\\
			$\beta$ & 1 & $i$ & $j$ & $-k$\\
			$\gamma$ & 1 & $-i$ & $-j$ & $-k$
			\end{tabular}
			\endgroup
		\end{center}
	\end{itemize}
\end{example}

Note that for $\Cl (\R) = \C$ or $\Cl (\C) = \mbb{H}$, $\gamma$ is the usual conjugation map.

\subsection{The Spin Groups}

We now posses the machinery to introduce the Spin groups. Let $V = \K^n$ with $\langle\ ,\ \rangle$ the standard inner product, and construct its Clifford algebra $\Cl (V)$ with respect to $- \langle\ ,\ \rangle$. With $\{ e_i \}$ the standard basis for $V$, we have the products $e_r e_s$ for $r < s$ in $\Cl (V)$---there are $n (n-1) /2$ such products---and they span a Lie algebra with Lie bracket $[ e_r e_s , e_t e_u] = e_r e_s e_t e_u - e_t e_u e_r e_s$. By corollary \ref{cliffordbasiscor}, we have that 
\begin{itemize}
	\item if $r, t, s, u$ are different, the bracket vanishes, and also for $r = t, s = u$;
	\item if $r, t, u$ are different, $[ e_r e_u , e_t e_u ] = 2 e_r e_t$.
\end{itemize}
We wish to see this Lie algebra as the Lie algebra of a Lie group.

\begin{defn}[The $\Pin$ Groups]
	$\Pin (V) \subset \Cl (V)$ is the subset of elements $x$ such that
	\begin{itemize}
		\item $x (\gamma x) = (\gamma x) x = 1$, and
		\item the map $\pi x : V \to \Cl (V)$ defined by $(\pi x) v = xv (\beta x) $ maps $V \subset \Cl(V)$ into $V$.
	\end{itemize}
\end{defn}

\begin{example}
	$\Pin (\R) = \{ z \in \Cl(\R) \mid z \cc{z} = 1, z^2 \in \R \} = \{ \pm 1, \pm i \}$.
\end{example}

\begin{prop} The $\Pin$ groups satisfy the following properties.
\label{pingrouppprop}
\begin{itemize}
	\item $\Pin (V)$ is a subgroup of the invertible elements of $\Cl (V)$ and its Lie algebra is the one specified above.
	\item The map $\pi : \Pin (V) \to \OR (V)$, where $\OR (V)$ is the group of $\K$-linear maps $V \to V$ preserving $- \langle \ ,\ \rangle$, is a surjection with $\ker \pi = \{ \pm 1\}$.
	\item The closed subsets $\pi\inv (\det\inv 1 )$ and $\pi\inv (\det\inv -1 )$ of $\Pin (V)$ are in $\Cl_0 (V)$ and $\Cl_1(V)$ respectively, and are connected for $n > 2$.
\end{itemize}
\end{prop}

\begin{proof}
	It is straightforward to check that $\Pin (V)$ is a closed subgroup of invertible elements of $\Cl (V)$. We move on to showing that $\pi (x) \in \OR(V)$:
	\begin{align*}
		\langle (\pi x )v , (\pi x) v \rangle &= ((\pi x) v)^2\\
		&= -xv (\beta x) (\alpha x) v (\gamma x)\\
		&= -x v v \gamma (x)\\
		&= \langle v , v \rangle x \gamma x\\
		&= \langle v, v \rangle\ ,
	\end{align*}
where in the second line we used the fact that $- (\pi x) v = - (\alpha x ) v (\gamma x)$.

It is easy to see that $\pi$ is a homomorphism; let us try to identify its kernel. Suppose that $x \in \ker \pi$. Then $v = x v \beta x$ for all $x \in V$, so $v \alpha x = xv (\beta x) (\alpha x) = xv$; by the lemma below, $x$ must be a scalar. But $x \gamma x = 1$, so $x^2 = 1 \implies x = \pm 1 \implies \ker \pi \subset \{ \pm 1\}$; since the inclusion certainly holds in the other direction, we conclude that $\ker \pi$ is identically $\{ \pm 1\}$.

\begin{lem}
\label{lemmaforpingroupsprop}
	If $\langle\ ,\ \rangle$ on $V$ is non-singular and $x \in \Cl (V)$ is such that $xv = v (\alpha x)$ for all $v \in V$, then $x$ is a scalar.
\end{lem}
\begin{proof}[Proof of Lemma]
	Over $\K = \R$ or $\C$, we can diagonalise $\langle\ ,\ \rangle$ and choose a basis $\{ e_i \}$ such that $\langle e_r , e_s \rangle = \delta_{rs} \lambda_r$, with $\lambda_r \neq 0$. Then any $x \in \Cl (V)$ can be written as $\sum_I \lambda_I \prod_j e_j^{i_j}$, where $\lambda_I \in \K$. Now, $e_s$ is invertible since $e_s e_s = \lambda_s \neq 0$, so that if $x e_s = e_s (\alpha x)$, we have $e_s\inv x e_s = \alpha x$. But
	\begin{equation*}
		e_s\inv \left( \prod_j e_j^{i_j} \right) e_s = \begin{cases}
		(-1)^{\sum i_j} \prod e_j^{i_j} & \T{if } i_s = 0\ ,\\
		(-1)^{-1 + \sum i_j} \prod e_j^{i_j} & \T{if } i_s = 1\ ,
		\end{cases}
	\end{equation*}
	while $\alpha (\prod_j e_j^{i_j}) = (-1)^{\sum i_j} \prod e_j^{i_j}$ in all cases; thus $x e_s = e_s \alpha x$ if and only if $\lambda_I = 0$ whenever $i_s = 1$, i.e.~$\lambda_I \neq 0$ only for $I = (0, 0, \ldots, 0)$, so $x$ is a scalar.
\end{proof}
	
	Next, we proceed to show that $\diff \pi$ from the Lie algebra of $\Pin (V) \cap \Cl_0 (V)$ to $\mf{o} (V)$ is a surjection. The first step is of course to check that $\Pin (V) \cap \Cl_0 (V)$ is a closed subgroup, but this is straightforward. Then, for $r < s$, $(e_r e_s) (e_r e_s) = -e_r e_r e_s e_s = - 1$, so $x = e^{(e_r e_s) t} = \cos t + (e_r e_s) \sin t$ is defined for $\R$ or $\C$. This means that
	\begingroup
	\renewcommand{\arraystretch}{0.8}
	\begin{equation*}
		\pi (x) = \begin{pmatrix}
			1 &&&&&&&&&&\\
			& \ddots &&&&&&&&&\\
			&& 1 &&&&&&&&\\
			&&& \cos 2t &&&& - \sin 2t &&&\\
			&&&& 1 &&&&&&\\
			&&&&& \ddots &&&&&\\
			&&&&&& 1 &&&&\\
			&&& \sin 2t &&&& \cos 2t &&&\\
			&&&&&&&& 1 &&\\
			&&&&&&&&& \ddots &\\
			&&&&&&&&&& 1\\
		\end{pmatrix}
	\end{equation*}
	\endgroup
	since
	\begin{equation*}
		\left( \cos t + (e_r e_s) \sin t \right) e_u \left( \cos t - (e_r e_s) \sin t \right) = \begin{cases}
		e_u & \T{if } u \neq r, s\ ,\\
		\left(\cos 2t + (e_r e_s) \sin 2t \right) e_u & \T{if } u = r, s\ .\end{cases}
	\end{equation*}
	Thus $\pi (x)$ maps $V$ to $V$, so $x \in \Pin (V)$; in fact, $x \in \Pin (V) \cap \Cl_0 (V)$. We see that the $\K$-multiples of $e_r e_s$ lie in $\mf{pin} (V) \cap \mf{cl}_0 (V)$ and map under $\diff \pi$ to the $\K$-multiples of $\begin{psmallmatrix}
	0 &&&&\\
	& \ddots && -2 &\\
	&& 0 &&\\
	& 2 && \ddots &\\
	&&&& 0
	\end{psmallmatrix} = (\pi (x))'|_{t = 0}$, which form a $\K$-basis for $\mf{o} (V)$, the skew-symmetric matrices. Hence, $\diff \pi$ is surjective.
	
	It follows from this that $\diff \pi : \mf{pin} (V) \to \mf{o} (V)$ is also a surjection; it must further be an injection, since $\ker \pi$ is finite. Thus, the Lie algebra of $\Pin (V) \cap \Cl_0 (V)$ is really the Lie algebra $\mf{pin} (V)$, and $\diff \pi : \mf{pin} (V) \to \mf{o} (V)$ is an isomorphism. As a consequence of this, $\{ e_r e_s | r < s \}$ is a $\K$-basis for $\mf{pin} (V)$.
	
	Let us make the transition from Lie algebra to Lie group. Using $\exp$ and $\log$, we find that $\pi$ maps a small neighbourgood of $1 \in \Pin (V) \cap \Cl_0 (V)$ onto the identity component of $\OR (V)$, i.e.~onto $\SO (V)$. But $\{ \pm 1 \}$ is contained in the identity component of $\Pin (V) \cap \Cl_0 (V)$ if $n \geq 2$, since $\cos t + (e_1 e_2) \sin t$ for $t \in [0, \pi]$ is a path from $1$ to $-1$ in $\Pin (V) \cap \Cl_0(V)$. Thus, $\pi\inv (\SO (V)) = \pi\inv (\det\inv(1))$ is connected and contained in $\Cl_0 (V)$. To show that $\pi\inv(\det\inv(-1))$ is connected and complete the proof, it suffices to produce an element in $\pi\inv (\det\inv (-1))$ multiplication by which will send $\pi\inv (\det\inv (1))$ to $\pi\inv (\det \inv (-1))$. The element $e_1$ will do, for one checks that it lies in $\Pin (V)$ and covers the reflection $\mrm{diag} (-1, 1, \ldots, 1)$.
\end{proof}

\begin{defn}[The $\Spin$ Groups]
	We define $\Spin (n)$ as the subgroup $\pi\inv (\det\inv 1) = \Pin (V) \cap \Cl_0(V)$. It comes with a homomorphism $\pi : \Spin (V) \to \SO (V)$.
\end{defn}

\begin{rmk}
\label{maximaltorusofspinn}
	Over $\R$, the maximal torus in $\Spin (m)$, $m = 2n$ or $2n + 1$, is usually taken to consist of the elements $\prod_{r = 1}^n \left( \cos \frac{x_r}{2} + (e_{2r -1} e_{2r}) \sin \frac{x_r}{2} \right)$ in $\Cl(V)$, $x_r \in \R$, which corresponds to
\begin{equation*}
	\begin{pmatrix}
		\cos x_1 & - \sin x_1 &&&&&&\\
		\sin x_1 & \cos x_1 &&&&&&\\
		&& \cos x_2 & - \sin x_2 &&&&\\
		&& \sin x_2 & \cos x_2 &&&&\\
		&&&& \ddots &&&\\
		&&&&& \cos x_n & - \sin x_n &\\
		&&&&& \sin x_n & \cos x_n &\\
		&&&&&&& 1\\
	\end{pmatrix}\ .
\end{equation*}
Here, row $n + 1$ has 1 as the last entry if $m = 2n + 1$, and is empty for $m = 2n$.
\end{rmk}

\subsection{Clifford Modules and Representations}
\label{cliffmodulesandrepssection}

Since $\Spin (V) \subset \Cl_0 (V)$, any $\Cl_0(V)$-module is a representation of $\Spin (V)$, and some extremely important representations of $\Spin (n)$ arise in this way. We study these now.

\begin{prop}
	The algebras $\Cl (V)$ and $\Cl_0(V)$ are semi-simple, so all their representations are completely reducible.
\end{prop}

\begin{proof}
	Let $\{e_i\}$ be the standard basis in $V = \K^m$, and consider $E = \left\{ \pm \prod_{j = 1}^m e_j^{i _j} \mid i_j = 0 \T{ or } 1 \right\}$, a subgroup of order $2^{m + 1}$ of $\Cl (V)$, corresponding to the matrices $\mrm{diag} ( \pm 1, \ldots , \pm 1 )$ of $\OR (m)$. In $\Cl_0 (V)$ consider the subgroup $E_0$ of $2^m$ elements with $\sum_j i_j$ even; we have $E_0 \subset E \subset \Pin (V)$. Let $\nu$ denote $-1 \in \Cl_0 (V)$ when considered an element of $E_0, E$ or $\Pin (V)$; a module over $\Cl (V)$ gives a representation in which $\nu$ acts as $-1$. Conversely, a representation of $E$ in which $\nu$ acts as $-1$ gives a module over $\K (E) / \langle \nu + 1 \rangle = \Cl (V)$, where $\K (E)$ is the group ring over $\K$ of $E$. The representation theory of $\Cl (V)$ can thus be inferred from that of the finite group $E$, and in particular all representations are completely reducible. We argue similarly for $\Cl_0(V)$, replacing $E$ by $E_0$.
\end{proof}

\begin{prop}
\label{diracspinorprop}
	If $\dim V$ is odd, $m = 2n + 1$, say, then $\Cl_0(V)$ has one irrep $\Delta$ of degree $2^n$, affording a representation $\Delta$ of $\Spin (2n + 1)$ with weights $\frac{1}{2} \left( \pm x_1 \pm x_2 \cdots \pm x_n \right)$; there are $2^n$ of these weights. If $\dim V = m = 2n$, then $\Cl_0(V)$ has two irreps $\Delta^+, \Delta^-$ of degree $2^{n -1}$, affording representations $\Delta^+, \Delta^-$ of $\Spin (2n)$ having weights	$\frac{1}{2} \left( \pm x_1 \pm x_2 \cdots \pm x_n \right)$, with an even number of $-$ signs for $\Delta^+$ and an odd number of $-$ signs for $\Delta^-$; there are $2^{n - 1}$ such  weights. If $\K = \C$, these are complex-analytic representations of $\Spin_\C (m)$.
\end{prop}
\begin{proof}
	By Schur's lemma, $\nu$ acts on any irrep as either 1 or -1. The ones in which it acts as 1 are representations of $\E_0 / \langle \nu \rangle$, which is an abelian group of order $2^{m-1}$, so there are exactly $2^{m-1}$ 1-dimensional representations of $E_0$ in which $\nu$ acts as 1. Since the kernel of $E_0 \to E_0 / \langle \nu \rangle$ has exactly two elements, the conjugacy classes in $E_0$ are either one element (if the element is central) or two elements $\pm g$. For $E_0$, the centre is $\{ \pm 1\}$ if $m = 2n + 1$ and $\left\{ \pm 1, \pm \prod_1^{2n} e_i \right\}$ if $m = 2n$; we can see this as follows. If we conjugate $g = \prod_1^m e_j^{i_j}$ with $e_r e_s$ where $i_r = 1$, $i_s = 0$, we change its sign. So if $g$ is in the centre, $g = \pm 1$ or $\pm e_1 e_2 \ldots e_m$. The latter is in the centre only for $m$ even.
	
	Recalling that the isomorphism classes of irreps (over $\C)$ of a finite group are in a $1 : 1$ correspondence with the conjugacy clasees, we see that $E_0$ has one (resp.~two) more irreducible class(es) of representation(s) than $E_0 / \langle \nu \rangle$ if $m = 2n + 1$ (resp.~m = 2n). Let $F \subset E_0$ be the subgroup generated by $e_1 e_2, \ldots , e_{2r-1} e_{2r} , \ldots e_{2n - 1}e_{2n}$. This is an abelian group of order $2^{n + 1}$, so the index of $F$ in $E_0$ is $2^n$ if $m = 2n + 1$, and $2^{n - 1}$ if $m = 2n$.
	
	Choose now a complex 1-dimensional representation $W$ of $F$ in which $\nu$ acts as $-1$ and $e_{2r - 1} e_{2r}$ acts as $i \epsilon_r$, $\epsilon_r = \pm 1$, $i = \sqrt{-1}$. Then the induced representation $\mrm{Ind}^{E_0}_F W$ is a representation of $E_0$ with degree $2^n$ for $m = 2n + 1$ and $2^{n - 1}$ for $m - 2n$, with $\nu$ acting as $-1$. It has a basis

\begingroup
\centering
\renewcommand*{\arraystretch}{2.5}
	\begin{tabular}{l l l l l}
		$\prod\limits_{\substack{i = 1\\ i\ \mrm{odd}}}^{2n + 1} e_i^{j_i}$ & with & $\sum\limits_{1}^{2n + 1} j_i$ even, & if & $m = 2n + 1$, \\
		$\prod\limits_{\substack{i = 1\\ i\ \mrm{odd}}}^{2n - 1} e_i^{j_i}$ & with & $\sum\limits_{1}^{2n - 1} j_i$ even, & if & $m = 2n$.
\end{tabular}

\endgroup
\medskip
When $m = 2n + 1$, there are $2^n$ choices for $W$ (because we have $n$ choices for $\epsilon_r$), $E_0 / F$ permutes them transitively, and by conjugating with $e_{2r} e_{2r + 1}$, we can change the sign of $\epsilon_r$ without changing anything else. Each choice appears once in $\mrm{Ind}^{E_0}_F W$. We thus get a representation $\Delta$ of $E_0$ with character $2^n$ at $-1$, $-2^n$ at $-1$ and $0$ elsewhere. By the orthogonality relations (corollary \ref{cliffordbasiscor}), $\Delta$ is an irrep of $E_0$.

When $m = 2n$, there are $2^n$ choices for $W$, and under $E_0 / F$ they fall into two orbits: those with $\prod \epsilon_r  = 1$, and those with $\prod \epsilon_r = -1$. We can only change the sign of an even number of $\epsilon_r$, since conjugating with $e_{2r} e_{2s}$, $r < s$, changes the sign of both $\epsilon_r$ and $\epsilon_s$. Hence we have one irrep of $E_0$ which as a representation of $F$ contains all the $W$ with $\epsilon := \prod \epsilon_r = 1$, and another containing all the $W$ with $\epsilon = \prod \epsilon_r = -1$.  The character of these representations is $2^{n - 1}$ at $1$, $- 2^{n - 1}$ at $-1$, $i^n \epsilon$ at $\prod_1^{2n} e_i$, $-i^n \epsilon$ at $-\prod_1^{2n} e_i$, and zero elsewhere. This is an irreducible character by the orthogonality relations, so these are inequivalent representations.
\end{proof}

\begin{rmk}
\label{weightsofspinrepsrmk}
	We calculate the weights of $\Delta$ for $\K = \R$ as follows: the element
		\begin{equation*}
			\prod_{r = 1}^n \left( \cos \frac{2 \pi x_r}{2} + (e_{2r - 1} e_{2r}) \sin \frac{2 \pi x_r}{2}  \right) 
		\end{equation*}
		of the maximal torus acts on $\Delta$ with eigenvalues
		\begin{equation*}
			\prod_{r = 1}^n \left( \cos \frac{2 \pi x_r}{2} + (i \epsilon_r) \sin \frac{2 \pi x_r}{2} \right) = \exp \left( 2 \pi i \sum_{ r = 1}^n \frac{1}{2} \epsilon_r x_r \right) 
		\end{equation*}
		and weights $\frac{1}{2}\sum_{r = 1}^n \epsilon_r x_r$, where $\epsilon_r = \pm 1$.
\end{rmk}

\begin{defn}[Spinors]
	For $\K = \C$, the representations $\Delta$, $\Delta^+$ and $\Delta^-$ are called spinor representations of the (complex) Clifford algebra.
\end{defn}

\begin{prop}
\label{selfdualspinrepprop}
	The representation $\Delta$ of $\Spin (2n + 1)$ is self-dual. The representations $\Delta^+, \Delta^-$ of $\Spin (2n)$ are self-dual if $n$ is even and dual to each other if $n$ is odd. 
\end{prop}
\begin{proof}
	We have to prove the isomorphism in the various senses of $M^*$ with $N$, where we write $M$, $N$ for the spinor representations in the proposition statement. Consider that by definition, $E_0$ acts on an $h$ in the dual representation $M^* = \Hom_\C (M, \C)$ as $(gh) (m) = h (g\inv m )$ for $m \in M$; generalising this, we see that $\Cl_0 (\C^m)$ acts on $M^*$ as $(ah) (m) = h ((\gamma a ) m)$. From the discussion in the final two paragraphs of the proof of proposition \ref{diracspinorprop}, it is clear that we have an isomorphism of the representations $N$ and $M^*$ of the finite group $E_0$, which is an isomorphism of the $\Cl_0 (\C^m)$-modules. But $\Spin (m) \subset \Cl_0(\C^m)$, so the isomorphism preserves the action of the elements of $\Spin (m)$, provided this action is defined by $(gh)m = h (g\inv m)$, which is the usual action.
\end{proof}
 
We now in fact have almost everything we need to discuss the $\Spin (10)$ grand unified theory. But before we do so, we end this section by stating a technical result, required to construct $\E_6$ (and also $\mrm{G}_2$, in the appendix). In particular, we need to understand how the representations of the $\Spin$ groups behave under certain inclusions. To this end, we first note that the inclusion $\K^m \hookrightarrow \K^{m + 1}$ induces an inclusion $\Cl (\K^m) \hookrightarrow \Cl (\K^{m + 1})$, so we get $\Spin (m) \hookrightarrow \Spin (m + 1)$ covering the usual map $\SO (m) \hookrightarrow \SO (m + 1)$, $A \mapsto \begin{psmallmatrix}
	A & 0\\
	0 & 1
\end{psmallmatrix}$. We also have, by proposition \ref{cliffordalgebradirectsum}, that $\Cl(\K^p) \otimes \Cl (\K^q) = \Cl (\K^{p + q})$, which gives $\Spin (p) \times \Spin (q) \to \Spin (p + q)$.

\begin{prop} \leavevmode
\label{inclusionofspingroupsinside}
	\begin{enumerate}
		\item Under the inclusions
	\begin{equation*}
		\Spin (2n) \hookrightarrow \Spin (2n + 1) \hookrightarrow \Spin (2n + 2)\ ,
	\end{equation*}
	we have
	\begin{equation*}
		\begin{tikzcd}[row sep = 0.5em]
			& & \ar[dl, mapsto] \Delta^+\\
			\Delta^+ + \Delta^- & \Delta \ar[l, mapsto] & \\
			& & \ar[ul, mapsto] \Delta^-
		\end{tikzcd}
	\end{equation*}
	\item \begingroup
	\renewcommand*{\arraystretch}{1.2}
	\begin{tabular}[t]{c c c c}
	Under &  $\Spin (2r) \times \Spin (2s)$ & $\to$ & $\Spin (2r + 2s)$\\
	& $\Delta^+ \otimes \Delta^+ + \Delta^- \otimes \Delta^-$ & $\mapsfrom$ & $\Delta^+$\\
	& $\Delta^+ \otimes \Delta^- + \Delta^- \otimes \Delta^+$ & $\mapsfrom$ & $\Delta^-$
\end{tabular}
\endgroup
	\end{enumerate}
\end{prop}

The proofs of these inclusions are found in \cite[pp.~23--24]{adams96}. We will henceforth denote, as we have here, the direct sum of vector spaces (and representations) by a simple + instead of an $\oplus$.


\section{The $\Spin (10)$ Extension of $\SU(5)$}
\label{thespin10gutsubsection}

Let us revisit the $\SU(5)$ theory. Viewed from a different light, the core idea behind the embedding $\SU(2) \times \SU(3) \hookrightarrow \mrm{S} (\U (2) \times \U(3))$, which subsequently split each irrep of $\SU(5)$ into an isospin and colour piece (each twisted with hypercharge), can be stated as follows: since the Standard Model representation is $32$-dimensional, each particle or antiparticle in the first generation of fermions can be named by a $5$-bit code. Roughly speaking, these bits are the answers to five binary queries.\footnote{There are subtleties when we answer ``yes'' to both of the first two questions, or ``yes'' to more than one of the last three, but we ignore this problem here; it has no bearing on our argument.}
\begin{itemize}
	\item Is the particle isospin up?
	\item Is it isospin down?
	\item Is it red?
	\item Is it green?
	\item Is it blue?
\end{itemize}
This binary code interpretation of the $\SU(5)$ theory requires the dimension of $F + \cc{F}$ to be $32$, and this raises some questions, as we shall see now.

At the time of writing, there is no direct experimental evidence for the existence of the right-handed neutrino, even though they are extremely desirable theoretically, as they could account for several phenomena that have no explanation within the Standard Model.\footnote{For thorough reviews of the current theoretical and phenomenological status of this elusive particle, see \cite{drewes13, verma15} and references therein.} The right-handed neutrino $\nu_R$ has a direct bearing on our grand unified theories; in particular, it presents a mystery for the $\SU(5)$ theory. $\SU(5)$ does not require us to use the full $32$-dimensional representation $\Lambda^* \C^5$. It works just as well with the smaller representation
\begin{equation*}
	\Lambda^1\C^5 + \Lambda^2\C^5 + \Lambda^3\C^5 + \Lambda^4\C^5\ ,
\end{equation*}
which is less-aesthetically pleasing, and moreover, clearly does not allow for the existence of $\nu_R$. It would be nicer to have a theory that required us to use all of $\Lambda^*\C^5$; better still, if our theory were an extension of $\SU(5)$, our explanation for the arbitrary hypercharges of the Standard Model particles would live on. The $\Spin(10)$ grand unified theory is an attempt at such an extension; \cite{fritzsch75} and \cite{georgi74b} are the original references for the same.

In proposition \ref{diracspinorprop}, we constructed the spinor representations for $\Spin (2n)$, $\Delta^\pm$, each of dimension $2^{n - 1}$. It is perhaps not immediately apparent from the somewhat technical proof of that result, but these irreps are intimately related to $\Lambda^* \C^n$, and we will exploit this fact to forge a path to the $\SU(5)$ theory.

Let $V$ be a complex vector space with $\dim V = 2n$, equipped with the standard inner product $\langle\ ,\ \rangle$. Write $V = W + W'$, where the $W$'s are $n$-dimensional isotropic spaces for $\langle\ ,\ \rangle$.\footnote{Recall that a space is isotropic when the chosen symmetric bilinear form restricts to the zero form on it.} In fact, under $\langle\ ,\ \rangle$, we can simply take $W$ to be spanned by the first $n$ standard basis vectors, and $W'$ by the last $n$.

\begin{prop}
\label{diracspinorrepprop}
	The decomposition $V = W + W'$ determines an isomorphism of algebras, $\Cl (V) \cong \End (\Lambda^* W)$.
\end{prop}

\begin{proof}
	We follow \cite[p.~304]{fulton99}. Mapping $\Cl (V)$ to the algebra $E = \End (\Lambda^* W)$ is the same as defining a linear mapping from $V$ to $E$, satisfying the relation (\ref{cliffordalgebrarelationeqn}). That is, we must construct maps $l : W \to E$ and $l' : W' \to E$ such that $l (w)^2 = 0 = l' (w')^2$ and
	\begin{equation}
	\label{cliffordrelationforformsequationresult}
		l (w) l' (w') + l' (w') l(w) = 2 \langle w, w' \rangle\ ,
	\end{equation}
	for any $w \in W$ and $w' \in W'$. To do this, we will ``deform'' the usual wedge product on the exterior algebra (this is sometimes referred to as \emph{Clifford multiplication of forms}). For each $w \in W$, $w' \in W'$ and $\xi \in \Lambda^* W$, define
	\begin{align*}
		l_w (\xi) &= w \wedge \xi\ ,\\
		l'_{w'} (\xi) &= \iota_{w'} \xi\ ,
	\end{align*}
	where $\iota_{w'} : \Lambda^k W \to \Lambda^{k-1} W$ is the usual contraction by $w'$; on a basis vector it acts as
	\begin{equation*}
		\iota_{w'} (w_1 \wedge \cdots w_k) = \sum_{j = 1}^k (-1)^{j+1} \langle w', w_j \rangle \ w_1 \wedge \cdots \wedge \widehat{w_j} \wedge \cdots \wedge w_k\ .
	\end{equation*}
	It is immediately clear that $l^2$ and $l'^2$ vanish on their domains, and it is a straightforward exercise to check that equation (\ref{cliffordrelationforformsequationresult}) holds. Finally, one confirms that the resulting map from $\Cl (V) \to \End (\Lambda^* W)$ is an isomorphism by computing it on a basis set.	
\end{proof}

The maps $l$ and $l'$ are far more important than they perhaps appear. The first clue is that if we extend them to all of $V = \C^n$, they are in fact adjoint with respect to the inner product induced on $\Lambda^* \C^n$ by $\langle\ ,\ \rangle$, i.e.~for $v \in \C^n$, $p, q \in \Lambda^* \C^n$, $\langle p, l_v q \rangle = \langle l'_v p, q \rangle$. Adjoint operators are the bread and butter of quantum mechanics, so one might ask if these maps have a physical interpretation; indeed, there is one readily available. In the parlance of physics, particles are vectors, so $l_v = v \wedge$ can be said to ``create'' a particle of type $v$ by wedging; analogously $l'_v = \iota_v$ ``destroys'' a particle of type $v$ by contraction. In other words, these maps return, for each $v$, the corresponding \emph{creation} and \emph{annihilation operators}. It is customary to denote the $n$ creation and annihilation operators corresponding to the $n$ basis vectors $e_j$ of $\C^n$ by $a_j^*$ and $a_j$ respectively, and we will do so below.

Consider now the splitting $\Lambda^* W = \Lambda^{\mrm{even}} W + \Lambda^{\mrm{odd}} W$ into the sum of even and odd exterior powers; $\Cl_0(W)$ clearly respects this splitting. Hence, we deduce that there is an isomorphism
\begin{equation*}
	\Cl_0 (V) \cong \End (\Lambda^{\mrm{even}} W) + \End(\Lambda^{\mrm{odd}} W)\ .
\end{equation*}
Restricting now to the case $n = 5$, we conclude that since $\Spin (10) \subset \Cl_0(\C^{10})$, the above Clifford modules, i.e.~the even- and odd-graded powers of the exterior algebra $\Lambda^* \C^5$, are representations of $\Spin (10)$. Moreover, by proposition \ref{diracspinorprop}, they are irreducible. Elements of these two irreps, $\Delta^+$ and $\Delta^-$, are called \emph{left- and right-handed Weyl spinors} respectively, while elements of their direct sum, $\Lambda^* \C^5$, are called \emph{Dirac spinors}.

We are tantalisingly close now to the $\Spin (10)$ grand unified theory; there remains but one question. Does the Dirac spinor representation of $\Spin (10)$ extend the representation of $\SU(5)$ on $\Lambda^* \C^5$? Or more generally, does the Dirac spinor representation of $\Spin(2n)$, which we will call $\rho'$, extend the representation of $\SU(n)$ on $\Lambda^* \C^n$? Recall that this latter representation $\rho : \SU(n) \to \Lambda^* \C^n$ acts as the fundamental representation on $\Lambda^1 \C^n \cong \C^n$ and respects wedge products. The result that we need is answered in the affirmative by the following theorem, which appears in a classic paper by Atiyah, Bott and Shapiro, wherein they also founded the abstract theory of Clifford modules \cite{atiyah64}.

\begin{thm}
\label{atiyahbottshapiro}
	There exists a Lie group homomorphism $\psi$ that makes this triangle commute:
	\begin{equation*}
	\begin{tikzcd}[row sep=large, column sep = large]
		\SU(n) \ar[r, "\psi"] \ar[dr, "\rho"']  & \Spin (2n) \ar[d, "\rho'"]\\
		& \Lambda^* \C^n
	\end{tikzcd}
	\end{equation*}
\end{thm}
\begin{proof}
	We follow the proof as laid out in \cite{baez10}. The connected component of the identity in $\OR(2n)$ is $\SO(2n)$; since $\U(n)$ is connected and $\U(n) \hookrightarrow \OR (2n)$, it follows that there is an inclusion $\SU (n) \hookrightarrow \U(n) \hookrightarrow \SO (2n)$. Passing to Lie algebras, we obtain an inclusion $\su(n) \hookrightarrow \so(2n)$. A homomorphism of Lie algebras gives a homomorphism of the corresponding simply-connected Lie groups, so we now have a map $\psi : \SU (n) \to \Spin (2n)$; we must check that it makes the above triangle commute.
	
	Since all the groups involved are connected, it suffices to check that this diagram
	\begin{equation}
	\label{commutativediagramforliealgebrasabs}
	\begin{tikzcd}[row sep = large, column sep = large]
		\su(n) \ar[r, "\diff \psi"] \ar[dr, "\diff \rho"']  & \so (2n) \ar[d, "\diff \rho'"]\\
		& \Lambda^* \C^n
	\end{tikzcd}
	\end{equation}
commutes. Since the Dirac representation $\diff \rho'$ is defined in terms of creation and annihilation operators, we should try to express $\diff \rho$ in this way. A good choice of basis for $\su (n)$ will be extremely helpful in this regard: we pick the so-called \emph{generalised Gell-Mann matrices}. Let $E_{jk}$ denote the matrix with $1$ in the $jk$th entry, and $0$ everywhere else. Then $\su (n)$ has the basis\smallbreak
\begingroup
\centering
\renewcommand*{\arraystretch}{1.2}
	\begin{tabular}{c l}
		$E_{jk} - E_{kj}$ & for $j < k$,\\
		$i (E_{jk} + E_{kj})$ & for $j > k$, and\\
		$i (E_{jj} - E_{j+1, j+1})$ & for $j = 1, \ldots, n-1$.
	\end{tabular}
	
\endgroup\smallbreak
\noindent These matrices satisfy $E_{jk} (e_l) = \delta_{kl}$, which is in fact how $a_j^* a_k$ acts on $\Lambda^1 \C^n$. So on this space at least, we have the simple relations
\begin{equation}
\label{actionofsu(n)equation}
\begin{aligned}
	\diff \rho (E_{jk} - E_{kj}) &= a_j^* a_k - a_k^* a_j\ ,\\
	\diff \rho (i (E_{jk} + e_{kj})) &= i (a_j^* a_k - a_k^* a_j)\ ,\\
	\diff \rho (i (E_{jj} - E_{j + 1, j + 1})) &= i (a_j^* a_j - a_{j+1}^* a_{j+1}\ .
\end{aligned}
\end{equation}

We claim that these hold on all of $\Lambda^* \C^n$. To see this, first recall that $\rho$ preserves wedge products:
\begin{equation*}
	\rho (x) (v \wedge w) = \rho (x) v \wedge \rho (x) w\ ;
\end{equation*}
differentiating this condition, we see that $\su(n)$ must act as a derivation:
\begin{equation*}
	\diff \rho (X) (v \wedge w) = \diff \rho (X) v \wedge w + v \wedge \diff \rho (X) w\ .
\end{equation*}
Since both the derivative and taking wedge products are linear, derivations on $\Lambda^* \C^n$ are determined by their action on $\Lambda^1 \C^n$; hence, for equations (\ref{actionofsu(n)equation}) to hold on $\Lambda^* \C^n$, it suffices to check that all the operators on the right hand side of the equation are derivations. The annihilation operator is given by contraction, which acts like so on a wedge product:
\begin{equation*}
	a_j (v \wedge w) = a_j (v) \wedge w + (-1)^p v \wedge a_j w\ ,
\end{equation*}
where $p$ is the order of the tensor $v$; this is almost a  derivative, but not quite. On the other hand, the creation operators act in a completely different way:
\begin{equation*}
	a_j^* (v \wedge w) = a_j^* v \wedge w = (-1)^p v \wedge a_j^* w\ ,
\end{equation*}
since $a_j^*$ acts by wedging with $e_j$, and moving this through $v$ introduces $p$ minus signs. This combines with the pervious relation to ensure that $a_j^* a_k$ is a derivation for every combination of $j$ and $k$, as can be checked explicitly. Hence, $\diff \rho$ can be expressed directly as a sum of creation and annihilation operators. Checking now that the diagram \ref{commutativediagramforliealgebrasabs} commutes is straightforward (though tedious). 
\end{proof}

The homomorphism $\psi$ is precisely what allows us to extend the $\SU(5)$ model to $\Spin(10)$, and makes this square commute
\begin{equation*}
	\begin{tikzcd}[row sep = large, column sep = large]
		\SU(5) \ar[r, "\psi"] \ar[d, "\rho"'] & \Spin(10) \ar[d, "\rho'"]\\
		\Lambda^*\C^5 \ar[r, "\cong"] & \Lambda^*\C^5
	\end{tikzcd}
\end{equation*}
From theorem \ref{su(5)isagut} then, we have the final result of this chapter.

\begin{thm}
\label{spin10isagutthm}
	$\Spin(10)$ is a grand unified theory, i.e.~the following diagram commutes:
	\begin{equation*}
		\begin{tikzcd}[column sep = large, row sep = large]
			G_\mrm{SM} / \Z_6 \ar[r, hook] \ar[d] & \Spin(10) \ar[d, "\rho'"]\\
		F \oplus \cc{F} \ar[r, "\cong"] & \Lambda^*\C^5
		\end{tikzcd}
	\end{equation*}
\end{thm}

\chapter{The $\E_6$ Grand Unified Theory}

A grand unified theory based on the exceptional group $\E_6$ first appeared in a 1976 paper by G\"ursey, Ramond and Sikivie \cite{gursey75}. The authors were motivated by the fact that $\E_6$ has as a maximal subgroup $\SU(3) \times \SU(3) \times \SU(3)$: they took these components to be, respectively, the symmetry groups of the left- and right-handed quarks, and the colour group of the quarks, and considered two assignments of this subgroup into a 27 dimensional irrep of $\E_6$. We will not follow their treatment in this chapter, choosing instead to focus on the following ``cascade'' of theories \cite{barbieri80, gursey81, harvey80}:
\begin{equation*}
	\E_6 \to \Spin(10) \to \SU (5) \to G_\mrm{SM}\ .
\end{equation*}

We will first construct $\E_8$ and $\E_6$ in section \ref{theconstructionofe8sec} below. In the process, we will see how the group $\Spin (10) \times \U(1) / \Z_4$ arises naturally as a maximal subgroup of $\E_6$, which will lead us directly into the proof that $\E_6$ extends the Standard Model in section \ref{e6isagutsection}. Thereafter, we will analyse the new fermions that appear in the $\E_6$ theory.
 
\section{The Construction of $\E_8$ and $\E_6$}
\label{theconstructionofe8sec}

We closely follow \cite{adams96} in this section. Our strategy will be the following: to describe an unknown group $G$, it is useful to find a known subgroup of maximal rank $H \subset G$ and to give an account of $G / H$.\footnote{See the construction of $\mrm{G}_2$ in the appendix for a prototypical example.} The main theorem of this section is the following, the proof of which will be in stages.

\begin{thm}
\label{bigliegrouptabletheorem}
There exist Lie groups $G$ with subgroups $H$ as specified in the following table.

\begingroup
\renewcommand{\arraystretch}{1.7}
\begin{center}
	\begin{tabular}{c c c c c c c c}
	\hline
	$G$ & Rank & Dim. & Local type of $H$ & Rank & Dim. & $\mf{g}/\mf{h}$ as $\C$ Rep. & Dim. \\
	\hline
	\hline
	$\E_6$ & 6 & 78 & $\Spin (10) \times \U(1) / \Z_4$ & 6 & 46 & $\Delta^+ \otimes \xi^3 + \Delta^- \otimes \xi^{-3}$ & 32\\
	$\E_8$ & 8 & 248 & $\Spin (16) / \Z_2$ & 8 & 120 & $\Delta^+$ & 128\\
		\hline
	\end{tabular}
\end{center}
\endgroup

\noindent Here, $\xi$ is the fundamental representation of $\U(1)$ on $\C$.
\end{thm}

\begin{rmk} We will see in due course that
	\begin{enumerate}
		\item in $\Spin (10) \times \U(1) / \Z_4$, the $\Z_4$ is generated by $( \prod_1^{10} e_j , i)$.
		\item in $\Spin (16) / \Z_2$, the $\Z_2$ is generated by $\prod_1^{16} e_i$, and
	\end{enumerate}
\end{rmk}

The first column we will fill is that of $\dim \mf{g}/ \mf{h}$, proceeding thereafter to find groups which have the required representations of these dimensions. We begin with the construction of the Lie algebra of $\E_8$.

\subsubsection{The Construction of a Lie Algebra of Type $\E_8$}

For $\E_8$, there is no representation of smaller degree than $\Ad$, so let us use this fact. Take $L + \Delta^+$, where we denote $\mf{spin}(16) = L$, and consider this simultaneously over $\R$ and $\C$; its degree is 120 + $2^7$ = 248, as required.

For a while we can work with $\Spin (2n)$; let us try and define a suitable inner product on its Lie algebra $L$. By proposition \ref{pingrouppprop}, $L \subset \Cl_0 (V)$ has a basis $\{ e_r e_s \mid r < s \}$ and $\Delta^+$ is a representation of $\Spin (2n)$, and hence of $L$ over $\R$, i.e.~for all $a \in L$, $u \in \Delta^+$, we have $[a, u] \in \Delta^+$ satisfying the Jacobi identity, where the multiplication is Clifford multiplication. Assume now that $2n \equiv 0 \mod 8$ and consider $\Delta^+$ as a real representation of $\Spin (2n)$. Choose $(\ ,\ )_{\Delta^+} : \Delta^+ \otimes \Delta^+ \to \R$, a symmetric bilinear non-zero map, invariant under $\Spin (2n)$, i.e. $(gu, gv) = (u, v)$ for $g \in \Spin (2n)$, $u, v \in \Delta^+$; the linearised form of this invariance is $([a,u], v) + (u, [a,v]) = 0$. Since $\Spin (2n)$ is the double cover of $\SO (2n)$, we have $L = \spin (2n) \cong \so (2n)$, which is the space of skew-symmetric matrices;  on matrices the form $(A, B) = \tr AB$ is symmetric, bilinear and real on real matrices. The invariance property for $X \in \GL (2n)$ is $\tr X A X\inv X B X\inv = \tr AB$; if we set $X = \id + tY$ and pass to the limit, we obtain the linearised version: $\tr \left( [Y, A]B + A [Y, B] \right) = 0$; hence $([Y, A], B) + (A, [Y, B]) = 0$. Under this identification, $e_r e_s$ corresponds to the matrix with all entries zero except in positions $(r, s)$ and $(s, r)$ where we have respectively $-2$ and $2$, as we saw in proposition \ref{pingrouppprop}. This gives $(e_r e_s , e_r e_s ) = -8$, so to remove this undesirable factor, we set
\begin{equation*}
	(A, B)_L := -\frac{1}{8} \tr AB
\end{equation*}
so that $(e_r e_s , e_t e_u)_L = \delta_{rt} \delta_{su}$. We now transpose the action $L \otimes \Delta^+ \to \Delta^+$ to get a map $\Delta^+ \otimes \Delta^+ \to L$.

\begin{lem}[]
\label{technicalemmaforactionoflondelta}
	For all $u, v \in \Delta^+$, there is a unique $[u, v] \in L$ such that $(a, [u, v])_L = ([a, u], v)_{\Delta^+}$ for all $a \in L$ and $[u , v]$ is bilinear in $u, v$. Furthermore, if $v, w \in \C \otimes \Delta^+$ are such that
	
	\smallskip
	\begingroup
	\centering
	\renewcommand*{\arraystretch}{1.3}
		\begin{tabular}{l l l}
		$e_{2q - 1} e_{2q} v = iv$ & for all $q$ & \bigg(corresponding to a weight $ \displaystyle \frac{1}{2} \sum_{1}^n x_i$ \bigg),\\
		$e_{2q - 1} e_{2q} w = - iw$ & for all $q$ & \bigg(corresponding to a weight $\displaystyle - \frac{1}{2} \sum_1^n x_i$\bigg),
		\end{tabular}
		
	\endgroup
	\smallskip
	
\noindent and $(v, w) = 1$, then
	\begin{enumerate}
		\item $[v , w] = i \left(e_1 e_2 + e_3 e_4 + \cdots e_{2n -1} e_{2n} \right)$;
		\item $[e_{2q} e_{2r} v, w] = \left( e_{2q - 1} + i e_{2q} \right) \left( e_{2r-1} + i e_{2r} \right)$, $q < r$;
		\item $[e_{2q_1} e_{2q_2} \cdots e_{2q_m} v , w ] = 0$ if $m > 1$ and $q_1 < q_2 < \cdots q_{2m}$.
	\end{enumerate}
\end{lem}
\begin{proof}
	Clearly, $([a, u], v)_{\Delta^+}$ is a linear function of $a$. Since the inner product on $L$ is non-singular, we must have $([a, u] , v)_{\Delta^+} = (a, b)_L$ for some $b = [u, v] \in L$. Since $([a, u] , v)_{\Delta^+}$ is bilinear in $u$ and $v$, so is $b$. We proceed to derive the explicit formulae. All the inner products in the following are over $\Delta^+$.
	\begin{enumerate}
		\item First we have $(e_{2q - 1} e_{2q} v , w) = (iv, w) = i$ and $(e_r e_s v , w) = 0$ if $e_r e_s$ is not one of the basis elements $e_{2q -1} e_{2q}$. Thus $[v, w]$ is paired to $i$ if $a = e_{2q - 1} e_{2q}$ and to $0$ for all the other basis elements.
		\item For simplicity of notation, consider $[e_2 e_4 v, w]$. Then $( e_r e_s e_2 e_4 v, w) = 0$ except when $(r, s) = (1, 3), (1, 4), (2, 3)$ or $(2, 4)$, when we get respectively $1, i, i, -1$. This yields $[e_2 e_4 v, w] = e_1 e_3 + i e_1 e_4 + i e_2 e_3 - e_2 e_4 = (e_1 + i e_2) (e_3 + i e_4)$, as desired. (Notice that $e_2 e_4 v$ has weight $\frac{1}{2} (-x_1 - x_2 + x_3 + \cdots + x_n)$ while $w$ has weight $\frac{1}{2} (-x_1 - \cdots - x_n)$ so that $[e_2 e_4v, w]$ must have weight $-x_1 -x_2$. In fact, $e_1 + i e_2$ has weight $-x_1$ and $e_3 + i e_4$ has weight $-x_2$.)
		\item It suffices to note that for all $e_r e_s, r < s$, we have $(e_r e_s e_{2q_1} e_{2q_2} \cdots e_{2q_{2m}} v, w) = 0$. \qedhere
	\end{enumerate}
\end{proof}

\begin{rmk}
	Note that the map $[\ ,\ ] : \Delta^+ \otimes \Delta^- \to L$ is invariant under $\Spin (2n)$ because everything in the construction is invariant under $\Spin (2n)$. The linearised form of invariance, i.e.~invariance under $L$, is
\begin{equation*}
	[a , [u, v]] = [[a, u], v] + [u, [a, v]]\ ;
\end{equation*}
this is established as follows. It is sufficient to show that for all $b \in L$,
\begin{equation*}
	(-b, [a, [u,v]])_L + (b, [[a, u], v])_L + (b, [u, [a, v]])_L
\end{equation*}
is zero. This expression, using the invariance of $(\ ,\ )_L$ under $L$ and the definition of $[u, v]$, is equivalent to
\begin{equation*}
	([a, b], [u, v])_L + ([b, [a, u]], v)_L + ([b, u], [a, v])_L\ ;
\end{equation*}
the invariance of $(\ ,\ )_{\Delta^+}$ means that this in turn can be written as
\begin{equation*}
	([[a, b], u], v)_{\Delta^+} + ([b, [a, u]], v)_{\Delta^+} - ([a, [b, u]], v)_{\Delta^+} = 0\ ,
\end{equation*}
where we used the properties of the action of $L$ on $\Delta^+$.
\end{rmk}

We now proceed to give $L + \Delta^+$ the inner product with $L$ and $\Delta^+$ orthogonal, and
\begin{equation*}
	(a + u, b + v) = (a , b)_L + (u, v)_{\Delta^+}
\end{equation*}
for all $a, b \in L$ and $u, v \in \Delta^+$. The Lie bracket $[a, u]$ is as in $L$; $[a, u]$ as the action of $L$ on $\Delta^+$ satisfies $[u, a] = - [a, u]$, and $[u, v]$ as in \ref{technicalemmaforactionoflondelta}.

\begin{thm}
	If $2n = 16$, $L + \Delta^+$ becomes a Lie algebra with an invariant inner product.
\end{thm}
\begin{proof}[Proof (Sketch)]
	The inner product is invariant under $L$ by definition, and under $\Delta^+$ by the definition of $[u, v]$. We need to prove anti-commutativity and the Jacobi identity.
	
	Clearly, $[a, b] = -[b, a]$ since $L$ is a Lie algebra, and we define $[u, a]$ to be $- [a, u]$. To see that $[u, v] = - [v, u]$, observe that for all $a \in L$, $u, v \in \Delta^+$, we have
	\begin{align*}
		(a, [u, v] + [v, u])_L &= ([a, u] , v)_L + ([a, v], u)_L\\
		&=([a, u] , v)_L + (u, [a, v])_L\\
		&= 0\ ,
	\end{align*}
	where the penultimate equality follows from the symmetry of $(\ ,\ )_L$, and the last one from the invariance of $(\ ,\ )_L$ under $L$. Since this is true for all $a$, we have $[u, v] + [v, u] = 0$.
	
	For the Jacboi identity, there are several cases that need to be discussed
	\begin{itemize}
		\item Three variables in $L$, none in $\Delta^+$. The identity will hold here since $L$ is a Lie algebra.
		\item Two variables in $L$, one in $\Delta^+$. We have $[[a, b], u] = [a, [b, u]] - [b, [a, u]] = 0$. Thus $[a, [b,u]] + [u, [a,b]] + [b, [u, a]] = 0$.
		\item One variable in $L$, two in $\Delta^+$. Here, by the invariance of the bracket under $L$, $[a, [u, v]] = [[a, u], v] + [u, [a, v]]$, which leads to the Jacobi identity.
		\item All three variables in $\Delta^+$. This is where one needs the fact that $n = 8$.  Reference \cite[pp.~40--42]{adams96} argues this case in full detail, setting down a general procedure for checking identities of this type by using symmetry. \qedhere
	\end{itemize}
\end{proof}

\subsubsection{The Construction of a Lie Group of Type $\E_8$}

Our construction of the simple, connected, compact Lie group with Lie algebra of type $\E_8$ proceeds according to the following steps.
\begin{enumerate}
	\item Take the Lie algebra $L + \Delta^+$ (over $\R$ or $\C$).
	\item Take the group of automorphisms of this Lie algebra; this is closed subgroup of $\GL (L + \Delta^+)$ preserving the Lie bracket.
	\item Take the identity component and call it $\E_8$. (In fact, the result of step 2 is already connected.)
\end{enumerate}

All our constructions are invariant under $\Spin (16)$, over $\R$ or $\C$, so we get a map $\Spin (16) \to \Aut (L + \Delta^+)$, and since $\Spin (16)$ is connected, we get a homomorphism into $\E_8$. To find the kernel, note that $e_1 e_2 \cdots e_{16} \in \Spin (16)$ acts as $i^8 = 1$ on $\Delta^+$. It covers $-\id \in \SO (16)$, so it acts as $-1$ on $\R^{16}$ but it acts as $1$ on $L$. Therefore it acts as 1 on $L + \Delta^+$. This and the identity are the only elements which act as 1 on $L + \Delta^+$, so we get an embedding $\Spin (16)/ \Z_2 \to \E_8$. We now check that $\E_8$ has the required properties.

Let $A$ be a finite dimensional algebra over $\R$ or $\C$ (for example, a Lie algebra) and let $\Aut (A)$ be the group of automorphism of $A$, that is, linear bijections $\alpha : A \to A$ such that $\alpha (ab) = \alpha (a) \alpha (b)$. Then $\Aut (A)$ is a closed subgroup of $\GL (A)$, hence a Lie group.

\begin{defn}
	A linear map $\delta : A \to A$ is a \emph{derivation} if $\delta (ab) = (\delta a) b + a \delta (b)$.
\end{defn}

The commutator $[ \delta, \delta']$ is a derivation if $\delta$ and $\delta'$ are, so the derivations form a Lie algebra, $\mf{der} (A)$. We then have

\begin{lem}
\
	The Lie algebra of $\Aut (A)$, $\mf{aut} (A)$, is the algebra of derivations of $A$.
\end{lem}
\begin{proof}
	First we show that $\mf{aut} (A) \subset \mf{der} (A)$. To see this, take a short curve in $\mf{aut} (A)$, $\alpha_t = 1 + \gamma t$ starting at the identity. Then $\alpha_t (ab) = \alpha_t (a) \alpha_t (b)$, which gives us $\gamma (ab) = \gamma (a) b + a \gamma (b)$, so $\gamma \in \mf{aut} (A)$ is a derivation.
	
	For the reverse inclusion, consider a derivation $\delta : A \to A$; we have (by induction) the standard formula $\delta^n (ab) =\sum_{i + j = n} \binom{n}{i} (\delta^i a) (\delta^j b)$. Define now $\alpha_t : A \to A$ by $\alpha_t = \sum_{n = 0}^{\infty} \frac{t^n}{n!} \delta^n$. Then
	\begin{equation*}
		\alpha_t (ab) = \sum_{i, j} \frac{t^{i + j}}{i!\,j!} (\delta^i a)(\delta^j b) = (\alpha_t a )(\alpha_t b)\ ,
	\end{equation*}
	so $\alpha_t \in \Aut (A)$, and the tangent vector to $\alpha$ is $\delta \in \mf{aut} (A)$.
\end{proof}

The most familiar example of a derivation is the map $\ad_x (y) = [x, y]$ for $x \in A$. It is easy to see that we further have the formula $[\ad_x, \delta] = \ad_{x \delta}$ for $\delta \in \mf{der} (A)$. The next result says that $\ad_x$ is an isomorphism if the Killing form is non-degenerate.


\begin{lem}[Zassenhaus]
\label{zassenhauslemma}
	If $A$ is a Lie algebra with non-degenerate Killing form, then its algebra of derivations is $A$.
\end{lem}
\begin{proof}
	To see that $\ad_x$ is injective, suppose that $\ad_x = [x, y] = 0$ for all $y \in A$. Then $[w, [x, y]] = 0$ for all $y$, i.e.~$[w, [x, \cdot]]$ is the zero function, so $\tr [w, [x, \cdot]] = (w, x)_K = 0$ for all $w$. This implies that $x = 0$, as the Killing form is non-degenerate.
	
	To prove that $\ad_x$ is surjective\footnote{We follow \cite[p.~74]{jacobson79}.}, let us first identify $\ad_x$ with $x$, and so extend the Killing form from elements $x$ to derivations. Now fix a derivation $\delta$; the linear map $x \mapsto \tr (\ad_x) \delta$ is a then a linear mapping of $A$ into $\C$, i.e.~an element of the dual vector space $A^*$; since $(\ ,\ )_K$ is non-degenerate, it follows that there exists an element $d \in A$ such that $(d, x)_K = \tr (\ad_x) \delta$ for all $x \in A$. Let us denote $\partial := \delta - \ad_d$. Then,
	\begin{equation*}
		\tr (\ad_x) \partial = \tr (\ad_x) \delta - \tr \ad_x \ad_d = \tr (\ad_x) \circ \delta - (d, x)_K = 0\ .
	\end{equation*}
	Now consider, for $x, y \in A$,
	\begin{align*}
		(x\partial, y)_K &= \tr \ad_{x \partial} \ad_y\\
		&= \tr [\ad_x, \partial] \ad_y\\
		&= \tr \left( (\ad_x) \partial \ad_y - \partial \ad_x \ad_y \right)\\
		&= \tr ( \partial \ad_y \ad_x - \partial \ad_x \ad_y)\\
		&= \tr \partial [\ad_y, \ad_x]\\
		&= \tr \partial \ad_{yx} = 0\ ,
	\end{align*}
	by the above result. Since $(\ ,\ )_K$ is non-degenerate, this implies that $\partial = 0$; hence $\delta = \ad_d$ for some $d \in A$.
\end{proof}

\begin{lem}[]
\label{killingformone8}
	The Killing form on $L + \Delta^+$ is non-singular: indeed, $(x, y)_K = -240 (x, y)$.
\end{lem}
\begin{proof}
	Both $(\ ,\ )_K$ and $(\ ,\ )$ are invariant under $\Spin (16)$ and $L$ and $\Delta^+$ are irreps of $\Spin (16)$ which are not dual to one another. We can define a map $f : A \to A$, where $A = L + \Delta^+$, by $(x, y)_K = (fx, y)$ for all $y \in A$; $A$ splits as a sum of eigenspaces of $f$ invariant under $\Spin (16)$. Thus $(a + u, b + v)_K = \lambda (a, b) + \mu (u, v)$. But $(\ ,\ )_K$ and $(\ ,\ )$ are invariant under $A$ and $A$ is an irrep of $A$, for the only possible subspaces closed under $L$ are $L$ and $\Delta^+$ and they are not closed under $\Delta^+$. Thus $\lambda = \mu$.
	
To find $\lambda$, we calculate $(e_1 e_2, e_1 e_2)_K = \tr (z \mapsto [e_1 e_2 , [e_1 e_2 , z]]) = \tr (\delta)$,say. Now from lemma \ref{technicalemmaforactionoflondelta}, we have $[e_1 e_2 , e_1 e_2] = 0, [e_1 e_2 , e_1 e_r] = 2e_2 e_r , [e_1 e_2, e_2 e_r] = -2e_1 e_r$ and $[e_1 e_2 , e_r e_s] = 0$ for $2 < r < s$. So $\delta (e_1 e_2) = 0$, $\delta (e_1 e_r) = -4 e_1 e_r$, $\delta (e_2 e_r) = -4 e_2 e_r$, $\delta (e_r e_s) = 0$ and $\tr_L \delta = -112$. On $\Delta^+$ the action is Clifford multiplication, so $[e_1 e_2, [e_1 e_2, u]] = e_1 e_2 e_1 e_2 = -u$, so $\tr_{\Delta^+} (\delta) = -128$ and we have $\tr_A (\delta) = -240$. Hence $\lambda = -240$ since $(e_1 e_2, e_2 e_1) = 1$.
\end{proof}

\begin{cor}
	With the above constructions, $\mf{e}_8 = L + \Delta^+$.
\end{cor}
\begin{proof}
	Immediate from the preceding three lemmas.
\end{proof}

%


\subsubsection{The Construction of $\E_6$}

The map $\Spin (10) \times \Spin (6) \to \Spin (16)$ gives
\begin{equation*}
	\begin{tikzcd}[row sep = large]
		&& \Spin (6) \ar[r] \ar[d] & \Spin (10) \times \Spin (6) \ar[r] & \Spin (16) \ar[r] & E_8\\
		\SU(3) \ar[urr] \ar[r] & \U(3) \ar[r] & \SO(6) &&&
	\end{tikzcd}
\end{equation*}
Consider the centraliser\footnote{The \emph{centraliser} of a subgroup $S \subset G$ is the set of elements in $G$ which commute with $S$.} of the image of this $\SU(3)$ in $\E_8$; we will call its identity component $\E_6$. Note that since the identity component of a topological group is a closed (and normal) subgroup, the group obtained, $\E_6$, is automatically compact. We proceed to check that it has the subgroup of type $\Spin (10) \times \U(1) / \Z_4$, as claimed in theorem \ref{bigliegrouptabletheorem}.

Consider the diagrams
\begin{equation*}
	\Spin (10) \times \Spin (6) \to \Spin (16) \to \E_8\,
\end{equation*}
and
\begin{equation*}
	\begin{tikzcd}[row sep = large]
		(z, g) \ar[d, mapsto] & S^1 \times \SU(3) \ar[d] \ar[rr] && \Spin (6) \ar[d]\\
		(z^2 , g) & S^1 \times \SU(3) \ar[r] & \U(3) \ar[r] & \SO(6)
	\end{tikzcd}
\end{equation*}
where the maps are the obvious ones.\footnote{In proposition \ref{spin6congsu4}, we show that $\Spin (6) \cong \SU(4)$.} We have $\Spin (10) \times S^1 \times \SU(3) \to \Spin (16) \to E_8$ where $\Spin (10) \times S^1$ centralises the (image of) $\SU(3)$ in $\E_8$. It remains to identify the kernel of $\Spin (10) \times S^1 \to E_8$.

The kernel of $\Spin (16) \to \E_8$ is $\Z_2$ generated by $e_1 e_2 \cdots e_{16}$ and the kernel of $\Spin (10) \times \Spin (6) \to \Spin (16)$ is $\Z_2$ generated by $(-1, -1)$. Thus $\ker (\Spin (10) \times \Spin (6) \to \E_8)$ has four elements, generated by $(e_1 \cdots e_{10}, e_{11} \cdots e_{16})$ and $(-1, -1)$, so this kernel is $\Z_4$. To see that all four elements lie in $S^1$, note that $S^1$ is the image of $t \mapsto (\cos t +  (e_{11}e_{12})\sin t)(\cos t + (e_{13}e_{14})\sin t)(\cos t + (e_{15} e_{16})\sin t)$, and for $t = \pi /2, \pi, 3\pi/2$, this image goes through $e_{11}e_{12} \cdots e_{16}$, $-1$, and $-e_{11}e_{12} \cdots e_{16}$ respectively, corresponding to the points $i, -1$ and $-i$ of $S^1$ in $\C$. This completes the check of subgroups mentioned in theorem \ref{bigliegrouptabletheorem}.

\subsubsection{Identification of $\mf{e}_6$}

Call the subgroups $\Spin (10) \times \U(1) / \Z_4$ and $\SU(3)$ $H$ and $K$ respectively.\footnote{Our strategy here is as in the proof of theorem \ref{g2thm}.} We know that $\mf{e}_8 = \mf{spin}(16) + \Delta^+$ as a representation of $\Spin (16)$, and we have a subgroup $H \times K$ mapping into $\Spin (16)$. We wish to determine the centraliser $\E_6$ of $K$ in $\E_8$, so we write $\mf{e}_8$ as a representation of $H \times K$, and take the part fixed under $K$. This is $\mf{e}_6$, and we regard it as a representation of $H$.

When we restrict from $\Spin (16)$ to $\Spin (10) \times \Spin (6)$, $\spin (16)$ restricts to $\spin (10) + \spin (6) + \Lambda^1_{10} \otimes \Lambda^1_6$, where we have introduced the notation $\Lambda^1_{n} := \Lambda^1 (\K^{n})$. We can see this as follows: since $\Spin (10)$ is the double cover of $\SO(n)$, it shares its Lie algebra $\so(n)$, the skew-symmetric $n \times n$ matrices; hence, $\so(16)$ can be decomposed into a block-diagonal $\so(10) + \so(6)$ plus a leftover $10 \times 6$ block, which is isomorphic to $\Lambda^1_{10} \otimes \Lambda^1_{6}$. On the other hand, $\Delta^+$ restricts to $\Delta^+ \otimes \Delta^+ + \Delta^- \otimes \Delta^-$ by proposition \ref{inclusionofspingroupsinside}.

Recall that we defined $\xi$ to be the fundamental representation of $S^1 = \U(1)$. Then, on restricting $\Spin (6)$ to $S^1 \times \SU(3)$ under our map $S^1 \times \SU(3) \to \Spin (6)$, we find that $\spin(6)$ restricts to\footnote{The roots of $\Spin (6)$ which are not roots of $S^1 \times \SU(3)$ are $\pm (x_i + x_j)$.} $\mf{u}(1) + \su(3) + (\xi^4 \otimes \Lambda^2_3 + \xi^{-4} \otimes \Lambda^1_3)$. By looking at weights, we see that $\Delta^+, \Delta^-$ and $\Lambda^1_6$ restrict respectively to $\xi^3 \otimes 1 + \xi\inv \otimes \Lambda^1_3, \xi^{-3} \otimes 1 + \xi \otimes \Lambda^2_3$ and $\xi^2 \otimes \Lambda^1_3 + \xi^{-2} \otimes \Lambda^2_3$. Putting this all together gives
\begin{align*}
	\mf{e}_8 =\ &\spin (10) + \mf{u}(1) + \su(3)\\
	&+ \xi^4 \otimes \Lambda^2_3 + \xi^{-4} \otimes \Lambda^1_3 + \Lambda^1_{10} \otimes \Lambda^1_3 + \Lambda^1_{10} \otimes \xi^{-2} \otimes \Lambda^2_3\\
	&+ \Delta^+ \otimes \xi^3 \otimes 1 + \Delta^+ \otimes \xi\inv \otimes \Lambda^1_3 + \Delta^- \otimes \xi^{-3} \otimes 1 + \Delta^- \otimes \xi \otimes \Lambda^2_3\ , 
\end{align*}
where $\mf{e}_6$ is the part on which $\SU(3)$ acts trivially:
\begin{equation*}
	\mf{e}_6 = \spin (10) + \mf{u}(1) + (\Delta^+ \otimes \xi^3 + \Delta^- \otimes \xi^{-3})\ .
\end{equation*}
This completes the proof of theorem \ref{bigliegrouptabletheorem}.

Finally, note that we have a map $\E_6 \times \SU(3) \to \E_8$, so we may consider $\mf{e}_8$ as a representation of $\E_6 \times \SU(3)$. Regarded thus,
\begin{equation*}
	\mf{e}_8 = \mf{e}_6 + \su(3) + (\xi^{-4} + \Lambda^1_{10} \otimes \xi^2 + \Delta^+ \otimes \xi\inv) \otimes \Lambda^1_3 + (\xi^4 + \Lambda^1_{10} \otimes \xi^{-2} + \Delta^- \otimes \xi) \otimes \Lambda^2_3\ ;
\end{equation*}
this leads to a result of paramount importance for us.

\begin{cor}
\label{e6repscor}
	$\E_6$ has two representations, whose restrictions\footnote{We use the word ``restriction'' here a little loosely. What we mean is that we obtain a representation on $\Spin (10) \times \U(1)$ as it is homomorphic to the subgroup $\Spin (10) \times \U(1)/ \Z_4$ of $\E_6$. We will pick up this point in the next section.} to $\Spin (10) \times \U (1)$ are respectively $\xi^{-4} + \Lambda^1_{10} \otimes \xi^2 + \Delta^+ \otimes \xi\inv$ and $\xi^4 + \Lambda^1_{10} \otimes \xi^{-2} + \Delta^- \otimes \xi$. These are of degree 27 and complex conjugate.
\end{cor}

\section{The $\E_6$ Extension of $\Spin (10)$}
\label{e6isagutsection}

We only need a few more things to be able to write down a theorem for $\E_6$ as a grand unified theory. Firstly, we have not shown that the above 27-dimensional representations of $\E_6$ are irreducible. This is in fact the case, but the proof of this result is unfortunately quite involved, and we will omit it in this paper; the interested reader is referred to \cite[Ch.~11]{adams96}. 

The second thing that we need to check is that these representations, call them $N$ and $\cc{N}$, are unitary. This seems problematic, since we have no direct description of them; the only thing we know is their dimension, and how they reduce to $\Spin (10) \times \U(1) \to \E_6$. Fortunately, there is a way to circumvent this difficulty. We have used several times already that an equivalent charecterisation of a unitary representation $V$ of a group $G$ is the requirement that the action of $G$ on $V$ is an isometry---indeed, this is sometimes taken to be the definition; with this in mind, we have the following handy result, often referred to as \emph{Weyl's unitarian trick}. It requires the notion of a \emph{Haar measure}, which we do not define here; see~\cite[Ch.~1.5]{broecker95}.

\begin{prop}
Any representation $V$ of a compact group $G$ possesses a $G$-invariant inner product.
\end{prop}
\begin{proof}[Proof (Sketch)]
	Let $b : V \times V \to \C$ be any inner product, and define
	\begin{equation*}
		c (u, v) := \int_G b (gu, gv) \diff g\ ,
	\end{equation*}
	where the integral is normalised.	$c : V \times V \to \C$ is then linear in $u$, conjugate linear in $v$, $G$-invariant since the integral is left-invariant, and positive definite since the integral of a positive continuous function is positive.
\end{proof}

$N + \cc{N}$ endowed with this natural $\E_6$-invariant inner product is thus a direct sum of unitary irreps of the compact group $\E_6$. To extend theorem \ref{spin10isagutthm} and prove that $\E_6$ is a grand unified theory however, we need to check something still further: we need a homomorphism $\Delta^+ + \Delta^- \to N + \cc{N}$ as unitary representations of $\Spin (10)$ and $\E_6$ respectively. But since $\Spin (10) \hookrightarrow \Spin (10) \times \U(1) \to \E_6$, and we know how $N + \cc{N}$ restricts to $\Spin (10) \times \U(1)$, it suffices to produce a homomorphism between $\Delta^+ + \Delta^-$ and these restricted representations. But first, let us run through our usual checklist: the restricted representations, as the direct sum of irreps, are clearly irreps. Are they unitary? (i) $\xi$ was defined to be the fundamental representation of the unitary group $\U (1)$, so there is nothing to check here. (ii) From proposition \ref{diracspinorrepprop}, the spinor representations can be seen to be unitary: recall that these are defined via the creation and annihilation operators, which are adjoint; we therefore have $(l^\dagger_v l_v) (\psi) = \iota_v (v \wedge \psi) = \id \psi$, so this is indeed a unitary representation. Lastly, (iii) $\Spin (10) \times \U(1) \ni (s, u)$ acts unitarily on the complex representation $\Lambda^1_{10} \otimes \xi^2 \ni a \otimes p, b \otimes q$:
\begin{align*}
	\langle s \cdot a \otimes u \cdot p , s \cdot b \otimes u \cdot q \rangle_{\Lambda^1_{10} \otimes \xi^2} &= \langle s \cdot a , s \cdot b \rangle_{\Lambda^1_{10}} \langle u \cdot p , u \cdot q \rangle_{\xi^2} \\
	&= \langle a , b \rangle_{\Lambda^1_{10}} \langle p , q \rangle_{\xi^2}\\
	&= \langle a \otimes p , b \otimes q \rangle_{\Lambda^1_{10} \otimes \xi^2}\ ,
\end{align*}
where we simply used the definitions of the tensor product of representations and Hilbert spaces, and the fact that $\Spin (10)$ and $\U(1)$ are each isometries on these representations. So, (i), (ii) and (iii), together with the fact that the tensor product of unitary representations is again unitary, means that we are done, and can write down the following commuting diagram:
\begin{equation*}
	\begin{tikzcd}[row sep = large]
		\Spin (10) \ar[d, "\rho'"'] \ar[r, hook] & \Spin (10) \times \U(1) \ar[d] \ar[r] & \E_6 \ar[d] \\
		\Delta^+ + \Delta^- \ar[r, hook] & (\Delta^+ \otimes \xi\inv) + (\Delta^- \otimes \xi) + \cdots \ar[r] & N + \cc{N}
	\end{tikzcd}
\end{equation*}

We have but one final check. Recall that the homomorphism $\Spin (10) \times \U(1) \to \E_6$ has the kernel $\Z_4$; it is hence incumbent on us to verify, just as we did for the $\SU(5)$ theory, that this kernel acts trivially on every fermion. Explicitly, the four elements of the kernel are
\begin{equation*}
	\{k_1, k_2, k_3, k_4\} := \left\{ (1 , 1) , (\textstyle{\prod_1^{10}} e_j, i) , (-1, -1) , (- \textstyle{\prod_1^{10}} e_j, -i) \right\}\ ; 
\end{equation*}
let us begin with the easiest pieces of the restrictions of $N$ and $\cc{N}$. For $\xi^{-4}$ and $\xi^4$, there is nothing to check for $\Spin (10)$, and since the $\U(1)$ components of the $k_i$'s are precisely the fourth roots of unity, they do in fact act trivially. Now what about the representations $\Delta^+ \otimes \xi\inv$ and $\Delta^- \otimes \xi$? The elements $k_1$ and $k_3$ clearly act trivially. For $k_2$and $k_4$, recall the construction of the spinor representations in proposition \ref{diracspinorprop}: we saw there that $\prod_1^{10} e_j$ acts as $(\prod \epsilon_r) i^5$, where $\prod \epsilon_r = \pm 1$ for $\Delta^\pm$ respectively; coupling this with the fact that the $\U(1)$ components act as $i^{\mp 1}$, means that $k_2$, for example, acts on $\Delta^+$ as $i^5 \otimes i\inv = 1$. The other three cases work out just as easily. The final representations we need to consider are $\Lambda^1_{10} \otimes \xi^{\pm2}$, where $\Spin (10)$ acts by conjugation. Once again, $k_1$ and $k_3$ pose no problem. To tackle $k_2$ and $k_4$, we will need the following

\begin{claim}[]
	$\prod_1^{2n} e_j \in \Spin (2n)$ acts on $v \in \Lambda^1_{2n}$ as $v \mapsto -v$.
\end{claim}
\begin{proof}
	This is a direct computation. Since Clifford multiplication is linear, it suffices to show this for $v = e_k$, for some $1  \leq k \leq 2n$. Consider then
	\begin{align*}
		(e_1 \cdots e_{2n}) \cdot e_k &= (e_1 \cdots e_{2n}) e_k (e_1 \cdots e_{2n} )\inv\\
		&= (-1)^{2n-k} e_1 \cdots \underbrace{e_k e_k}_{-1} \underbrace{e_{k + 1} \cdots e_{2n} (-e_{2n}) \cdots (-e_{k+1})}_{1} \cdots (-e_1)\\ 
		&= (-1)^{2n- k - 1} e_1 \cdots e_{k -1} (-e_k) (-e_{k-1}) \cdots (-e_1)\\
		&= (-1)^{2n- k - 1}(-1)^{k - 1} (-e_k) \underbrace{e_1 \cdots e_{k-1} (-e_{k-1}) \cdots (-e_1)}_{1}\\
		&= -e_k\ . \qedhere
	\end{align*}
\end{proof}

Thus, $k_2$ acts on $\Lambda^1_{10} \otimes \xi^2$ as $(-1) \otimes i^2 = 1$; the other cases are similar. We conclude that the kernel $\Z_4$ does indeed act trivially on all the fermions in the $\E_6$ theory. By theorem \ref{spin10isagutthm} then, we can write down

\begin{thm}
\label{e6isagutthm}
	$\E_6$ is a grand unified theory, i.e.~the following diagram commutes:
	\begin{equation*}
		\begin{tikzcd}[row sep = large, column sep = large]
			G_\mrm{SM} / \Z_6 \ar[r, hook] \ar[d] & \E_6 \ar[d]\\
		F + \cc{F} \ar[r, hook] & N + \cc{N}
		\end{tikzcd}
	\end{equation*}
\end{thm}

\section{The New Fermions}
\label{thenewfermionssubsec}

We are now in uncharted territory: this latest extension of the Standard Model has, for the first time, yielded new particles. We started  with a 32 dimensional representation $F \oplus \cc{F}$ of all the standard model fermions, and found that they fit exactly into the irrep $\Lambda^* \C^5$ of $\SU(5)$; this was in turn shown to be isomorphic to the spinor representation $\Delta^+ + \Delta^-$ of $\Spin (10)$. But now, we have added a significant number of dimensions: $N \oplus \cc{N}$ is $27 + 27 = 54$ dimensional, which means that we have 11 new fermions and antifermions. How can we understand them?

Let us think again about the $\SU(5)$ grand unified theory. There, we matched irreps, one by one, of $\SU(5)$ and $G_\mrm{SM}$; this perhaps obscured the fact that the particles of the $\SU(5)$ theory as such, are \emph{not} characterised by the $G_\mrm{SM}$ charges. Said another way, if we lived in a universe governed by an unbroken $\SU(5)$ theory, there would be no need to think of the Standard Model charges, in the same way that the representation theory of the strong force is remarkably simple because its symmetry group $\SU(3)$ is unbroken at the vacuum. But more often than not, we find ourselves in the opposite situation, and so out of necessity, we characterise particles based on how they transform under the broken symmetry of our vacuum, $G_\mrm{SM}$. In short, to understand these new fermions, we need to think about symmetry breaking, and in particular, we need to understand what charges they carry under $G_\mrm{SM}$.

Happily, one can see a whole lot at the level of representation theory, without venturing into the (complicated) dynamics of symmetry breaking; indeed, without saying so explicitly, we have laid most of the groundwork. Consider again the irrep matching of the $\SU(5)$ theory, equation (\ref{su5particleassignment}). Once we confirmed that that $\Z_6$ kernel of $\phi : G_\mrm{SM} \to \SU (5)$ acted trivially on $F$, matching irreps was precisely the act of understanding how the $\SU(5)$ symmetry broke down to a $G_\mrm{SM}$ theory. In much the same way, theorem \ref{atiyahbottshapiro} was the attempt to see how $\Spin (10)$ broke to $\SU(5)$. In both cases, we had no need for any new charges; with $\E_6$, the situation is different. The proof that $\E_6$ is a grand unified theory rested on the inclusion  $\Spin (10) \hookrightarrow \Spin (10) \times \U(1) \to \E_6$, so a new $\U(1)$ charge seems to be demanded by the mathematics; let us denote it with $\U(1)'$ to differentiate it from the $\U(1)$ of electromagnetism. Then since we have no obvious reason to not do so, let us simply declare that each particle now carries the $\U(1)'$ charge $Q'$ dictated by the superscript of the $\xi$ representation to which it belongs. For example, the particles which live in the representation $1 \otimes \xi^{-4}$ of $\Spin (10) \times \U(1)'$ will carry a charge $Q'$ of $-4$.

The ease with which we able to incorporate a new symmetry into our theory should not obscure the fact that this has huge physical implications: if this $\U(1)'$ symmetry were to remain unbroken at the vacuum, this would imply the presence of a new force (similar to electromagnetism) mediated by a hitherto unobserved massless boson (akin to the photon). No such force has been detected to date, so let us take this into account and posit the following cascade of theories:
\begin{equation*}
	\begin{tikzcd}[row sep = large]
		G_\mrm{SM} / \Z_6 \ar[r, hook] \ar[d] & \SU(5) \times \U(1)' \ar[r, hook] \ar[d] & \Spin (10) \times \U(1)' \ar[r] \ar[d] & \E_6 \ar[d] \\
		F + \cc{F} + \cdots \ar[r, "\cong"] & \Lambda^* \C^5 \otimes Q' + \cdots \ar[r, "\cong"] & \Delta^+ \otimes \xi\inv + \Delta^- \otimes \xi + \cdots \ar[r, hook] & N + \cc{N}
	\end{tikzcd}
\end{equation*}
Following the discussion in the previous paragraphs, we have introduced here the notation $\Lambda^* \C^5 \otimes Q'$ for the extended $\SU(5)$ theory\footnote{Masiero's paper \cite{masiero80} considers some of the phenomenological implications of such an extension to the $\SU(5)$ theory. The article by King \cite{king81} is a general reference for extended $\SU(5)$ theories. Some of these extensions are still viable as grand unified theories \cite{abe17}.} to indicate that the particle representations are now tensored with an additional $\xi^{Q'}$: for the left-handed electron for example, we would write $e_L^- \in \Lambda^4 \C^5 \otimes \xi\inv$, since in the $\Spin(10)$ theory, $e_L^-$ lives in $\Delta^+$, and this is now tensored with $\xi\inv$. In fact, it should be clear that this analysis works for all the Standard Model fermions: we know already which Weyl spinor representation they live in, so it is a simple matter to assign to them a $Q' = \mp 1$, according to whether they are in $\Delta^\pm$, respectively.

The first legitimately new particles appear in $\xi^{\pm 4}$, but these are easy to understand since they do not transform in any group other than $\U(1)'$. Hence, at the level of $\SU(5) \times \U(1)'$, we can simply state that they are the sole elements of the one-dimensional representations $1 \otimes \xi^{\pm4}$; under this assignment, they would be antiparticles of each other, and not interact with any of the Standard Model particles. We will return to this interesting point in section \ref{furtherreadingsec}.

The representations $\Lambda^1_{10} \otimes \xi^{\pm 2}$ will take the most work to sort through. Clearly, the first step is to understand how the $\Spin (10)$ representation $\Lambda^1_{10}$ breaks to $\SU(5)$. We make the following

\begin{claim}
\label{su5assubsetspin10breakingclaim}
	Under $\SU(5) \hookrightarrow \Spin (10)$, the representation $\Lambda^1_{10}$ restricts to $\Lambda^1_5 + \cc{\Lambda^1_5}$, where $\SU(5)$ acts on the former as its fundamental representation, and on the latter as the complex conjugate thereof.
\end{claim}

The proof of this will be in stages. The first thing we will do is to ask whether it suffices to consider the same question at the level of Lie algebras, since in that case we have the explicit embedding (and corresponding eigenvalue problem),
\toenvnum{\tau : \&[-29pt] \su(n)}{\so(2n) \cong \spin(2n)}{\& A_1 + i A_2}{\begin{pmatrix}
	A_1 & A_2\\
	-A_2 & A_1
\end{pmatrix}\ ,}{sunembeddingintoso2neqn}
where $A_1$ and $A_2$ are real $n \times n$ matrices such that $A_1^T = - A_1$, $A_2^T = A_2$, and $\tr A_2 = 0$. The result that we will need comes from a classic query in the theory of representations: can every representation of the Lie algebra of a Lie group be associated with a representation of the group itself, where we moreover require that the differential of the group representation returns the one of the algebra? The answer turns out to be in the affirmative in the case where the Lie group is simply connected \cite[p.~105]{zelobenko73}, which works out nicely for us since both $\SU(n)$ and $\Spin (2n)$ are indeed simply connected: $\Spin (2n)$ is simply connected by virtue of being the universal cover of $\SO(2n)$; for a proof for $\SU(n)$, see \cite{wong67}.

Now as we saw above, $\Spin (2n)$ acts on $\Lambda^1_{2n}$ by conjugation; the differential of this action is the commutator, $X \cdot v = [X , v]$, for $X \in \spin (2n)$, $v \in \Lambda^1_{2n}$. Note that the multiplication on the right hand side of the equation is Clifford multiplication, where we canonically embed both $\Lambda^1_{2n} \cong \C^{2n}$ and
\begin{equation*}
	\so (2n) \cong \spin (2n) = \spann \left\{e_r e_s \in \Cl (\C^{2n}) \mid 1 \leq r < s \leq 2n \right\}
\end{equation*}
into $\Cl (\C^{2n})$. How do we now relate this to our other embedding, (\ref{sunembeddingintoso2neqn})?

\begin{lem}
	For $X \in \so (2n)$ and $v \in \C^{2n}$,
	\begin{equation*}
		X \cdot v = [X , v]\ ,
	\end{equation*}
	where the on the left we have the standard action of $\so (2n)$ on $\C^{2n}$, and on the right, Clifford multiplication.
\end{lem}
\begin{proof}
	As with all linear algebra results, it suffices to check this on a basis. As we have seen, a natural one for $\so(2n)$, the space of skew-symmetric matrices, is
	\begin{equation*}
		\left\{ \epsilon_{rs} = E_{sr} - E_{rs} \mid 1 \leq r < s \leq 2n \right\}\ ,
	\end{equation*}
	where $E_{rs}$ is the $2n \times 2n$ matrix with 1 in the $rs$ entry, and 0 everywhere else. We have, for $e_j$ a standard basis vector of $\C^{2n}$,
	\begin{equation*}
		\epsilon_{rs} e_j = \delta_{jr} e_s - \delta_{js} e_r\ .
	\end{equation*}
	From proposition \ref{pingrouppprop}, the isomorphism between $\spin (2n)$ and $\so (2n)$ is given by $e_r e_s = 2 \epsilon_{rs}$; we thus compute, 
	\begin{align*}
		[\epsilon_{rs} , e_j] &= \frac{1}{2} [e_r e_s , e_j]\\
		&= \frac{1}{2} \left(e_r (-2 \delta_{js} - e_j e_s) - (-2 \delta_{jr} - e_r e_j) e_s \right)\\
		&= \frac{1}{2} \left(2 \delta_{jr} e_s - 2 \delta_{js} e_r \right)\\
		&= \epsilon_{rs} e_j\ . \qedhere		
	\end{align*}
\end{proof}

Therefore, we now have an honest-to-goodness eigenvalue problem for the matrix $\begin{psmallmatrix} A_1 & A_2\\ -A_2 & A_1 \end{psmallmatrix} \in \spin (2n)$. A quick calculation shows that the two $n$-dimensional eigenspaces of this matrix are spanned by $(u,\pm i u)$, $u \in \C^n$:
\begin{equation*}
	\begin{pmatrix}
		A_1 & A_2\\
		-A_2 & A_1
	\end{pmatrix} \begin{pmatrix}
		u\\
		\pm i u
	\end{pmatrix} = \begin{pmatrix}
		(A_1 \pm i A_2) u\\
		(\pm i A_1 \mp A_2) u
	\end{pmatrix} \cong (A_1 \pm i A_2) u\ ,
\end{equation*}
whence we conclude that $\SU(5) \hookrightarrow \Spin (10)$ does indeed act as its fundamental representation (resp.~complex conjugate fundamental representation) on $\Lambda^1_5$ (resp.~$\cc{\Lambda^1_5}$). This completes the proof of claim \ref{su5assubsetspin10breakingclaim}.

We are almost done. The last step we must make is to understand how the $\Lambda^1_5$ and $\cc{\Lambda^1_5}$ of $\SU(5)$ break down to $G_\mrm{SM} / \Z_6$ so we can assign the Standard Model charges to these particles; but this is easy. Indeed, the homomorphism from before,
\toenv{\phi : \&[-29pt] G_\mrm{SM}}{\SU(5)\ ,}{\& (\alpha, g, h)}{\begin{pmatrix}
	\alpha^3 g & 0\\
	0 & \alpha^{-2} h
\end{pmatrix}}
contains all the information that we need. We can just read off the restricted representations (recall that because of how the hypercharge representation $\C_Y$ was defined, we have to divide the exponent of $\alpha$ by 3): 
\begin{equation*}
	\begin{tikzcd}[row sep = large]
		\SU(5) \ar[r] \ar[d] & (\U(1) \times \SU(2)) + (\U(1) \times \SU(3)) \ar[d]\\
		\Lambda^1_5 \ar[r] & (\C_1 \otimes \C_2) + (\C_{-2/3} \otimes \C^3)
	\end{tikzcd}
\end{equation*}
Let us consider a quick example to see how we might catalogue these particles: to the particles in the $\C_1 \otimes \C^2$  doublet of $\U(1) \times \SU(2)$, we would assign as usual the isospins $\pm 1/2$, and they would each carry a hypercharge $Y$ of $1$. In addition, at the level of the $\SU(5)$ theory and beyond, they would carry a new $Q' = \pm 2$, according whether the $\Lambda^1_5$ came from the $\Lambda^1_{10} \otimes \xi^{\pm 2}$. Finally, we note that for the (antiparticle) representation $\cc{\Lambda^1_5}$, one simply passes to the complex conjugate of the representation on the bottom right of the commuting diagram above.

We summarise all of the information in this section in table \ref{particlesinthenofe6}.\footnote{Up to the sign of $Q'$ (he chooses the opposite convention), we have reproduced Table 21 in \cite{slansky81}, for example.} The hypercharge $Y$ therein gives the corresponding Standard Model $\U(1)$ representation $\C_Y$; only the doublets (and hence the particles with non-zero isospin) transform in $\SU(2)$; the $\SU(3)$ representations are written down explicitly. The electromagnetic charge $Q$ can be obtained from the $Y$ and $I_3$ columns via the NNG formula. The $Q'$ column gives the corresponding $\U(1)'$ representation that should be tensored with the representations of $G_\mrm{SM}$, $\SU(5)$ or $\Spin (10)$. Finally, the corresponding table for $\cc{N}$ is easily obtained from this one by passing to the dual representations and taking the opposite charges throughout.
\begingroup
\renewcommand{\arraystretch}{1.5}
\begin{table}[t]
\centering
\caption{Particles in the Representation $N$ of $\E_6$}
\begin{tabular}{c c c c c c c}
		\hline
		\textbf{Symbol} & $Y$ & $I_3$ & $\SU(3)$ \textbf{Rep.} & $Q'$ & $\SU(5)$ \textbf{Rep.} & $\Spin (10)$ \textbf{Rep.}\\[0.5em]
		\hline \hline \\[-1em]
		$\cc{\nu}_L$ & $0$ & $0$ & $\C$ & $-1$ & $\Lambda^0_5$ & $\Delta^+$\\
		$e_L^+$ & $2$ & $0$ & $\C$ & $-1$ & $\Lambda^2_5$ & $\Delta^+$\\
		$\begin{pmatrix} u_L \\ d_L \end{pmatrix}$ & $1/3$ & $\pm 1/2$ & $\C^3$ & $-1$ & $\Lambda^2_5$ & $\Delta^+$\\
		$\cc{u}_L$ & $-4/3$ & 0 & $\cc{\C^3}$ & $-1$ & $\Lambda^2_5$ & $\Delta^+$\\
		$\begin{pmatrix} \nu_L \\ e_L^- \end{pmatrix}$ & $-1$ & $\pm 1/2$ & $\C$ &  $-1$ & $\Lambda^4_5$ & $\Delta^+$\\
		$\cc{d}_L$ & $2/3$ & $0$ & $\cc{\C^3}$ & $-1$ & $\Lambda^4_5$ & $\Delta^+$\\ \\
		$\begin{pmatrix} ? \\ ? \end{pmatrix}$ & $1$ & $\pm 1/2$ & $\C$ & $2$ & $\Lambda^1_5$ & $\Lambda^1_{10}$\\
		$?$ & $-2/3$ & $0$ & $\C^3$ & $2$ & $\Lambda^1_5$ & $\Lambda^1_{10}$\\
		$\begin{pmatrix} ? \\ ? \end{pmatrix}$ & $-1$ & $\pm 1/2$ & $\C$ & $2$ & $\cc{\Lambda^1_5}$ & $\Lambda^1_{10}$\\
		$?$ & $2/3$ & $0$ & $\cc{\C^3}$ & $2$ & $\cc{\Lambda^1_5}$ & $\Lambda^1_{10}$\\ \\
		$?$ & $0$ & $0$ & $\C$ & $-4$ & $\C$ & $\C$\\[1em]
		\hline
	\end{tabular}
\label{particlesinthenofe6}
\end{table}
\endgroup

\chapter{Aspects of Phenomenology}

\label{aspectsofphenochapter}

As mathematically interesting as grand unified theories are, they are ultimately statements about the real world. So in this section, we pose the following question: at the level of group (representation) theory, what can we say about the phenomenology of these theories? Clearly, a healthy amount of physics is required to motivate and supplement any such discussion, but the aim is to stay as close to mathematics as possible; the relevant physics is introduced where necessary.

In the first section, we discuss a proper prediction of grand unified theories, the weak mixing angle, which has a simple closed formula in terms of the eigenvalues of the intertwining operators $\hat{I}_3$ and $Q$. Following this, we will discuss anomalies, which are not so much a phenomenological prediction as they are a basic physical requirement on unification groups. They have a rather nice interpretation in terms of a certain Casimir operator on the Lie algebras of said groups, so this issue is completely reduced to a mathematical property that we can understand fairly easily, given the machinery we have already built. Finally, section \ref{furtherreadingsec} functions as something of a survey section, where we consider other expected signatures of the $\E_6$ theory, and discuss its outlook.

\section{The Weak Mixing Angle}
\label{weakmixinganglesec}

One of the unambiguous predictions of grand unified theories is the weak mixing angle or Weinberg angle, which we have already encountered in section \ref{theewinteractionsec}. Recall that equation (\ref{photozbosonthetaw}) offered a rather geometric interpretation of this angle, as the parameter that characterised the rotation of the $W^0 - B$ boson plane after symmetry breaking; it can also be written in terms of the gauge couplings $g_2$ and $g_1$, of the $\SU(2)$ and $\U(1)$ groups of the electroweak theory respectively, as
\begin{equation}
	\sin^2{\thetaw} = \frac{g_2^2}{g_1^2 + g_2^2}\ .
\label{weinbergangleeqn}
\end{equation}

In 1974, Georgi, Quinn, and Weinberg derived a formula for the weak mixing angle $\thetaw$ in grand unified theories \cite{georgi74c}. The only assumption that they needed in the proof thereof was that the $\U(1) \times \SU(2)$ group of the electroweak theory is embedded in the grand unification group $G$ in such a way that the NNG formula still holds. We have assumed this throughout, so this theorem is applicable to all the grand unified theories we have analysed; let us hence state and prove their result.

\begin{thm}
	Let $R$ be some fermion representation of the unification group $G$. Then the weak mixing angle is given by
	\begin{equation*}
		\sin^2{\thetaw} = \frac{\displaystyle\sum_{\mrm{fermions}} I_3^2}{\displaystyle\sum_{\mrm{fermions}} Q^2}\ .
	\end{equation*}
\end{thm}
\begin{proof}
	We follow the lecture notes of Bjorken \cite{bjorken77}. The $W^0 = 2 \hat{I}_3$ and $B$ bosons, appropriate to the broken $\SU(2)$ and $\U(1)$ theory, are gauge bosons of the full gauge group $G$; the coupling of $W^0$ to any fermion is proportional to $I_3$, and the coupling of the $B$ boson is proportional to the hypercharge $Y$. Because $W^0$ and $B$ are both gauge particles for the group $G$, we must have, for any representation of $G$,
	\begin{equation}
	\label{isospinhyperchargerelation}
		\sum_\mrm{fermions} I_3^2 = \sum_\mrm{fermions} Y^2\ ,
	\end{equation}
	since there is a symmetry operation of the group that can transform $W^0$ into $B$, but that transforms the representation $R$ into itself.
	
	Completing the proof is now a matter of simple algebra. From equation (\ref{photozbosonthetaw}), we must have that the electric charge is given by
	\begin{equation*}
		Q \propto \left(Y \cos{\thetaw} + I_3 \sin{\thetaw} \right)\ ; 
	\end{equation*}
	in order to have the difference of $Q$ between two members of the same isospin doublet be $\pm 1$, we must set the constant of proportionality to $(\sin{\thetaw})\inv$, i.e.
	\begin{equation*}
		Q = I_3 + \cot{\thetaw} Y\ .
	\end{equation*}
	We now square this equation, and sum over all fermions. The cross term $\sum I_3 Y$ vanishes, because the only non-zero contributions to this sum come from isospin doublets, and for each doublet this term is zero (since $Y$ is constant on a doublet, while the $I_3$'s come with opposite signs). We hence obtain
	\begin{equation*}
		\sum_\mrm{fermions} Q^2 = \sum_\mrm{fermions} I_3^2 + \cot^2{\thetaw} \sum_\mrm{fermions} Y^2\ . 
	\end{equation*}
	Utilising equation (\ref{isospinhyperchargerelation}) above, we obtain the formula stated in the theorem.
\end{proof}

In the same paper, Georgi et al.~immediately applied this result to the $\SU(5)$ theory, leading to the famous prediction $\sin^2_{\SU(5)} \thetaw = 3/8$. It should be clear that since the $\Spin (10)$ theory introduces no new fermions, the prediction for the Weinberg angle is same as for $\SU(5)$. In our $\E_6$ theory however, we do have new fermions, so we should see a different value; indeed, on consulting table \ref{particlesinthenofe6} and doing the necessary arithmetic, we obtain the following:
\begin{equation*}
	\sin^2_{\E_6} \thetaw = \frac{9}{20} = 0.45\ .
\end{equation*}
As far as representation theory goes, this is all we can say. But it is too tempting to not compare such a definite phenomenological prediction with the real world; unfortunately, the comparison is none too comforting: one standard estimate \cite{mohr12} for the weak mixing angle is $\sin^2 \thetaw = 0.2223$. Is there a way to fix this massive discrepancy?

The most plausible answer comes from \emph{renormalisation theory}, a catch-all term for techniques used to deal with the infinities that plague quantum field theory. We have neither the desire nor the pages to get into any details here\footnote{We point the reader to \cite{peskin95, schwartz14, weinberg05} or any standard quantum field theory reference.}, but we would like to at least state a result from Marciano \cite{marciano79} which succinctly accounts for renormalisation effects on the value of $\sin^2{\thetaw}$ in grand unified theories. What follows is hence necessarily sketchy; the reader is encouraged to consult the original paper for an excellent discussion. The main assumption that he needed in its derivation goes back to an earlier paper of Weinberg's \cite{weinberg72}: all gauge bosons in the grand unified theory must have large masses (on the order of some superheavy $M_S$, say) compared with the $W^\pm$ and the $Z$, (the order of which we denote by $M_W$) and also compared with the standard model fermions in the theory (also on the order $M_W$).\footnote{He leaves open the possibility that there might be exotic fermions in the theory with masses on the order of $M_S$. As we will see in section \ref{furtherreadingsec}, this is the case with $\E_6$. See also \cite{robinett82}.} The motivation for this was mostly phenomenological: effects mediated by these gauge bosons had eluded detection thus far.\footnote{Even today, the lower bounds on the masses of grand unified theory gauge bosons are at least two orders of magnitude larger than the known masses of the $W$ and $Z$ bosons \cite{patrignani16}.} Of course, there was a further technical aspect to this assumption, but since it involes some quantum field theory, we relegate it to a footnote dedicated to the interested reader.\footnote{The argument, as seen in \cite{georgi74c}, runs as follows. The gauge couplings---of the grand unified symmetry group $G$, and the Standard Model subgroups $\U(1)$, $\SU(2)$ and $\SU(3)$---are functions of the momentum scale which we denote by $\mu$; in particular, equation (\ref{weinbergangleeqn}) only holds when $\mu$ is much larger than the superheavy boson masses, where the breaking of $G$ may be neglected. However, the observed values of the gauge couplings refer to much smaller values of $\mu$, of the order of the $W^\pm$ and $Z$ masses, or even smaller. The problem is therefore to bridge the gap between superlarge values of $\mu$, where $G$ imposes relations among the gauge couplings, and ordinary values of $\mu$, where the gauge couplings are observed. In order to deal with this, Georgi and collaborators employed a theorem from Appelquist et al.~\cite{appelquist73}, which proved that all matrix elements involving particles with masses much less than the superheavy scale could be calculated in an effective renormalisable theory. In this case, one could simply consider the original theory with all the superheavy particles omitted (but with coupling constants that could depend on the superheavy masses). All other effects of the superheavy particles are suppressed by factors of an ordinary mass divided by a superheavy mass.} In any case, Marciano's result is written as\footnote{We note that this formula is specifically for theories with $\sin^2{\thetaw^0} \neq 3/8$.}
\begin{multline*}
	\sin^2{\thetaw} (M_W) = \sin^2{\thetaw^0} \Bigg\{ 1 - \frac{\alpha (M_w)}{2\pi} \bigg[ \frac{22}{3} \cot^2{\thetaw^0} - \frac{1}{6} N_H \left( \frac{1}{\sin^2{\thetaw^0}} - 2 \right) \\
	- \frac{2}{3} N_f \left( \frac{1}{\sin^2{\thetaw^0}} - \frac{8}{3} \right) \bigg] \log{\frac{M_S}{M_W}} \Bigg\}\ .
\end{multline*}
From left to right, here are the quantities we have not yet defined. The superscript $0$ in $\sin^2{\thetaw^0}$ simply indicates that this is the theortical value of the weak angle given by the theory, i.e.~$9/20$ in the case of $\E_6$. $\alpha (M_W)$ is the \emph{fine structure constant}; it depends on the mass scale because it is defined through the electric charge as $e^2/4 \pi$, and $e$ has a mass scale dependence. Marciano provides the estimate $\alpha (M_W) \approx 1/128.5$; we will use the same. Next, the term $N_H$ is the number of complex Higgs doublets in the theory; we set $N_H = 1$, the minimum value. The quantity $N_f$ is the number of fermion flavours: for the Standard Model (and $\SU(5)$ and $\Spin (10)$), this equals $6$, as seen in table \ref{quarksbygenstable}; similarly, from table \ref{particlesinthenofe6}, we see that we have $N_f = 4 \times 3 = 12$ for $\E_6$, since we add two new flavours (one each in $N$ and $\cc{N}$) per generation. Plugging all this in, the above formula simplifies to
\begin{equation*}
	\sin^2_{\E_6} \thetaw (M_W) = \frac{9}{20} \left( 1 - 0.015 \log{\frac{M_S}{M_W}} \right)\ .
\end{equation*}

So for example, if we take the measured values $\sin^2 \thetaw = 0.2223$ and $M_W = 80.385$, we see that the superheavy mass scale for the $\E_6$ theory is of the order $M_S = 3.592 \times 10^{16}\,\mrm{GeV}$. We caution that the value obtained from the formula is is quite sensitive to changes in the value of $\sin^2 \thetaw$ because of the logarithm; it decreases by about $50 \%$ for each increase of $0.005$ in $\sin^2 \thetaw$. One final remark is that we fixed the value of $N_H = 1$ (keeping with Marciano) since in general, Higgs scalars are often considered the ugliest features of gauge theories, and one would prefer to have as few of them as possible\footnote{See section \ref{furtherreadingsec} for references for the Higgs mechanism in $\E_6$.}; if this restriction is relaxed, there is some wiggle room in the above formulae to increase the value of $\sin^2{\thetaw}$ by increasing the value of $N_H$, and this in turn obviously has a direct bearing on $M_S$; table II in \cite{marciano79} estimates the size of this effect.

\section{Anomaly Cancellation}
\label{anomalycancellationsec}

In section \ref{symmetriessubsec}, we discussed Lagrangian symmetries, and saw their paramount importance; the transformation laws that we considered there were indeed the foundation for everything that came after. We return to this theme now, but with a different question as our starting point: which classical symmetries of the Lagrangian are elevated to \emph{quantum} symmetries?

The business of quantising a classical Lagrangian is a messy one. By way of illustration, consider the simplest case: given a Lagrangian $\mc{L}$ of a (real) scalar field $\phi$, one defines the \emph{generating functional} as
\begin{equation*}
	Z[J] := \int \Diff \phi \exp \left[i \int \diff^4 x \left( \mc{L} + J \phi \right) \right]\ ,
\end{equation*}
where $J \phi$ suggestively denotes a \emph{source} term, akin to electromagnetism.\footnote{See \cite[\S\ 28--30]{landau95} for a cogent presentation of the same.} The measure of integration, $\Diff \phi$, represents an integration over all possible field configurations.\footnote{It is common knowledge that foundational questions about the mathematical validity of this definition, and about the path integral formalism in general, remain. See \cite{albeverio11} for a mathematical overview, and \cite{helling12} for a physicist's take on the same.} We can now define the \emph{effective action}:
\begin{equation*}
	\Gamma [\phi_\mrm{cl}] = W[J] - \int \diff \tau \diff x J \phi_\mrm{cl}\ ,
\end{equation*}
where $W[J]$ is defined implicitly via $Z[J] = e^{-W[J]}$, $\phi_\mrm{cl}$ is the functional derivative\footnote{See \cite[Chs.~1.1.2, 1.3.3]{nakahara03}} $\delta W[J] / \delta J$, and $\tau = i t$ is the so-called \emph{Wick-rotated} time. Now, the effective action is given its name for obvious reasons: just as in classical mechanics, where the equations of motion are derived from the principle of stationary action, the equations of motion for the vacuum expectation values of quantum fields can be derived from the requirement that the effective action be stationary.\footnote{For example, if our Lagrangian includes a potential $V(\phi)$, at a low temperatures, the quantum field $\phi$ will not settle in a local minimum of $V(\phi)$ as in the classical case, but rather in a local minimum of the effective potential.} The details, and examples of such calculations can be found in \cite[Ch.~9]{peskin95}, or any other standard reference on quantum field theory; what we have here should suffice to motivate an effective quantum action for a classical Lagrangian.

The first anomaly that we will consider is the \emph{chiral anomaly}; to introduce the same, we will need the \emph{Dirac equation}. We have encountered in some detail already the Dirac spinors in chapter \ref{spin10gutchapter}, as elements of certain irreps of the $\Spin$ groups. This is a description free of dynamics, and therefore, far removed from Dirac's original conception of these particles.\footnote{In the 1928 paper \cite{dirac28} presenting his equation for the first time, he begins by asking ``why Nature should have chosen this particular model for the electron, instead of being satisfied with the point charge;'' his remarkable solution to this quandary was a theory that, for the first time, fully accounted for special relativity in the context of quantum mechanics. For a captivating account of the history and a lucid derivation of the equation, see \cite[Ch.~1.1]{weinberg05}.} The equation describes all spin-$1/2$ massive particles for which parity is a symmetry, i.e.~the leptons. In symbols, for a field $\psi$, which we take to be massless, and a gauge boson $A_\mu = \sum_\alpha \tensor{A}{_\mu^\alpha} T_\alpha$, written in a basis of generators\footnote{We have already seen in section \ref{symmetriessubsec} that gauge bosons are Lie-algebra valued 1-forms; hence, they can be expressed locally in a basis of (anti-Hermitian) generators $\{ T_\alpha \}$ of the Lie algebra $\mf{g}$. For more details, c.f~\cite[Ch.~4.6]{fuchs03}.} of the compact semi-simple symmetry group $G$, the theory is described by the Lagrangian
\begin{equation}
\label{diraceqnlagrangianeqn}
	\mc{L} = \sum_\mu i \cc{\psi} \gamma^\mu (\partial_\mu - A_\mu) \psi\ .
\end{equation}
The $\gamma^\mu$, $0 \leq \mu \leq 3$, are the \emph{gamma matrices}, which form a basis for the Clifford algebra $\Cl_{1,3} (\R)$; the subscript 1,3 denotes that instead of relation (\ref{cliffordalgebrarelationeqn}), we use $\{\gamma^\mu, \gamma^\nu \} = 2 \eta^{\mu \nu}$, where $\eta^{\mu\nu}$ is the Minkowski metric with the (physicists') signature $(+ - - -)$. As in section \ref{symmetriessubsec}, the Lagrangian is invariant under the local gauge transformation $\psi \mapsto g \psi$, $A_\mu \mapsto g A_\mu g\inv + \partial_\mu g g\inv$, but there is now an additional \emph{global} symmetry $\psi \mapsto e^{i \alpha \gamma^5} \psi$, where $\alpha \in \R$, and $\gamma^5 := i \gamma^0 \gamma^1 \gamma^2 \gamma^3 = \begin{psmallmatrix}
	\id_2 & 0\\
	0 & -\id_2
\end{psmallmatrix}$.\footnote{We note here that adding a mass term $m \psi \cc{\psi}$ to the Lagrangian spoils this symmetry, which is why we restrict ourselves to the massless situation. One can show that only such massless fermions contribute to anomalies anyway; see \cite[\S\ 7.2]{bilal08}.} This symmetry is chiral: in the standard (Dirac) basis, the left- and right-handed Weyl components of Dirac 4-spinor correspond to the first two and second two components of the 4-vector respectively; the effect of the $e^{i \alpha \gamma^5}$ is then to rotate these Weyl spinors in opposite directions, by the angle $\alpha$.

As with any other Lagrangian symmetry, the chiral symmetry corresponds to a current, which in this case can be shown to be
\begin{equation*}
	j^\mu_5 := \cc{\psi} \gamma^\mu \gamma_5 \psi\ .
\end{equation*}
Recall that we \emph{need} our theory of leptons to be chiral. The question of the hour is therefore the following: does this classically conserved quantity (i.e.~$\sum_\mu \partial_\mu j^\mu_5 = 0$) stay conserved when we pass to the quantised Dirac Lagrangian? The answer turns out to be \emph{no}. Unfortunately, deriving this result is a nuanced, technical calculation in quantum field theory, far outside the scope of this paper; we list some references in a footnote.\footnote{The lectures of Bilal \cite{bilal08} are specifically on this topic; Schwartz's book \cite[Ch.~30]{schwartz14} also has a detailed treatment; Nakahara \cite[Ch.~13]{nakahara03} derives the same result from a geometric point of view, making the connection to Atiyah-Singer-Index theory; the book of Nash \cite{nash91} goes even deeper into the mathematics.} The final result of this computation is stated as
\begin{align*}
	\sum_\mu \partial_\mu j^\mu_5 &= \sum_{\kappa, \lambda, \mu, \nu} \frac{1}{16 \pi^2} \epsilon^{\kappa \lambda \mu \nu} \tr F_{\kappa \lambda} F_{\mu \nu}\\
	&= \frac{1}{4 \pi^2} \tr \left[ \sum_{\kappa, \lambda, \mu, \nu} \epsilon^{\kappa \lambda \mu \nu} \partial_k \left( A_\lambda \partial_\mu A_\nu + \frac{2}{3} A_\lambda A_\mu A_\nu \right) \right]\ .
\end{align*}
This is not particularly illuminating as it stands. One can show however \cite[\S\ 4]{bilal08} that the right hand side can be rewritten such that it contains the term (recall that the $A_\mu$'s are written in terms of the group generators $T_\alpha$)
\begin{equation}
\label{anomalyvanisheqn}
	\tr \left( T_a \{ T_b , T_c\} \right)_\mrm{L}^{} - \tr \left( T_a \{ T_b , T_c \} \right)_\mrm{R}^{}\ ,
\end{equation}
where the subscript $\mrm{L}$ (resp.~$\mrm{R}$) denotes the representation of the left-handed (resp.~right-handed) fermions under consideration; our theory is said to be ``anomaly-free'' if this quantity vanishes.

References \cite[\S\ 7.3]{bilal08} and \cite[Ch.~30.4]{schwartz14} show how the explicit check proceeds in the case of the Standard Model gauge group; it is a surprisingly tame affair. In the case of the $\U(1)-\U(1)-\U(1)$ anomaly for instance, equation (\ref{anomalyvanisheqn}) simply reduces to $\sum_{\mrm{L}}^{} Y^3 - \sum_{\mrm{R}}^{} Y^3$, and a glance at table \ref{thesmfermionstable} will confirm that this indeed vanishes. The checks for the other subgroups of $G_\mrm{SM}$ are just as straightforward.\footnote{It is not necessary to carry out this check for all possible triples that can be made from $\U(1)$, $\SU(2)$ and $\SU(3)$; cf.~\cite[Ch.~30.4]{schwartz14}.}

Closer to our purposes, the vanishing of equation (\ref{anomalyvanisheqn}) is also a requirement on the groups used for grand unified theories; what can we say about these? Georgi and Glashow showed in \cite{georgi72} that if a group has only real (or pseudoreal) representations, it is automatically anomaly-free---as far as simple Lie groups go, this immediately allowed all theories with gauge groups $\SO(2n + 1)$ (including $\SU(2) \cong \SO(3))$, $\SO(4n)$ for $n \geq 2$, $\Sp (2n)$ for $n \geq 3$, as well as all the exceptional Lie groups other than $\E_6$ (which we know has complex representations). In the same paper, they went on to prove that all $\SO(n)$'s, excluding $\SO(6)$, are also anomaly-free, thus adding the case $\SO(4n + 2)$ for $n \geq 2$ to the list of allowed groups (note that this includes the $\Spin (10)$ theory). This left only the $\SU(n)$'s, for $n \geq 3$, and $\E_6$. The case of the $\SU(5)$ theory is discussed in \cite[Ch.~5.2]{mohapatra02} for example; it suffices for us to note that the representation that we employed is indeed free of anomalies. As for $\E_6$, the following fact was clear to G\"ursey and collaborators \cite{gursey75} right at the advent of this theory:

\begin{thm}
	All representations of $\E_6$ are anomaly-free.
\end{thm}

The remainder of this section will be dedicated to sketching a proof of this result; we will follow Okubo's paper from 1977 \cite{okubo77}. At the crux of his proof is the following claim: the calculation of an $n$-fermion closed-loop diagram is related to a study of the $n$th order Casimir invariant of the algebra $\mf{g}$ of the symmtery group $G$. We will slowly work towards understanding what this means, and how the proof proceeds therefrom.

Introducing Feynman diagrams in detail is outside our scope here\footnote{The reader is once again encouraged to consult \cite{peskin95, schwartz14, weinberg05} or any standard quantum field theory reference.}, but we must say a few words: these diagrams are representations of certain mathematical expressions that arise in perturbative (read: almost all) calculations in quantum field theory, usually of the scattering amplitudes of particles. The chiral anomaly we have been considering thus far is often called a \emph{triangular anomaly} because to one-loop (first order), the Feynman diagram looks like so.
\begin{center}
	\begin{tikzpicture}[scale=0.4]
	\draw [thick]
	(0,0) -- (3,1.73)
	(3,1.73) -- (3,-1.73)
	(0,0) -- (3,-1.73);
	\draw [thick] decorate [decoration={coil, aspect=0}]
	{(-3.5,0) -- (-0.5,0)
	(3,1.73) -- (5.6,3.23)
	(3,-1.73) -- (5.6,-3.23)};
	\end{tikzpicture}
\end{center}
Here, the external legs represent any of the gauge bosons of the theory, while the fermions circulating in the internal lines can be in any relevant representation of the gauge groups. So according to Okubo, since our anomaly arises in the $3$-fermion closed-loop diagram, we need to concern ourselves with the $3$rd order Casimir invariant of $\E_6$. To introduce the same, we shift our analysis to the level of Lie algebras, recalling (section \ref{thenewfermionssubsec}) that for simply connected Lie groups, this involves no sacrifice of generality.

\begin{defn}[Structure Constants]
	For a $d$-dimensional Lie algebra, consider any set of $d$ basis vectors, or generators, $\{ t_a \}$ . Because of bilinearity, the Lie bracket is determined uniquely if it is known on a basis set; therefore one can define the Lie bracket, and hence the Lie algebra $\mf{g}$ abstractly through the expansions
	\begin{equation*}
		[t^a , t^b] = \sum_{c = 1}^d \tensor{f}{^{ab}_c} t^c\ ,
	\end{equation*}
	where the $\tensor{f}{^{ab}_c}$ are called the structure constants of $\mf{g}$.
\end{defn}

Because of the antisymmetry of the Lie bracket, the structure constants satisfy $\tensor{f}{^{ab}_c} = - \tensor{f}{^{ba}_c}$; from the Jacobi identity, we further have
\begin{equation*}
	\sum_{c = 1}^d \left( \tensor{f}{^{ab}_c}\tensor{f}{^{cd}_e} + \tensor{f}{^{da}_c}\tensor{f}{^{cb}_e} + \tensor{f}{^{bd}_c}\tensor{f}{^{ca}_e} \right) = 0\ .
\end{equation*}
The next idea that we wish to consider requires the tensor algebra of the vector space (definition \ref{tensoralgebradefn}) over which $\mf{g}$ is defined. To endow this very general product with the structure that $\mf{g}$ carries, we make an obvious identification: an element of the form (we suppress the $\otimes$ symbol for brevity)
\begin{equation*}
	x_1 x_2 \cdots x_i x_{i+1} x_{i +2} \cdots x_n - x_1 x_2 \cdots x_{i + 1} x_i x_{i +2} \cdots x_n \in V^{\otimes n}
\end{equation*}
is identified with
\begin{equation*}
	x_1 x_2 \cdots [x_i , x_{i + 1} ] x_{i + 2} \cdots x_n \in V^{\otimes (n -1)}\ .
\end{equation*}
This quotient still has the structure of an associative algebra (with a unit element), and is called the \emph{universal enveloping algebra} of $\mf{g}$.

\begin{defn}[Vector Operator]
	A collection  of $d$ elements $\{x_a \}$ belonging to the universal enveloping algebra of $\mf{g}$ is called a vector operator on $\mf{g}$ if the following relation holds:
	\begin{equation}
	\label{vectoroperatoreqn}
		[t^a , x^b] = \sum_{c = 1}^d \tensor{f}{^{ab}_c} t^c\ .
	\end{equation}
\end{defn}

Obviously, $\{ t_a \}$ is a vector operator, but it is not in general the only one. We can also define vector operators $\{ x_a \}$ for a given $n$-dimensional representation $\rho$ of $\mf{g}$, if $\{ t_a \}$ and $\{ x_a \}$ are $n \times n$ matrices satisfying the structure equation and the relation above.

Let us restrict to the case that $\rho$ is an irrep of $\mf{g}$, and moreover demand that $\mf{g}$ be simple. Then Okubo showed in \cite{okubo77b} that there is a simple relationship between the number of all linearly independent vector operators on the representation $\rho$, and the \emph{highest weight} $\Lambda$ of that representation. This latter quantity is defined by
\begin{equation*}
	\Lambda = m_1 \Lambda_1 + m_2 \Lambda_2 + \cdots m_l \Lambda_l\ ,
\end{equation*}
where $l$ is the rank\footnote{We have technically only defined the notions of rank and roots for Lie groups, but it should be clear that these concepts can be extended quite naturally to Lie algebras. For the details, see \cite[Ch.~6]{hall04}.} of $\mf{g}$, the $\Lambda_i$'s are its roots (definition \ref{rootsdefn}), and the $m_i$'s are non-negative integers specified uniquely\footnote{It is not at all obvious that such a unique decomposition in terms of roots should exist for an arbitrary irrep of $\mf{g}$; we refer the reader to \cite[Ch.~7]{hall04} for a proof of this, the so-called \emph{highest weight theorem}.} by the representation $\rho$. Let us denote the number of $m_i$'s which are zero by $\nu_0 (\rho)$; then the number of linearly independent vector operators $\nu (\rho)$ is given by
\begin{equation*}
	\nu (\rho) = l - \nu_0 (\rho)\ .
\end{equation*}
In other words, $\nu (\rho)$ is equal to the number of $m_i$'s which are positive. We can apply this theorem immediately: for the standard Lie algebras, their roots (and weights) have been studied and tabulated \cite{patera89}; consulting these, we see quickly that the algebras with Dynkin diagram type $A_n$, for $n \geq 2$, have $\nu (\ad) = 2$, and $\nu (\ad) = 1$ for all other algebras. This fact will be extremely important in what follows.

Let us consider the adjoint representation of $\mf{g}$ in some more detail. Set $T_a = \ad \, t_a$, so that the $pq$-th entry of this matrix is given by $\tensor{(T_a)}{^p_q} = \tensor{f}{^p_{aq}}$; the $d \times d$ matrices $\{T_a\}$ clearly satisfy the structure equation
\begin{equation*}
	[T^a , T^b] = \sum_{c = 1}^d \tensor{f}{^{ab}_c} T^c\ .
\end{equation*}
Recall now the Killing form, definition \ref{thekillingformdefn}; we will denote the same by $g_{ab} = \tr T_a T_b$. Introduce now the vector operator $\{ X_a \}$ on the adjoint representation; by definition, it satisfies $[T^a , X^b] = \sum_{c = 1}^d \tensor{f}{^{ab}_c} X^c$. We further introduce the triple linear forms
\begin{align*}
	\tensor{X}{^p_{aq}} &= \tensor{(X_a)}{^p_q}\ ,\\
	X_{abc} &= \sum_{i = 1}^d g_{ai}^{} \tensor{X}{^i_{bc}}\ .
\end{align*}
In particular, for $\{ X_a \} = \{ T_a \}$, the second equation reduces to $T_{abc} = \sum_{i = 1}^d g_{ai}^{}\tensor{f}{^i_{bc}}$; so $T_{abc}$ is completely antisymmetric in its indices. Now, Okubo showed that an equivalent way of writing the vector operator equation (\ref{vectoroperatoreqn}) is the following:
\begin{equation}
\label{equivwayofwritestreqn}
	\sum_{a = 1}^d \tensor{f}{^a_{bc}} \tensor{X}{_{ade}} + \tensor{f}{^a_{bd}} \tensor{X}{_{cae}^{}} + \tensor{f}{^a_{be}} \tensor{X}{_{cda}} = 0\ ; 
\end{equation}
he further showed that the $X_{abc}$ satisfying these equations can be chosen to be either completely symmetric or completely antisymmetric, and moreover, that the completely antisymmetric $X_{abc}$ must be proportional to the structure constants $f_{abc}$. Hence, by the result in the previous paragraph, we conclude that for all simple Lie algebras other than the $A_n$'s, the $X_{abc}$'s must be antisymmetric, because $\{ X_a \}$ must be proportional to $\{ T_a \}$, the one and only vector operator on the adjoint representation. For the $A_n$'s, there is an additional vector operator.

In order to obtain the ``upper-indices'' version of the form $X_{abc}$, we do the obvious thing, raising indices with the Killing form: $X^{abc} = \sum_{p,q,r} g^{ap} g^{bq} g^{cr} X_{pqr}$; the form we recover must clearly also be either symmetric or antisymmetric. If we then set
\begin{equation*}
	I_3 = \sum_{a, b, c = 1}^d X^{abc} T_a T_b T_c\ ,
\end{equation*}
this defines a \emph{Casimir operator} (of the third order) of $\mf{g}$. These operators are sometimes referred to as Casimir invariants\footnote{For further discussion of these invariants, see \cite[Ch.~17.8]{fuchs03} and references therein.}, because they are invariant under the action of $\mf{g}$ on its adjoint representation---they are, in other words, intertwining operators, as defined in section \ref{symmetriessubsec}. What is important for our purposes here is that the formula for $I_3$ above is in fact the most general form of a third-order Casimir operator, and that the coefficients $X_{abc}$ \emph{must} be symmetric. So from the result in the previous paragraph, we conclude that $I_3 \equiv 0$ for every algebra other than the $A_n$'s.\footnote{For completeness, we note that in the case of the antisymmetric $X_{abc}$'s, $I_3$ reduces to a scalar multiple of the Killing form, which is in fact the second-order Casimir invariant.}

We are now ready to understand the claim that the calculation of an $n$-fermion closed-loop diagram is related to a study of the $n$th order Casimir invariant of $\mf{g}$. Let $\{ T_a \}$ be the representation matrices of the $\{ t_a \}$ for a generic irrep of $\mf{g}$; set
\begin{equation*}
	X_{abc} = \tr{T_a \{ T_b , T_c \}}\ .
\end{equation*}
Obviously, $X_{abc}$ is completely symmetric. Moreover, it satisfies equation (\ref{equivwayofwritestreqn}) if we note the trivial identity $\tr{[T_p , T_a \{ T_b, T_c \} ]} = 0$. Thus, $X_{abc}$ is a symmetric triple form, corresponding to some vector operator $\{ X_a \}$; it must hence vanish identically, and so must equation (\ref{anomalyvanisheqn}). This completes the sketch of the proof that all representations of $\E_6$ are anomaly-free.\footnote{We have of course shown more: we have shown that all representations of every algebra other than $\SU (n)$ for $n \geq 3$ are anomaly-free.}

\section{Other Signatures}
\label{furtherreadingsec}

Since grand unified theories have ever been within the purview of physicists, it is only fitting that we devote this final section of the paper to these tireless inquirers, for whom it is never good enough that a theory be beautiful, and rightly so; they demand that it be predictive, and hence, falsifiable. So in the following paragraphs, we will attempt to paint in broad strokes the answers to some of the questions that have naturally arisen during our analysis, but on which we have hitherto been silent; it is scarcely necessary to add that we are striving neither for comprehensiveness nor exhaustiveness here.

Let us begin with the simplest question: what signatures might we expect from the $\E_6$ theory? We quote from a recent survey paper precisely about this topic that will set the tone for the rest of the discussion:
\begin{quote}
	Signatures of $\E_6$ include [an] extension of the Higgs sector; existence of neutral $Z'$ gauge
bosons at masses above the electroweak scale\ldots; the production of new vector-like quarks
and leptons, and manifestations of the neutral fermion\ldots through its mixing with
other neutral leptons, giving rise to signatures of ``sterile'' neutrinos. Up to now,
with the possible exception of weak evidence for sterile neutrinos there has been no
indication of the extra degrees of freedom entailed by the $27$-plet of $\E_6$.  \cite{joglekar16}
\end{quote}
Setting aside for the moment that the outlook for the $\E_6$ theory is rather bleak, let us try and understand the terms above that we have not yet encountered.

We have refrained from speaking about the Higgs mechanism thus far, and unfortunately, our silence on the same will continue---references \cite{barbieri81, barbieri80, bowick81, gursey78, gursey81} are some early papers that examine this mechanism within the context of mass scales in the $\E_6$ model. One thread that runs through them all is worth examining, since it is directly concerned with the most obvious thing one would think to look for to ratify the $\E_6$ theory, namely, new fermions. Here we introduce the \emph{Survival Hypothesis} \cite{barbieri80b, georgi79}: stated succinctly, it says that low-mass fermions are those that cannot receive $G_\mrm{SM}$ invariant masses. To understand what this means, recall from the previous section that mass terms spoil the invariance of the Lagrangian under chiral symmetry; the survival hypothesis thus postulates that when the grand unification symmetry group is broken down to the Standard Model gauge group, the fermions which do not acquire mass are those that cannot receive mass terms invariant under $G_\mrm{SM}$; in particular, this means that all the particles that do admit such a mass term will receive a superheavy mass, since the symmetry breaking occurs at grand unification scales.\footnote{In the previous section, for instance, the superheavy mass scale was found to be on the order of $10^{16}\,\mrm{GeV}$.} Put another way, most fermions in a grand unified theory should have masses on unification scales; those that do not are associated with one of the two unbroken gauge symmetry groups $G_\mrm{SM}$ or $\U(1) \times \SU(2)$, since both of these demand chiral symmetry.\footnote{They could also be Higgs $\SU(2)$ doublets, but we ignore this possibility here.} The upshot is the following: we do not see the new fermions of the $\E_6$ theory because they are phenomenally heavy, many orders of magnitude outside the reach of even the most powerful detectors. For a thorough examination of mass scales in grand unified theories in general, and $\E_6$ in particular, see \cite{robinett82}. We note that this discussion is of course not valid only for $\E_6$, but is formulated as a general principle in grand unified theories; this is the ``fermion desert''. In Georgi's words, ``If [the above] picture is correct, physics between $300\,\mrm{GeV}$ [the Standard Model mass scale] and $10^{14}\,\mrm{GeV}$ is boring. There is a grand plateau in momentum scale on which the world is well-described by an $\SU(3) \times \SU(2) \times \U(1)$ gauge theory\ldots there will be no new interactions below $10^{15}\,\mrm{GeV}$.'' \cite{georgi79}

Let us turn our attention now to the aforementioned $Z'$ bosons. A detailed analysis of their phenomenology is outside the scope of this paper, but they arise quite naturally in representation theory, and \emph{this} we can certainly understand. If gauge bosons live in the complexified adjoint representation of the symmetry group $G$, it is clear that there is something special about the maximal set of elements in $\mf{g}$ that commute with each other, i.e.~the\footnote{There are of course many different choices of a Cartan subalgebra for a given semisimple Lie algebra, but all of these are related by automorphisms of $\mf{g}$, so we abuse terminology and speak here as though our choice were unique.} \emph{Cartan subalgebra} of $\mf{g}$. In general, these commuting elements make for good quantum numbers (charges); the best way to see this is by choosing the \emph{Cartan-Weyl basis} for $\mf{g}$, which we describe now. Let us consider the complexified adjoint representation of $\mf{g}$ (while suppressing the use of the $\ad$ operator notation): if $\mf{g}$ has rank $l$ and dimension $d$, let $\{ x^i \}$, for $i = 1, \ldots, l$, be the Cartan subalgebra; then one can show that the set $\{ x^i \}$ can be completed to a basis $\{ x^i, t^\alpha \}$ of $\mf{g}$ such that
\begin{equation*}
	[x^i , t^\alpha] = \alpha (x^i) t^\alpha \qquad \text{for } i = 1, 2, \ldots l\ ,
\end{equation*}
where the eigenvalue $\alpha (x^i)$ is non-vanishing for at least one value of $i$. This is the Cartan-Weyl basis for $\mf{g}$, sometimes called the \emph{canonical} or \emph{standard} basis. The most pertinent example of this construction is something that we have already encountered in the $\SU(2)$ weak force. The relevant Lie algebra, $\mf{sl} (2, \C)$, is of rank $1$ and spanned by the $W$ bosons; they form a Cartan-Weyl basis since they satisfy
\begin{equation*}
	[W^0, W^\pm] = \pm 2 W^\pm\ , \qquad [W^+, W^-] = W^0\ .
\end{equation*}
Hence, we see that the basis for the 1-dimensional Cartan subalgebra is indeed given by a quantum number operator, the isospin matrix $\hat{I}_3 = 2 W^0$.

To extend the notion of charges to gauge bosons, recall first this aspect of $\hat{I}_3$: in a standard (doublet) basis for the space of weak-theory fermions $\C^2$, the isospin of a particle was simply given by the eigenvalue of the action of the $\hat{I}_3$ operator on said fermion. It should be clear that the correct generalisation of the above concept is the following: since the gauge bosons transform in the adjoint representation of the symmetry group $G$, they can act on \emph{each other} through the adjoint action of $\mf{g}$ on itself\footnote{Clearly, this action is non-trivial only if $G$ is non-abelian.}; moreover, in the Cartan-Weyl basis, the charges of the gauge bosons are given (up to normalisation) by their roots. In the case of the $\SU(2)$ theory for instance, this yields the correct isospin for $W^\pm$, i.e.~$\pm 1$, since $[\hat{I}_3, W^\pm] = \pm W^\pm$. But notice that we now have a highly interesting statement about these number-generating gauge bosons themselves: all of their quantum charges must vanish since they belong to the Cartan subalgebra (and hence commute with each other). In other words, the number of neutral gauge bosons in a theory is given by the rank of its symmetry group \cite{leike99}. The bearing of this discussion on grand unified theories is straightforward: when we break from $\E_6 \to \Spin (10) \to \SU(5) \to G_{\mrm{SM}}$, we break from a rank $6$ group to a $5$ to a $4$, and then once again to a $4$. Hence, while additional neutral gauge bosons are forbidden in the (standard) $\SU(5)$ theory, both $\Spin (10)$ and $\E_6$ each obtain one additional neutral gauge boson when their symmetry is broken.

This is where representation theory ends and quantum field theory begins. The literature on neutral $Z'$ bosons is vast and varied, and we will not undertake a review of the same here; we point the reader instead to \cite{langacker09, leike99} and references therein for summaries of the physics and the phenomenology respectively; both cover $\E_6$ in some detail. Reference \cite{patrignani16} has the current exclusion limits on the masses of these $Z'$ bosons for the $\Spin (10)$ and $\E_6$ theories: the lower bounds are all on the order of $10^3\,\mrm{GeV}$.

Let us now consider sterile neutrinos. We have assumed throughout our analysis that these exist, incorporating them into the $\SU(5)$ theory, and then consequently into the $\Spin (10)$ and $\E_6$ theories. Rosner \cite{rosner14} recently carried out a phenomenological analysis for the sterile neutrinos in the $\E_6$ theory, of which there are three, from the three copies of $\cc{N}$. In his framework, the traditional candidates for sterile neutrinos, the $\nu_R$'s, obtain extremely heavy masses and become unimportant, while the $\Spin (10)$ singlets  (the final entry in table \ref{particlesinthenofe6}) acquire light masses, and are promoted to sterile neutrino status. The other exotic fermions in the theory remain heavy, and mix only weakly with the Standard Model fermions, as per the survival hypothesis. Finally, he notes that only two of these three $\Spin (10)$ singlets are required to account for present data in neutrino oscillation experiments, leaving one neutrino free to be a candidate dark matter particle. This would fit neatly into the picture of ``dark electromagnetism'' proposed by Ackerman et al.~\cite{ackerman09}: in this scenario, dark matter particles interact via a new gauge boson corresponding to some $\U(1)$ theory that is unbroken at the vacuum---the $\U(1)'$ gauge group that arises in the breaking of the $\E_6$ theory is a natural candidate for the same. Schwichtenberg however, argues in \cite{scwichtenberg17} that the $\Spin (10)$ singlet is not the correct choice for a candidate dark matter particle in the $\E_6$ theory. Instead, he makes a case for the exotic neutrino in the $\cc{\Lambda^1_5} \otimes \xi^2$ representation of $\SU(5) \times \U(1)'$ (this is the third entry from the bottom in table \ref{particlesinthenofe6}). Diving into the details of this interesting debate is unfortunately beyond our scope here; we simply wanted to note that the exotic fermions in the $\E_6$ theory provide a playground to explore such ideas.

We end with a brief discussion on perhaps the most famous prediction of grand unified theories: proton decay. Simply put, since each inclusion of $G_\mrm{SM}$ into the unification gauge groups that we have been considering involves significant jumps in dimension---from $12$ to $24$ to $45$ and finally to the $78$-dimensional $\E_6$---we obtain at each step a huge number of new gauge bosons. These mediate new interactions between particles, one of which is the decay of the proton, which is stable in the standard model; the review article \cite{langacker80} by Langacker treats the subject of proton decay in depth. The original $\SU(5)$ model predicted a maximum proton lifetime on the order of $10^{31}$ years \cite{frampton83} which was subsequently disproved by the Super-K(amiokande) experiment.\footnote{\cite{abe17} is the most recent publication from the collaboration, summarising data from the 20 years (!) that this experiment has been running.} Following a long period where the $\Spin (10)$ theory was thought to be as dead as the $\SU(5)$, Bertolini et al.~\cite{bertolini09} re-examined proton decay in $\Spin (10)$and discovered that it was still viable. For the $\E_6$ theory, \cite{london86} and \cite{robinett82} are two early references that go into great detail regarding proton decay via many possible symmetry breaking chains of $\E_6$; one take-away point is that in almost every case, the proton decay rate for the $\SU(5)$ theory is a lower-bound for the same in the $\E_6$ theory, with the possibility of extending the proton lifetime by some orders of magnitude above the $\SU(5)$ bound depending on how certain parameters in the theory are chosen. 

This brings us to our final point: the parameter space of grand unified theories is (usually) large enough for all manner of tinkering and fine-tuning to match data. In some sense then, they still have a shot at corresponding to reality. And yet, we steadily seem to be approaching a point where testing them is getting difficult to the point of being infeasible. As a recent article notes,
\begin{quotation}
	But while [the aforementioned] Super-K could suddenly strike gold in the next few years and confirm one of these models, it could also run for another 20 years, nudging up the lower limit on the proton's lifetime, without definitively ruling out any of the models.
	
	Japan is considering building a \$1 billion detector called Hyper-Kamiokande, which would be between 8 and 17 times bigger than Super-K and would be sensitive to proton lifetimes of $10^{35}$ years after two decades. It might start seeing a trickle of decays. Or it might not. ``We could be unlucky,'' [S.~M.] Barr\footnote{One of the inventors of the flipped $\SU(5)$ theory \cite{barr82}.} said. ``We could build the biggest detector that anyone is ever going to build and protons decay just a little bit too slow and then we're out of luck.''~\cite{fisker16}
\end{quotation}
Indeed. So while the incompleteness and seeming arbitrariness of the Standard Model remain strong motivators to seek a more complete, natural physics, it seems that the best we can do right now, at least as regards grand unified theories, is simply to wait. It would appear that nature is not so keen to give up her secrets just yet.

\begin{appendices}
\setcounter{chapter}{-1}
\chapter{The Construction of $\mrm{G}_2$}
\renewcommand{\thechapter}{A}
\counterwithin{defn}{chapter}

\label{constructionofg2appen}

The reference for this appendix is \cite{adams96}. For large values of $n$, the Dynkin diagrams of type $A_n$, $B_n$, $C_n$, $D_n$ are distinct. For small values of $n$, we have the possibility of exceptional isomorphisms between the classical groups as follows.
\begin{enumerate}
	\item $\Spin (6) \cong \SU (4)$. Both have dimension 15, rank 3 and the Dynkin diagram
	\raisebox{-1em}{\begin{tikzpicture}[scale=0.8, transform shape]
\draw[fill=black]
(0,0.5) circle [radius=.08] 
(0,-0.5) circle [radius=.08]
(1,0) circle [radius=.08];
\draw 
(0,0.5) -- (1,0)
(0,-0.5) -- (1, 0);
\end{tikzpicture}}\ .
	\item $\Spin (5) \cong \Sp (2)$. Both have dimension 10, rank 2 and the Dynkin diagram \begin{tikzpicture}[scale=0.8, transform shape]
\draw[fill=black]
(0,0) circle [radius=.08] 
(1,0) circle [radius=.08];
\draw 
(0,-0.06) -- (1,-0.06)
(0,+0.06) -- (1,0.06);                      
\end{tikzpicture}\ .
	\item $\Spin (3) \cong \SU (2) \cong \Sp (1) = S^3 \subset \mbb{H}$.
	\item $\Spin (4) \cong S^3 \times S^3$.
\end{enumerate}
We prove the first two isomorphisms.

\begin{prop}
\label{spin6congsu4}
	We have the following isomorphisms of Lie groups: $\Spin (6) \cong \SU (4)$ and $\Spin (5) \cong Sp (2)$.
\end{prop}
\begin{proof}
	We do both cases in parallel. $\Spin (5)$ has the representation $\Delta$ of degree 4 over $\C$ and degree 2 over $\mbb{H}$. We can impose a Hermitian form, invariant under the compact group $\Spin (5)$, giving us a homomorphism\footnote{This is because $\Sp (2)$ is \emph{defined} to be the subgroup of $\GL (2, \mbb{H})$ under which the inner product on $\mbb{H}$ is invariant.} from $\Spin (5) \to \Sp (2)$ which we also denote with $\Delta$. Similarly, $\Spin (6)$ has the representation $\Delta^+$ of degree 4 over $\C$ and we have $\Delta^+ : \Spin (6) \to \U(4)$. We first wish to show that $\imag \Delta^+ \subset \SU (4)$.

Let $t \in T \subset \Spin (6)$. Then $t$ acts with eigenvalues defined by weights
		\begin{equation*}
			\left\{ \frac{1}{2} \left(x_1 + x_2 + x_3 \right) , \frac{1}{2} \left(x_1 - x_2 - x_3 \right), \frac{1}{2} \left(- x_1 + x_2 - x_3 \right), \frac{1}{2} \left(-x_1 - x_2 + x_3 \right) \right\}\ ;
		\end{equation*}
		these add to zero so the eigenvalues multiply to 1 and $t$ must act with determinant $1$. Hence, any $gtg\inv$ acts with $\det 1$ and $\Delta^+$ maps to $\SU(4)$.
		
		Now $\Delta$ is faithful: if $g \in \Spin (2n + 1)$ and $g \mapsto 1$ in $\Hom_\C (\Delta, \Delta) \cong \Cl (V)_0$ then $g$ is 1. Also $\Delta^+$ is faithful if $n$ is odd, for if $g$ acts as 1 on $\Delta^+$, then $g \in \Spin (2n)$ acts as 1 on the dual, $\Delta^-$, so $g \mapsto 1$ in $\Hom_\C (\Delta^+ , \Delta^+) + \Hom_\C (\Delta^- , \Delta^-) = \Cl(V)_0$. Thus $g = 1$. Hence the two maps $\Delta$, $\Delta^+$ are injective homomorphisms, and induce injective homomorphisms $\diff \Delta$, $\diff \Delta^+$; they are thus isomorphisms for dimensional reasons. Hence, $\Delta$ and $\Delta^+$ map small neighbourhoods in $\Spin (5)$, $\Spin (6)$ onto small neighbourhoods of $\Sp (2)$, $\SU(4)$ respectively. But $\Sp (2)$ and $\SU(4)$ are connected, so $\Delta : \Spin (5) \to \Sp (2)$ and $\Delta^+ : \Spin (6) \to \SU(4)$ must be surjections.
\end{proof}

\begin{cor}
	The group $\Spin (5)$ acts transitively on the unit sphere $S^7 \subset \Delta$.
\end{cor}
\begin{proof}
	$\Sp (2)$ acts transitively on $S^7 \subset \mbb{H}^2$.
\end{proof}

\begin{cor}
\label{subgroupofspin(6)cor}
	The group $\Spin (6)$ acts transitively on pairs $(x, z)$, $x \in S^5 \subset \R^6$, $z \in S^7 \subset \Delta^+$. Moreover the subgroup fixing $z \in S^7 \subset \Delta^+$ may be taken, by a suitable choice of $z$, to be $\SU(3) \subset \Spin (6)$, where this inclusion arises by lifting the composite $\SU(3) \hookrightarrow \U(3) \hookrightarrow \SO(6)$ to $\Spin (6)$.
\end{cor}
\begin{proof}
	$\Spin (6)$ covers $\SO(6)$, which acts transitively on $S^5$. Taking a suitable $x$, e.g.~$x = (0, \ldots , 0, 1)$, the subgroup fixing $x$ is $\Spin (5)$. The restriction of $\Delta^+$ to $\Spin (5)$ is $\Delta$ (by proposition \ref{inclusionofspingroupsinside}) and $\Spin (5)$ acts transitively on $z \in S^7 \subset \Delta$ by the last corollary. The weights of $\Delta^+$ are $\left\{ \frac{1}{2} (x_1 + x_2 + x_3 ) , \frac{1}{2} (x_1 - x_2 - x_3) , \ldots \right\}$ and their restrictions to $T \subset \SU(3)$ are\footnote{This follows from \ref{rootsoftheothermatrixgroups} and \ref{maximaltoursofun}.} $\{ 0, x_1, x_2, x_3 \}$; hence the restriction of $\Delta^+$ to $\SU(3)$ is $1 + \Lambda^1$. So this $\SU(3)$ fixes points of $S^7 \subset \Delta^+$: in fact, a whole circle of them.\footnote{The representation $1 + \Lambda^1 = 1 + \C^3 \subset \C^4$ of $\SU (3)$ has $\SU(3)$ acting by $\diag (1, A)$ for $A \in \SU(3)$. Hence, the action of $\SU(3)$ on $S^7 \subset \C^4$ fixes a (complex) circle. A familiar analogy is $\SO (2)$ rotating the 3-sphere about the $z$-axis, which fixes an $S^0$, the north and south poles.} For any such fixed point, the subgroup fixing it is no bigger, since in $\SU(4) \cong \Spin (6)$ the subgroup fixing a unit vector in $\C^4$ is an $\SU(3)$.  
\end{proof}

\begin{cor}
	The group $\Spin (7)$ acts transitively on triples $(x, y, z)$ where $x, y$ are orthogonal and $z \in S^7 \subset \Delta$.
\end{cor}
\begin{proof}
	$\Spin (6)$ covers $\SO(7)$, which is transitive on points $y \in S^6$. Choose $y = (0, \ldots, 0, 1)$. Then the subgroup fixing $y$ is $\Spin (6)$. Now $\Delta$ restricts to $\Delta^+ + \Delta^-$ over $\C$ (by proposition \ref{inclusionofspingroupsinside} again), but starting with $\Delta$ as a fixed vector space of dimension 8 over $\R$, the restriction to $\Spin (6)$ is the representation of dimension 8 underlying $\Delta^+$ (or $\Delta^-$). Finally, note that by the previous corollary, $\Spin (6)$ is transitive on pairs $(x, z)$, $x \in S^6$, $z \in S^7 \subset \Delta^+$.
\end{proof}

\begin{thm}
\label{g2thm}
	Consider the subgroup $G$ of $\Spin (7)$ which fixes a point $z \in S^7 \subset \Delta$. Then $G$ is a compact, connected, simply connected Lie group of rank 2 and dimension 14, with the Dynkin diagram \begin{tikzpicture}[scale=0.8, transform shape]
\draw[fill=black]
(0,0) circle [radius=.08] 
(1,0) circle [radius=.08];
\draw 
(0,-.06) --++ (1,0)
(0, 0) --++ (1, 0)
(0,+.06) --++ (1,0);                      
\end{tikzpicture}
and commonly called $\mrm{G}_2$. Moreover, $\mrm{G}_2$ is transitive on pairs $(x, y)$ of orthogonal vectors in $S^6 \subset \R^7$.
\end{thm}
\begin{proof}
	 The last sentence follows from the above corollary. Now clearly, $G$ is a closed subgroup of $\Spin (7)$. Since $\Spin (7)$ is transitive on $S^7$, we have $\dim G = \dim \Spin (7) - \dim S^7 = 21 - 7 = 14$. Let $H \subset G$ be the subgroup that fixes $y = (0, 0, \ldots, 1)$. Then $H$ is the same as the subgroup of $\Spin (6)$ which fixes $z$, and by a suitable choice of $z$ we can take $H = \SU (3) \subset \Spin (6)$ (this was corollary \ref{subgroupofspin(6)cor}). Since $H$ is connected and $G/ H = S^6$ is connected, we find that $G$ is connected.\footnote{Use the homotopy exact sequence of the fibration $H \to G \to G/H$.} Similarly, $G$ is simply connected.
	
	We determine the roots of $G$. $H = \SU (3)$ acts on $\mf{h} \subset \mf{g}$ by the adjoint action; we wish to know how $H$ acts on $\mf{g}/\mf{h}$, the tangent space to $S^6$ at $y$. By construction, $S^6 = \Spin (7) / \Spin (6)$, so we need to look at the action of $\Spin (6)$ on $\mf{spin}(7)/\mf{spin}(6)$. This is the geometrically obvious action where the tangent space to $S^6$ at $(0, 0, \ldots, 1)$ is $\R^6$, the space of the first 6 coordinates, and $\Spin (6)$ acts on it as usual. The weights are hence given by\footnote{This can be seen directly: the maximal torus $\tilde{T}$ of $\Spin (6)$ is given in remark \ref{maximaltorusofspinn}. Since we are in an even-dimensional space, the eigenvalues of this rotation matrix occur in complex conjugate pairs, $e^{\pm ix_j}$. Hence, the eigenvalues of the action of $\tilde{\mf{t}}$ on $T\R^6 \cong \R^6$ are $\pm x_j$. See also remark \ref{weightsofspinrepsrmk}.} $\{\pm x_1, \pm x_2, \pm x_3 \}$. Thus, $T \subset \SU(3)$ acts on $\mf{g}$ with weights\footnote{That is, with the (standard) weights of $\SU(3)$ given in example \ref{weightsofothermatrixgroups} together with the weights just computed.} $\{0, 0, \pm (x_1 - x_2) , \pm (x_2 - x_3), \pm (x_3 - x_1), \pm x_1, \pm x_2, \pm x_3 \}$. We conclude that $T$ is maximal\footnote{The trivial representation occurs exactly twice in the list of the 14 (i.e.~all) weights.} (so $G$ has rank $2$), that these are the roots of $G$, and that the Dynkin diagram is \begin{tikzpicture}[scale=0.8, transform shape]
\draw[fill=black]
(0,0) circle [radius=.08] 
(1,0) circle [radius=.08];
\draw 
(0,-.06) --++ (1,0)
(0, 0) --++ (1, 0)
(0,+.06) --++ (1,0);                      
\end{tikzpicture}.
\end{proof}

$\mrm{G}_2$ starts life with two obvious representations: a 7-dimensional representation $G_2 \subset \Spin (7)$ acting on $\R^7$ with weights $\{ 0, \pm x_1 , \pm x_2 , \pm x_3 \}$, and the 14-dimensional adjoint representation, $\Ad$ with weights as in the above theorem.
\end{appendices}

\bibliographystyle{plain}
\bibliography{biblio}

\end{document}